\documentclass[english]{article}
\usepackage[utf8]{inputenc}
\usepackage[margin=2.5cm]{geometry}
\usepackage{babel}

\usepackage{amsmath}
\usepackage{amssymb}
\usepackage{amsthm}
\usepackage{amsfonts}
\usepackage[round]{natbib}

\usepackage{algorithm}
\usepackage{algorithmicx}

\usepackage{color}
\usepackage{xspace}
\usepackage{supertabular}
\usepackage{enumitem}
\usepackage[colorlinks,linkcolor=blue,citecolor=red,urlcolor=blue]{hyperref}
\usepackage[capitalise,nameinlink]{cleveref}
\crefname{appendix}{\hspace*{-1.6mm}}{\hspace*{-1.6mm}}
\crefname{assumption}{Assumption}{Assumptions}
\crefname{lemma}{Lemma}{Lemmata}
\crefrangelabelformat{appendix}{#3#1#4 to~#5#2#6}
\crefalias{prop}{proposition}

\usepackage{graphicx, subfig}
\usepackage{tikz}
\usetikzlibrary{trees}
\usepackage{lineno}
\usepackage{multicol}

\newcommand{\real}{\ensuremath{\mathbb{R}}\xspace}

\newcommand{\bounded}[1][E]{\ensuremath{\mathcal{B}_b(#1)}}

\newcommand{\testfn}{\ensuremath{\varphi} }

\newcommand{\asconverges}{\ensuremath{\overset{a.s.}{\rightarrow}}}
\newcommand{\dconverges}{\ensuremath{\overset{d}{\rightarrow}}}
\newcommand{\pconverges}{\ensuremath{\overset{p}{\rightarrow}}}
\DeclareMathOperator{\N}{\mathcal{N}}

\def\weight{G_n}
\def\weightN{G_n^N}
\def\rmd{\mathrm{d}}
\def\predictiveN{\eta^N_n}
\def\predictive{\eta_n}
\def\updateN{\hat{\eta}_n^N}
\def\update{\hat{\eta}_n}
\newcommand{\bgN}{\ensuremath{\Psi_{G_n^N}}\xspace}
\newcommand{\bg}{\ensuremath{\Psi_{G_n}}\xspace}
\newcommand{\lp}{\ensuremath{\mathbb{L}_p}}
\DeclareMathOperator{\Exp}{\mathbb{E}}
\newcommand{\var}{{\rm{var}}}
\newcommand{\V}{V}
\def\sfmutation{\mathcal{G}_{n-1}^N}
\def\sfresampling{\mathcal{F}_{n}^N}
\makeatletter
\newcommand{\hathat}[1]{%
\begingroup%
  \let\macc@kerna\z@%
  \let\macc@kernb\z@%
  \let\macc@nucleus\@empty%
  \widehat{\raisebox{.5ex}{\vphantom{\ensuremath{#1}}}\smash{\widehat{#1}}}%
\endgroup%
}
\makeatother

\newcommand{\norm}[2]{\ensuremath{\left\vert\left\vert #1 \right\vert\right\vert_{#2}}}

\newcommand{\supnorm}[1]{\norm{#1}{\infty}}

\newtheorem{lemma}{Lemma}

\newtheorem{prop}{Proposition}
\newtheorem{assumption}{Assumption}
\newtheorem{corollary}{Corollary}

\usepackage{authblk}

\usepackage[textwidth=1.8cm, textsize=scriptsize]{todonotes}
\usepackage{setspace} 

\title{Properties of Marginal Sequential Monte Carlo Methods}
\author[2]{Francesca R. Crucinio\thanks{Corresponding author: francesca.crucinio@gmail.com}}
\author[2, 3]{Adam M. Johansen}
\affil[2]{Department of Statistics, University of Warwick}
\affil[3]{The Alan Turing Institute}

\date{ }
\begin{document}
\maketitle

\begin{abstract}
  We provide a framework which admits a number of ``marginal'' sequential Monte Carlo (SMC) algorithms as particular cases --- including the marginal particle filter [Klaas et al., 2005, in: Proceedings of  Uncertainty in Artificial Intelligence, pp. 308–315], the independent particle filter [Lin et al., 2005, Journal of the American Statistical Association 100, pp. 1412–1421] and linear-cost Approximate Bayesian Computation SMC [Sisson et al., 2007, Proceedings of the National Academy of Sciences (USA) 104, pp. 1760–1765.]. We provide conditions under which such algorithms obey laws of large numbers and central limit theorems and provide some further asymptotic characterizations. 
  Finally, it is shown that the asymptotic variance of a class of estimators associated with certain marginal SMC algorithms is never greater than that of the estimators provided by a standard SMC algorithm  using the same proposal distributions.
\end{abstract}

\section{Introduction}
Sequential Monte Carlo (SMC) methods are a class of Monte Carlo methods which approximate a sequence of distributions and integrals with respect to those distributions using a population of weighted samples (or particles) which evolve according to a combination of mutation and selection dynamics.
Such methods became popular in the context of filtering for state-space models in the engineering and statistics literature following the seminal works of \cite{smc:methodology:GSS93}, and have been extensively studied from a theoretical perspective as mean field approximations of a Feynman-Kac flow since \cite{del1996non}. SMC methods have grown in popularity and have been successfully employed in numerous problems including posterior approximation \citep{chopin2002sequential, smc:methodology:DDJ06b}, approximate Bayesian computation  \citep{smc:methodology:SFT07, amj26:DEJL11}, Bayesian model comparison  \citep{zhou2016toward}, parameter estimation \citep{poyiadjis2011particle}. A book-length introductory treatment of SMC is provided by \cite{chopin2020}.

SMC methods approximate a sequence of distributions $(\update)_{n\geq 0}$ defined on Polish spaces $(E^n, \mathcal{E}^n)$, where $\mathcal{E}$ denotes the $\sigma$-field associated with $E$, of increasing dimension with the relationship
  $$
  \update(\rmd x_{1:n}) \propto  U_n(x_{n-1}, x_{n})K_n(x_{n-1}, \rmd x_n) \hat{\eta}_{n-1}(\rmd x_{1:n-1})
  $$
for some Markov kernels $K_n: E\times \mathcal{E}\to [0,1]$ and non-negative functions $U_n:  E\times E\to \real$.

In this work we focus on a particular class of SMC methods, in which the past evolution of the process is marginalized out. Marginal sequential Monte Carlo (MSMC) methods have been introduced in \cite{smc:methodology:KFD05, smc:methodology:LZCC05} where they were shown empirically to improve upon standard SMC in terms of the conditional variance of their unnormalized importance weights.

Despite being a popular class of algorithms which provides good results in practice (e.g. \cite{smc:methodology:SFT07, amj26:DEJL11, everitt2017marginal, poyiadjis2011particle}), the theoretical properties of MSMC are less studied. This paper aims to fill this gap.
We provide a comprehensive analysis of the limiting behaviour of MSMC and show that it shares many of the key properties of standard SMC methods (see, e.g. \cite{smc:theory:Del04}).
In particular, we establish that the estimates provided by MSMC methods obey laws of large numbers (\Cref{prop:wlln,prop:slln}), a central limit theorem (\cref{prop:clt}) and show that their $\lp$ errors decay at the usual $N^{-1/2}$ rate (\cref{prop:lp}) and that the bias decays at rate $N^{-1}$ (\cref{prop:bias}).
Finally, we derive explicit expressions for the asymptotic variance of some popular classes of MSMC methods and provide a comparison with their non-marginal counterparts, showing that the asymptotic variance of the estimates obtained with marginal SMC is never larger than that one would obtain with standard SMC using the same proposals.

Our proofs combine well-known techniques to obtain the analogous results for standard SMC with novel results which control the effect of marginalization.
Embedding this result within the inductive approach normally used to establish the results for standard SMC we obtain a full characterization of MSMC.

We introduce a framework for describing marginal sequential Monte Carlo methods in \cref{sec:msmc}, where we also present some special cases of particular interest. Our main contributions are in \cref{sec:results} in which we present the theoretical properties of these algorithms and show that their asymptotic variance is never greater than that of standard SMC.
Proofs of the results are postponed to \cref{app:lp,app:wlln,app:clt,app:bias,app:variance}.
We close the paper with a discussion in \cref{sec:discussion}. 
\section{Marginal sequential Monte Carlo methods}
\label{sec:msmc}

Marginal sequential Monte Carlo methods (MSMC) are a class of sequential Monte Carlo algorithms introduced in \cite{smc:methodology:KFD05} to approximate a sequence of distributions $(\update)_{n\geq 0}$ defined on measurable spaces $(E, \mathcal{E})$, where $\mathcal{E}$ denotes the $\sigma$-field associated with $E$.
The sequence of distributions satisfies, for some Markov kernels $K_n: E\times \mathcal{E}\to [0,1]$ and non-negative functions $U_n:  E\times E\to \real$, the recursion 
\begin{align}
\label{eq:update}
    \update(\rmd x_n) \propto \int U_n(x_{n-1}, x_{n})K_n(x_{n-1}, \rmd x_n) \hat{\eta}_{n-1}(\rmd x_{n-1})
\end{align}
with $\hat{\eta}_0(\rmd x_0) \propto U_0(x_0)K_0(\rmd x_0)$.

Contrary to standard SMC algorithms, in which a cloud of particles is used to approximate a measure-valued recursion defined over $(E^n, \mathcal{E}^n)$ with approximations of time marginals obtained by discarding part of the sampled paths, MSMC algorithms deal directly with the marginal recursion and the presence of the integral w.r.t. $\hat{\eta}_{n-1}$ in~\eqref{eq:update} requires an additional approximation.
Hence, as described below, in MSMC a sample approximation of the integral w.r.t. $\hat{\eta}_{n-1}$ is used to define an alternative sequence of targets.

Before describing MSMC algorithms, we describe one step of the idealized algorithms targeting~\eqref{eq:update} which MSMC approximates.
At time $n$, we start with a cloud of equally weighted particles $\{\widetilde{X}_{n-1}^{i}, 1/N\}_{i=1}^N$ approximating $\hat{\eta}_{n-1}$.
In the idealized algorithm, the particles are propagated forward in time using proposals
\begin{align}
\label{eq:predictive}
   \predictive(\rmd x_n) =\hat{\eta}_{n-1}M_n(\rmd x_n):=\int M_n(x_{n-1}, \rmd x_n) \hat{\eta}_{n-1}(\rmd x_{n-1})
\end{align}
for $n\geq 1$ and some Markov kernels $M_n:E\times \mathcal{E}\to [0,1]$ to obtain a new cloud of particles $\{X_n^{i}, 1/N\}_{i=1}^N$.
The new cloud of particles is used as proposal in an importance sampling step targeting $\update$, with importance weights proportional to the Radon-Nykodim derivative $\rmd \update/\rmd \predictive$
\begin{align}
\label{eq:W}
\weight(x_n) & =\frac{\textrm{d} \hat\eta_{n-1}(U_n \cdot K_n)}{\textrm{d} \hat\eta_{n-1} M_n}(x_n)=\frac{\rmd \left( \int U_n(x_{n-1}, \cdot)K_n(x_{n-1}, \cdot)\hat{\eta}_{n-1}(\rmd x_{n-1})\right)}{\rmd \left(\int M_n(x_{n-1}, \cdot)\hat{\eta}_{n-1}(\rmd x_{n-1})\right)}(x_n),
\end{align}
where $(U_n \cdot K_n)(x_{n-1},\rmd x_n):= U_n(x_{n-1},x_n) K_n(x_{n-1},\rmd x_n)$, to obtain a weighted cloud of particles $\{X_n^{i}, W_n^i\}_{i=1}^N$, with $W_n^i\propto \weight(X_n^i)$, which approximates $\update$.
The equally weighted population approximating $\update$, $\lbrace \widetilde{X}_{n}^i, \frac{1}{N}\rbrace_{i=1}^N$, is obtained by resampling from $\lbrace X_n^i, W_n^i\rbrace_{i=1}^N$, see e.g. \cite{Gerber2019} for a recent review of popular resampling schemes.

The target distribution~\eqref{eq:update} involves an intractable integral w.r.t. $\hat{\eta}_{n-1}$. For any probability measure, $\eta$, and non-negative integrable function $G$, we can define $\Psi_G(\eta)(\rmd x) := G(x)\eta(\rmd x)/\eta(G)$. Given a cloud of weighted particles $\{X_{n-1}^{i}, W_{n-1}^i\}_{i=1}^N$ approximating $\hat{\eta}_{n-1} = \Psi_{G_{n-1}}(\eta_{n-1})$ we approximate $\hat\eta_{n-1}$ with
\begin{align}
\label{eq:bgN}
\Psi_{G^N_{n-1}}(\eta_{n-1}^N)(\rmd x_{n-1}):= \sum_{i=1}^N W_{n-1}^i \delta_{X_{n-1}^i}(\rmd x_{n-1})
\end{align}
and define the approximate targets
$$
\Psi_{G^N_{n-i}}(\eta_{n-1}^N)(U_n\cdot K_n)(\rmd x_n) \propto \sum_{i=1}^N W_{n-1}^i U_n(X_{n-1}^i, x_n) K_n(X_{n-1}^i, \rmd x_n).
$$
We then proceed as in the idealized algorithm, and, since the integrals w.r.t. $\hat{\eta}_{n-1}$ are intractable, replace $\hat{\eta}_{n-1}$ with its particle approximation obtained at time $n-1$.
Thus, the particles are propagated forward in time using an approximation of the proposals in~\eqref{eq:predictive},
$
\predictive(\rmd x_n) \approx\sum_{i=1}^N W_{n-1}^i M_n(X_{n-1}^i, \rmd x_n)
$,
and resampled using approximate weights
\begin{align}
\label{eq:WN}
    \weightN(x_n) &=\frac{\textrm{d} \Psi_{G^N_{n-i}}(\eta_{n-1}^N)(U_n \cdot K_n)}{\textrm{d} \Psi_{G^N_{n-i}}(\eta_{n-1}^N) M_n}(x_n)= \frac{\rmd \left(\sum_{i=1}^N W_{n-1}^i U_n(X_{n-1}^i, \cdot) K_n(X_{n-1}^i, \cdot)\right)}{\rmd\left(\sum_{i=1}^N W_{n-1}^i M_n(X_{n-1}^i,\cdot )\right)}(x_n).
\end{align}
The resulting MSMC method is summarized in \cref{alg:reformulation}. For convenience, we identify the three fundamental steps of \cref{alg:reformulation} as a \emph{mutation} step (\cref{alg:reformulated:mutate}), \emph{reweighting} steps (\cref{alg:reformulated:weightinga,alg:reformulated:weightingb}) and a \emph{resampling} step (\cref{alg:reformulated:resampling}). To each step, we associate a measure and its corresponding particle approximation: the mutated measure $\predictive$ in~\eqref{eq:predictive} is approximated by $\predictiveN := N^{-1} \sum_{i=1}^N \delta_{X_n^i}$ obtained after \cref{alg:reformulated:mutate}, \cref{alg:reformulated:weightinga,alg:reformulated:weightingb} provide a particle approximation \eqref{eq:bgN} of the reweighted measure $\update = \bg(\predictive)$,
after resampling we obtain another approximation of $\update$ in~\eqref{eq:update}, $\updateN := N^{-1} \sum_{i=1}^N \delta_{\tilde{X}_n^i}$.

\begin{algorithm}
  \caption{Marginal Sequential Monte Carlo (MSMC)}\label{alg:reformulation}
  \begin{multicols}{2}
\begin{algorithmic}[1]
\State Set $n = 0$.
\State Sample $X_0^i \sim M_0$ for $i=1,\dots, N$.
\State Compute $G_0(X_0^i) = \frac{\rmd \hat{\eta}_0}{\rmd M_0}(X_0^{i})$ $i=1,\dots, N$.
\State
Compute $W_0^i = G_0(X_0^i)/ \sum_{j=1}^N G_0(X_0^j)$. 
\State Resample
$\left\{X_0^{i},W_0^i\right\}_{i=1}^N$ to obtain $\{\tilde{X}_0^i, \frac1N\}_{i=1}^N$.
\columnbreak
\State Update $n \leftarrow  n+1$.  \label{alg:reformulated:iterate}
\State Sample $X_n^i \sim  M_n\left(\widetilde{X}_{n-1}^i, \cdot\right)$ for $i=1,\dots, N$. \label{alg:reformulated:mutate}
\State Compute $\weightN(X_n^i) $ in~\eqref{eq:WN} for $i=1,\dots, N$. \label{alg:reformulated:weightinga}
\State
Compute $W_n^i = \weightN(X_n^i)/ \sum_{j=1}^N \weightN(X_n^j)$. \label{alg:reformulated:weightingb}
\State Resample \label{alg:reformulated:resampling}
$\left\{X_n^{i},W_n^i\right\}_{i=1}^N$ to obtain $\{\widetilde{X}_n^i, \frac1N\}_{i=1}^N$.
\State Go to \cref{alg:reformulated:iterate}.
\end{algorithmic}
\end{multicols}
\end{algorithm}

\subsection{Examples of Marginal SMC}
\label{sec:ex}
In the previous section we introduced the general class of marginal sequential Monte Carlo methods. This class encompasses a number of well-known algorithms, we briefly discuss some examples below.

\subsubsection{Marginal Particle Filters}
\label{sec:mpf}
Marginal particle filters (MPF; \cite{smc:methodology:KFD05}) are a class of algorithms to perform inference on state space models (SSM), a family of time series models consisting of two discrete-time processes: a latent process $( X_n)_{n\geq 0}$ and conditionally independent observations $(Y_n)_{n\geq 1}$. Such a SSM $( X_n, Y_n)_{n\geq 0}$ is defined by the transition density $f_n(x_n|x_{n-1})$ of the latent process, with the convention that $f_0(x_0|x_{-1})\equiv f_0(x_0)$, and the observation likelihood $g_n(y_n|x_n)$.

In this case, the target distribution is the filtering distribution
\begin{align}
\label{eq:target_mpf}
    \update(\rmd x_n) \equiv p(x_n|y_{1:n}) \rmd x_n\propto g_n(y_n|x_n)\int f_n(x_n|x_{n-1}) \hat{\eta}_{n-1}(\rmd x_{n-1}) \rmd x_n,
\end{align}
with $\hat{\eta}_0(\rmd x_0) = f_0(x_0) \rmd x_0$, $K_n(x_{n-1}, \rmd x_n)=f_n(x_n|x_{n-1})\rmd x_n$ and $U_n(x_{n-1}, x_n) \equiv U_n(x_n)= g_n(y_n|x_n)$, and proposal density which can incorporate the observation $y_n$, $M_n(x_{n-1}, \rmd x_n) = q_n(x_n| y_n, x_{n-1})\rmd x_n$. The corresponding weights are given by 
\begin{align*}
    \weight(x_n) &= g_n(y_n|x_n)\frac{\int f_n(x_n|x_{n-1})\hat{\eta}_{n-1}(\rmd x_{n-1})}{\int q_n(x_n|y_n, x_{n-1})\hat{\eta}_{n-1}(\rmd x_{n-1})}; &
    \weightN(x_n) &= g_n(y_n|x_n)\frac{\sum_{i=1}^N  W_{n-1}^i f_n(x_n|X_{n-1}^i)}{\sum_{i=1}^N  W_{n-1}^i q_n(x_n| y_n, X_{n-1}^i)}.
\end{align*}
Similar algorithms have also been employed to approximate the score function when performing parameter estimation for SSM \citep{poyiadjis2011particle}.

\subsubsection{Marginal Auxiliary Particle Filters}
\label{sec:mapf}
Marginal auxiliary particle filters (MAPFs) are a variant of MPF introduced in \citet[Section 3.1]{smc:methodology:KFD05}, as a marginalized version of standard auxiliary particle filters \citep{smc:methodology:PS99, smc:methodology:CCF99}.
MAPF can be described as standard MPF applied to $\hat{\eta}_0( \rmd x_0) \propto f_0(x_0)\tilde{p}(y_1|x_0) \rmd x_0$ and
\begin{align}
\label{eq:mapf_target}
    \hat{\eta}_n(\rmd x_n)\propto \tilde{p}(y_{n+1}|x_n)p(x_n|y_{1:n}) \rmd x_n \propto& \tilde{p}(y_{n+1}|x_n)g(y_n|x_n)\left(\int f_n(x_n|x_{n-1})p(x_{n-1}|y_{1:n-1})\rmd x_{n-1}\right) \rmd x_n\\
    \qquad\qquad\propto& \tilde{p}(y_{n+1}|x_n)g(y_n|x_n)\left(\int \frac{f_n(x_n|x_{n-1})}{\tilde{p}(y_{n}|x_{n-1})}\hat{\eta}_{n-1}(\rmd x_{n-1}) \right)\rmd x_n,\notag
\end{align}
where $\tilde{p}(y_{n+1}|x_n)$ is an approximation of $p(y_{n+1}|x_n) := \int g_{n+1}(y_{n+1}|x_{n+1})f_{n+1}(x_{n+1}|x_n)\rmd x_{n+1}$, to which an importance sampling step is added to guarantee that we are targeting the correct distribution $p(x_n|y_{1:n})$ \citep{johansen2008note}.

Setting $K_n(x_{n-1}, \rmd x_n) = f_n(x_n|x_{n-1})\rmd x_{n}$ and $U_n(x_{n-1}, x_n) = g_n(y_n|x_n)\tilde{p}(y_{n+1}|x_{n})/\tilde{p}(y_{n}|x_{n-1})$,
one can apply \cref{alg:reformulation} with proposal kernel $M_n(x_{n-1}, \rmd x_n) = q_n(x_n|x_{n-1}, y_n)\rmd x_n$ which incorporates the current observation so that the weights are
\begin{align*}
    \weight(x_n) &= g_n(y_n|x_n)\frac{\tilde{p}(y_{n+1}|x_{n})\int f_n(x_n|x_{n-1})/\tilde{p}(y_{n}|x_{n-1})\hat{\eta}_{n-1}(\rmd x_{n-1})}{\int q_n(x_n|x_{n-1}, y_n)\hat{\eta}_{n-1}(\rmd x_{n-1})},\\
    \weightN(x_n) &= g_n(y_n|x_n)\frac{\tilde{p}(y_{n+1}|x_{n})\sum_{i=1}^NW_{n-1}^i f_n(x_n|X_{n-1}^i)/\tilde{p}(y_n|X_{n-1}^i)}{\sum_{i=1}^NW_{n-1}^i  q_n(x_n|X_{n-1}^i, y_n)}.
\end{align*}
To obtain an algorithm targeting $\pi_n(\rmd x_n) := p(x_n|y_{1:n})\rmd x_n$, one additional importance sampling step is applied, using as proposal the approximation of $\update$ before resampling, $\bgN(\predictiveN)$, and importance weights
\begin{align}
\label{eq:inferential_weights}
    \widetilde{w}_n(x_n) &= \frac{\rmd \pi_n}{\rmd \update}(x_n) \propto \frac{1}{\tilde{p}(y_{n+1}|x_n)}.
\end{align}
We point out that the description of MAPF given here is convenient to obtain the theoretical characterization below, but, in practice one would use $\predictiveN$ as a proposal in a importance sampling step with weights
\begin{align*}
   \weightN(x_n)\cdot\widetilde{w}_n(x_n) = g_n(y_n|x_n)\frac{\sum_{i=1}^NW_{n-1}^i f_n(x_n|X_{n-1}^i)/\tilde{p}(y_n|X_{n-1}^i)}{\sum_{i=1}^NW_{n-1}^i  q_n(x_n|X_{n-1}^i, y_n)},
\end{align*}
where $W_{n-1}^i$ denote the \emph{simulation} (or \emph{auxiliary}) weights defined in \citet[Section 3.1]{smc:methodology:KFD05} and $W_{n-1}^i/\tilde{p}(y_n|X_{n-1}^i)$ the \emph{inferential} weights (see also \citet[Section 10.3.3]{chopin2020}).

\subsubsection{Independent Particle Filters}

Independent particle filters (IPF) are a class of particle filters introduced in \cite{smc:methodology:LZCC05} built to deal with SSM for which the current observation provides significant information about the current state but the transition dynamics are weak. 
In these scenarios, it is natural to consider proposal kernels $M_n$ which result in draws of $x_n$ which are conditionally independent of individual past particles and of each
other, i.e. $M_n(x_{n-1}, \rmd x_n) = q_n(x_n| y_n)\rmd x_n$.
Taking $g_n$ and $f_n$ as in \cref{sec:mpf}, one obtains the basic IPF as a special case of marginal SMC (see also \citet[Appendix C]{lai2022variational}) in which $U_n(x_{n-1}, x_n) = g_n(y_n|x_n)$, $K_n(x_{n-1}, \rmd x_n)=f_n(x_n|x_{n-1})\rmd x_n$.

\subsubsection{SMC for Approximate Bayesian Computation}

ABC-SMC is an instance of SMC samplers \citep{smc:methodology:DDJ06b} studied in \cite{smc:methodology:SFT07, amj26:DEJL11} which
approximates the posterior distribution of a parameter $\theta$ when the likelihood function $p(y_{obs}|\theta)$ is intractable but can be sampled from.
In this case, the sequence of distributions is defined over the space of parameter and data, i.e. $x_n=(\theta_n, y_n)$, and is given by $\update(\rmd (\theta_n, y_n)) \propto p(\rmd \theta_n)p(\rmd y_n|\theta_n)\pi_{\epsilon_n}(y_{obs}|y_n)$, for some decreasing sequence $(\varepsilon_n)_{n\geq 0}$, where $p(\rmd\theta)$ is a prior on the parameter $\theta$, $p(\cdot|\theta)$ denotes the intractable likelihood and $\pi_{\epsilon_n}$ is the density of a normalized kernel with a degree of concentration determined by $\epsilon_n$ which measures how close $y_n$ is to the observed data $y_{obs}$.

The ABC-SMC algorithm of \cite{smc:methodology:SFT07} is an instance of marginal SMC with%
\begin{align*}
K_n((\theta_{n-1}, y_{n-1}), \rmd (\theta_n, y_n)) &= p(\rmd \theta_n)p(\rmd y_n|\theta_n)\\
U_n((\theta_{n-1}, y_{n-1}), (\theta_n, y_n))& = U_n((\theta_n, y_n))= \pi_{\epsilon_n}(y_{obs}|y_{n})\\
M_n((\theta_{n-1}, y_{n-1}), \rmd (\theta_n, y_n))&=q_n(\rmd\theta_n|\theta_{n-1})p(\rmd y_n|\theta_n),
\end{align*}%
for some proposal $q_n$ (we assume for brevity that the observations themselves are used, rather than some summary statistic but the use of such statistics does not present any substantial difficulties). 
A similar marginal SMC algorithm has also been considered in \cite{everitt2017marginal} for doubly intractable models.

\section{Convergence Results}
\label{sec:results}
We now state our main results, which show that marginal SMC methods have qualitatively the same convergence properties as standard SMC methods. 
We give an overview of the proofs for these results in \Cref{sec:mop}, full details are given in~\cref{app:lp,app:bias,app:wlln,app:clt}.  
For simplicity, we focus on the case of measurable bounded test functions $\testfn:E\to\real$ with $\supnorm{\testfn}:=\sup_{x\in E}\vert\testfn(x)\vert<\infty$, a set we denote by $\bounded$. 
For any distribution $\eta$ and any $\testfn\in\bounded$ we denote $\eta(\testfn):= \int \testfn(x)\eta(\rmd x)$, similarly for all empirical distributions $\eta^N:=N^{-1}\sum_{i=1}^N \delta_{X^i}$ we denote the corresponding average by $\eta^N(\testfn):= N^{-1}\sum_{i=1}^N \testfn(X^i)$.

These results are presented under fairly strong assumptions, which are somewhat standard in this literature, in the interests of brevity. The techniques which allow these to be relaxed in the standard case would also apply here, but their use would substantially complicate the presentation. Similarly, in our arguments we only consider multinomial resampling \citep{smc:methodology:GSS93}. Lower variance resampling schemes can be employed but considerably complicate the theoretical analysis \citep{Gerber2019}.

\begin{assumption}
\label{ass:smc0}
The potentials $\weight$ are positive everywhere, $\weight(x_n)>0$ for every $x_n\in E$.
\end{assumption}

\begin{assumption}
\label{ass:weak} For all $n\geq 0$, the functions $U_n$ are bounded above, i.e. 
$U_n(x_{n-1}, x_n) \leq \supnorm{U_n}<\infty$, and the Radon-Nykodim derivative $\rmd K_n(x_{n-1}, \cdot)/\rmd M_n(x_{n-1}, \cdot)$ is bounded above for all $x_{n-1}$, i.e. there exists some $\alpha>0$ such that, for every $x_{n-1}\in E$, $\supnorm{\rmd K_n(x_{n-1}, \cdot)/\rmd M_n(x_{n-1}, \cdot)} \leq \alpha<\infty$.
\end{assumption}

\begin{assumption}
\label{ass:smc2}
For all $n\geq 0$, the functions $U_n$ are bounded below, and the Radon-Nykodim derivative $\rmd K_n(x_{n-1}, \cdot)/\rmd M_n(x_{n-1}, \cdot)$ is bounded below for all $x_{n-1}$, i.e. there exist $\beta>0$ such that $0<\beta \leq U_n(x_{n-1}, x_n)$ and $0<\beta \leq \rmd K_n(x_{n-1}, \cdot)/\rmd M_n(x_{n-1}, \cdot)$ uniformly in $x_{n-1}$.
\end{assumption}
These assumptions are common in the SMC literature \citep{smc:theory:Del04, del2013mean}, in particular \cref{ass:smc0} (and its stronger version \cref{ass:smc2}) ensure that the system does not become extinct (i.e. the weights have never all simultaneously taken the value zero), and can be relaxed in various ways, including introducing stopping times \citep{smc:theory:Del04} or considering local boundedness \citep{whiteley2013stability}.

\cref{ass:weak} ensures that, uniformly in $x \in E$, $G_n(x) \leq m_g,\ G_n^N(x) \leq m_g$, whereas \cref{ass:smc2} further guarantees
$G_n(x) \geq m_g^{-1}, \ G_n^N(x) \geq m_g^{-1}$, where $m_g:=\max\lbrace\supnorm{U_n}\alpha, \beta^{-2}\rbrace < \infty$. 
To see this, observe that
\begin{align*}
G_n(x_n) = \nu_{x_n}\left(U_n(\cdot, x_n)\frac{\rmd K_n}{\rmd M_n}(\cdot, x_n)\right)\qquad\textrm{with } \nu_{x_n}(\rmd x_{n-1}):= \hat{\eta}_{n-1}(\rmd x_{n-1})\frac{\rmd M_n(x_{n-1}, \cdot)}{\rmd \hat{\eta}_{n-1}M_n}(x_n),
\end{align*}
and similarly for $\weightN$ with $\hat{\eta}_{n-1}$ replaced by $\Psi_{G_{n-1}^N}({\eta}_{n-1}^N)$ in the definition of $\nu$.
\cref{ass:smc0,ass:weak} allow us to obtain the following  weak law of large numbers (WLLN) whose proof is provided in~\cref{app:wlln}:

\begin{prop}[Weak law of large numbers]
\label{prop:wlln}
Under \cref{ass:smc0,ass:weak}, for all $n\geq0$ and for every $ \testfn\in \bounded$, we have $\bgN(\predictiveN)(\testfn)\pconverges\bg(\predictive)(\testfn)$ and $\hat{\eta}_n^N(\testfn) \pconverges\hat{\eta}_n(\testfn)$.
\end{prop}

\cref{ass:smc2} is used to obtain stronger results like finite-$N$ error bounds, whose proof is given in~\cref{app:lp}:

\begin{prop}[$\lp$-inequality]
\label{prop:lp}
Under \cref{ass:smc0,ass:weak,ass:smc2}, for every time $n\geq 0$, every $p\geq 1$ and $N\geq 1$ there exist finite constants $C_{p,n}, \bar{C}_{p, n}$ such that for every measurable bounded function $\testfn\in \bounded$
\begin{enumerate}[label=(\alph*)]
\item $\Exp\left[\vert\bgN(\predictiveN)(\testfn) -\bg(\predictive)( \testfn)\vert^p\right]^{1/p} \leq \bar{C}_{p,n}\frac{\supnorm{\testfn}}{\sqrt{N}},$ \label{prop:lp-pred} 
\item $\Exp\left[\vert\updateN(\testfn) - \update(\testfn)\vert^p\right]^{1/p} \leq C_{p,n}\frac{\supnorm{\testfn}}{\sqrt{N}},$ \label{prop:lp-upd}
 \end{enumerate}
 where the expectations are taken with respect to the law of all random variables generated within the SMC algorithm.
 \end{prop}

The strong law of large numbers requires stronger assumptions that the WLLN in \cref{prop:wlln} and can be obtained from the $\lp$ inequality obtained in \cref{prop:lp} using Markov's inequality within a Borel-Cantelli argument as shown in e.g. \citet[Appendix D]{amj26:BADJ20}.
\begin{prop}[Strong law of large numbers]
\label{prop:slln}
Under \cref{ass:smc0,ass:weak,ass:smc2}, for all $n\geq0$ and for every $ \testfn\in \bounded$, we have $\bgN(\predictiveN)(\testfn)\asconverges\bg(\predictive)(\testfn)$ and $\hat{\eta}_n^N(\testfn) \asconverges\hat{\eta}_n(\testfn)$.
\end{prop}

Using standard techniques (e.g. \cite{berti2006almost}) given in detail for the context of interest in \citet[Supplementary Material, Theorem 1]{schmon2018large}, the result of \cref{prop:slln} can be strengthened to the convergence of the measures in the weak topology:
\begin{prop}
\label{prop:asw}
Under \cref{ass:smc0,ass:weak,ass:smc2}, for all $n\geq 0$, $\bgN(\predictiveN)$ converges almost surely in the weak topology to $\bg(\predictive)$, $\bgN(\predictiveN)\rightharpoonup\bg(\predictive)$, and $\updateN$ converges similarly to $\update$, $\updateN\rightharpoonup\update$. 
\end{prop}

As it is the case for standard SMC algorithms, the reweighting step introduces a bias into estimates of normalized quantities, however, this decays at rate $N^{-1}$ as established in~\cref{app:bias}:
\begin{prop}[Bias estimate]
\label{prop:bias}
Under \cref{ass:smc0,ass:weak,ass:smc2}, for all $n\geq 0$ and any $\testfn\in \bounded$ we have
\begin{enumerate}[label=(\alph*)]
\item $\left\lvert\Exp\left[\bgN(\predictiveN)(\testfn) \right]-\bg(\predictive)( \testfn)\right\rvert \leq \bar{C}_{n}\frac{\supnorm{\testfn}}{N},$ \label{prop:bias-pred}
\item $\left\lvert \Exp\left[ \updateN(\testfn)\right] - \update(\testfn)\right\rvert \leq C_n \frac{\supnorm{\testfn}}{N}$, \label{prop:bias-upd}%
\end{enumerate}
for some finite $\bar{C}_n, C_n$.
The expectations are taken with respect to the law of all random variables generated within the SMC algorithm.
\end{prop}

The following result, proved in~\cref{app:clt}, quantifies the asymptotic variance of the estimates provided by \cref{alg:reformulation} using multinomial resampling. We focus on this resampling scheme because of its simplicity and because, as shown in \citet[Theorem 7]{Gerber2019}, it provides an upper bound on the asymptotic variance obtained with more sophisticated resampling schemes.

\begin{prop}[Central limit theorem]
\label{prop:clt}
Under \cref{ass:smc0,ass:weak,ass:smc2}, for every $n\geq 1$ and every $\testfn\in \bounded$:
\begin{enumerate}[label=(\alph*)]
\item $\sqrt{N}\left[\bgN(\predictiveN)(\testfn) - \bg(\predictive)(\testfn)\right] \dconverges \N\left( 0, \bar{\V}_{n}(\testfn)\right),$
\item $\sqrt{N}\left[\hat{\eta}_n^N(\testfn) - \hat{\eta}_n(\testfn)\right]\dconverges \N\left(0, \V_n(\testfn)\right)$,
 \end{enumerate}
where $ \bar{\V}_{n}(\testfn), \V_n(\testfn)$ are given by the following recursion, 
\begin{align*}
\widehat{\V}_n(\testfn) &= \var_{\predictive}(\weight\testfn) + \bar{\V}_{n-1}(K_n(U_n\testfn));\\
\bar{\V}_{n}(\testfn) &= \frac{1}{\predictive(G_n)^2}\widehat{\V}_n\left(\testfn - \bg(\predictive)(\testfn)\right);\\
\V_n(\testfn)&=\var_{\update}(\testfn) + \bar{\V}_{n}(\testfn),
\end{align*}
with initial condition $\bar{\V}_0(\testfn) = \var_{M_0}\left(\frac{G_0}{\eta_0(G_0)}(\testfn-\hat{\eta}_0(\testfn))\right)$.
\end{prop}
\subsection{Method of proof}
\label{sec:mop}
We emphasize that the proof techniques used are small extensions of those used in the standard SMC setting; the primary interest of these results is that they demonstrate that the marginal version of the algorithm inherits many of the good properties of the standard algorithm and allowing comparison between the two algorithms via their asymptotic variances.
The details of the proofs are postponed to~\cref{app:lp,app:wlln,app:clt,app:bias,app:variance} where we give self-contained arguments incorporating the novel elements discussed in this section with well-known techniques used to obtain similar results for standard SMC.

The main difference between standard sequential Monte Carlo methods and marginal SMC methods is the presence of an additional approximations in the weights~\eqref{eq:WN}; if we could use the idealized algorithm in which $\weight$ in~\eqref{eq:W} can be computed exactly, then we could apply the theoretical results for standard SMC (e.g., \cite{del2013mean}).
Hence, to obtain convergence results for marginal SMC we need to control the behaviour of the non-standard weights $\weightN$. We point out that since the weights $\weightN$ are biased approximations of $\weight$, we cannot use the arguments based on extensions of the state space (as in particle filters using unbiased estimates of the potentials; \cite{fearnhead2008particle}) to provide theoretical guarantees for MSMC.

To control the effect of $\weightN$, we identify the conditional expectation of $\predictiveN(\weightN\testfn)$ in \cref{prop:conditional_expe} extending the result of \citet[Lemma 1]{branchini2021optimized}. \cref{prop:conditional_expe}, whose simple proof is provided in~\cref{app:expe}, combined with a number of results presented in the appendix which employ this conditional expectation within expansions of various aspects of the sampling error allow us to obtain results for MSMC algorithms in some generality.

\begin{prop}
\label{prop:conditional_expe}
Let $\mathcal{F}_{n-1}^N$ denote the $\sigma$-field generated by the weighted samples up to (and including) time $n-1$. We have
$\Exp\left[\predictiveN(\weightN\testfn)\mid \mathcal{F}_{n-1}^N\right] =\Psi_{G_{n-1}^N}(\eta_{n-1}^N)(K_n(\testfn U_n)),$
for all $n\geq 1$ and all $\testfn\in\bounded$.
\end{prop}

For the proof of the $\lp$-inequality we combine the results of \cite{smc:theory:CD02, miguez2013convergence} which control the approximation introduced by the mutation step, the reweighting step and the resampling step with \cref{lp:lemma_weight} in~\cref{app:lp} which controls the error induced by using the approximated weights~\eqref{eq:WN}.
The proof of \cref{lp:lemma_weight} is based on a comparison between $\predictiveN(\weightN\testfn)$ and its conditional expectation given in \cref{prop:conditional_expe} combined with the inductive hypothesis and \cref{prop:conditional_expe}.
Similarly, the proof of the bias estimates in \cref{prop:bias} combines the inductive approach of~\cite{olsson2004bootstrap} with \cref{prop:conditional_expe}, showing that the approximate weights do not worsen the rate of decay of bias with $N$.

For the proof of the weak law of large numbers we again adopt an inductive strategy similar to that of \citet[Theorem 1]{douc2007limit} and \citet[Theorem 9.4.5]{cappe2005inference}, we combine standard arguments establishing a WLLN for the mutation, the reweighting and the resampling step with \cref{wlln:lemma_weight} in~\cref{app:wlln} which establishes a WLLN for $\predictiveN(\weightN\testfn)$.

The central limit theorem follows using the inductive approach of \cite{chopin2004central}, with an additional result (\cref{clt:weight_comparison}) which shows that $\predictiveN(\weightN\testfn)$ satisfies a central limit theorem. 
\subsection{Variance comparison}
\label{sec:variance_comparison}

\cref{prop:clt} gives a recursive formula for the asymptotic variance similar to that in \cite{chopin2004central}. Comparing \cref{prop:clt} and \citet[Section 2.3]{chopin2004central}, we find that the main difference in the variance expression appears in $\widehat{\V}_n(\testfn)$: in the case of standard SMC we have $\widehat{\V}^{\textrm{SMC}}_n(\testfn) =\var_{\predictive}(\weight\testfn)+ \bar{\V}_{n-1}(M_n(G_n\testfn))$.
This is not surprising since the main difference between SMC and MSMC is in the importance weights~\eqref{eq:W}, whose expression is taken into account in $\widehat{\V}_n(\testfn)$.

Using an inductive argument whose details are given in~\cref{app:variance_closed} (see also \citet[Eq. 9]{chopin2004central} and \citet[Appendix A]{amj26:JD07}), we obtain a closed form for $\bar{\V}_n(\testfn)$.
To this end, let us define the following operator akin that of \citet[Eq. 10]{chopin2004central} and \citet[Section 2.7.2]{smc:theory:Del04}: 
$$
  \forall q \in \mathbb{N}:\qquad
\Gamma_{q}(\testfn)(x_{q-1}) := \int K_q(x_{q-1}, \rmd x_q)U_q(x_{q-1}, x_q)\testfn(x_{q}),
$$
and the two-parameter semigroup $\Gamma_{p:q}(\testfn) = \Gamma_{p+1}(\testfn)\circ \dots\circ\Gamma_{q}(\testfn)$ for all $p<q$, with the convention $\Gamma_{q:q}=\textsf{Id}$.

We focus on the variance of the estimates before the resampling step which are usually preferred to those given by $\updateN$, as their finite$-N$ variance is smaller as a consequence of the Rao-Blackwell Theorem \citep{blackwell1947conditional, rao1992information}. \cref{app:variance_closed} provides an explicit characterization of the variance:
\begin{prop}
\label{prop:variance_closed}
The variance $\bar{\V}_n(\testfn)$ in \cref{prop:clt} can be equivalently written as
\begin{align}
\label{eq:variance_expression}
\bar{\V}_n(\testfn) = \sum_{k=0}^{n}\Exp_{\eta_k}\left[\left(G_k\left[\Gamma_{k:n}(\testfn) - \hat{\eta}_n(\testfn)\Gamma_{k:n}(1)\right]\right)^2\right]\prod_{j=k}^{n}\frac{1}{\eta_{j}(G_{j})^2},
\end{align}
for all $n\geq0$ and all $\testfn\in\bounded$.
\end{prop}

We now discuss some special cases of marginal SMC algorithms, obtain their asymptotic variances and compare them with those of their non-marginal counterparts.
The variances of the special cases considered here can be obtained from~\eqref{eq:variance_expression} by simple algebraic manipulations which we postpone to~\cref{app:variance}. See~\cref{app:variance_mpf,app:variance_comparison} for the calculations underlying the following corollary: 
\begin{corollary}[Variance of MPF]
\label{cor:mpf}
For all $n\geq 0$ and all $\testfn\in\bounded$ we have
\begin{align*}
V_n^{\textrm{MPF}}(\testfn) &=\int  \frac{p(x_0|y_{1:n})^2}{
q_0(x_0)}\left(\int p(x_{n}|y_{1:n}, x_0)\left[\testfn(x_n)-\bar{\testfn}_n\right]\rmd x_{n}\right)^2\rmd x_0\\
&+ \int \frac{p(x_k|y_{1:n})^2}{\int q_k(x_k|x_{k-1}, y_k)p(x_{k-1}|y_{1:k-1})\rmd x_{k-1}}\left(\int p(x_{n}|y_{k+1:n}, x_k)\left[\testfn(x_n)-\bar{\testfn}_n\right]\rmd x_{n}\right)^2\rmd x_k\notag\\
&+\int \frac{p(x_{n}|y_{1:n})^2}{\int q_{n}(x_n|x_{n-1}, y_n) p(x_{n-1}|y_{1:n-1})\rmd x_{n-1}} (\testfn(x_n)-\bar{\testfn}_n)^2\rmd x_{n},\notag
\end{align*}
with $\bar{\testfn}_n:= \int \testfn(x_n)p(x_n|y_{1:n})\rmd x_n$.
In addition, $V_n^{\textrm{MPF}}(\testfn) \leq V_n^{\textrm{PF}}(\testfn)$, where $V_n^{\textrm{PF}}(\testfn)$ denotes the asymptotic variance of a particle filter with the same proposals given in \citet[Section 2.4]{johansen2008note}.
\end{corollary}

Using the results in \cite{walker2014jensen}, we can show that equality occurs only when $q_k\equiv f_k$ and MPF coincides with the bootstrap particle filter of \cite{smc:methodology:GSS93}; the variance expression in \cref{cor:mpf} then coincides with that of BPF given explicitly in \citet[Section 2.4]{johansen2008note} --- see~\cref{app:variance_bpf}. 
In all other cases, the variance reduction can be quantified using, e.g., \citet[Theorem 3.1]{walker2014jensen}.

The variance of the MAPF can be obtained from $\bar{V}_n(\testfn)$ in \cref{prop:variance_closed} via an additional importance sampling step with weights $\tilde{w}_n$ in~\eqref{eq:inferential_weights}, as shown in~\cref{app:variance_mapf} (and compared with the standard case in~\cref{app:variance_comparison}).
\begin{corollary}[Variance of MAPF]
\label{cor:mapf}
For all $n\geq 0$ and all $\testfn\in\bounded$ we have
 \begin{align*}
 \V^{\textrm{MAPF}}_n(\testfn) 
 &= \int \frac{p(x_0|y_{1:n})^2}{ q_0(x_0)}\left(\int p(x_n|x_0, y_{1:n})\left[\testfn(x_n)-\bar{\testfn}_n\right]\rmd x_{n}\right)^2\rmd x_0\\
 &+\int\frac{p(x_{k}|y_{1:n})^2}{\int q_k(x_k|x_{k-1}', y_k)\hat{\eta}_{k-1}(\rmd x_{k-1}')}\left(\int p(x_n|x_k, y_{1:n})\left[\testfn(x_n)-\bar{\testfn}_n\right]\rmd x_{n}\right)^2\rmd x_k\\
 &+\int\frac{p(x_{n}|y_{1:n})^2\left[\testfn(x_n) - \bar{\testfn}_n\right]^2}{\int q_n(x_n|x_{n-1}', y_n)\hat{\eta}_{n-1}(\rmd x_{n-1}')}\rmd x_n.
 \end{align*}
with $\bar{\testfn}_n:= \int \testfn(x_n)p(x_n|y_{1:n})\rmd x_n$.
In addition, $V_n^{\textrm{MAPF}}(\testfn) \leq V_n^{\textrm{APF}}(\testfn)$, where $V_n^{\textrm{APF}}(\testfn)$ denotes the asymptotic variance of an APF with the same proposals given in \citet[Section 2.4]{johansen2008note}.
\end{corollary}

Using the results in \cite{walker2014jensen}, we can show that equality occurs only when $q_k\propto C_k f_k$ where $C_k$ is a positive function only depending on $x_k$ and $\tilde{p}(y_k|x_{k-1}) = \int C_k(x_k)f_k(x_k|x_{k-1})\rmd x_k$.
A special case is $C_k(x_k)=g_k(y_k|x_k)$ for which the MAPF collapses onto the fully adapted APF (FA-APF) and hence has the same asymptotic variance (see \citet[Corollary]{johansen2008note};~\cref{app:variance_fa-mapf}).

\section{Discussion}
\label{sec:discussion}

In this work we established that a class of marginal sequential Monte Carlo (MSMC) algorithms, which encompasses marginal particle filters and other popular algorithms in the literature, satisfies many of the key properties that standard sequential Monte Carlo methods have.
The results in \cref{sec:results} guarantee that the estimates provided by MSMC are consistent, asymptotically normal, with a bias decaying at rate $N^{-1}$ and $\lp$ error decaying at rate $N^{-1/2}$.

Comparing the asymptotic variances in \cref{prop:clt} with those for standard SMC obtained in, e.g., \cite{chopin2004central, smc:theory:Del04}, we find that marginal particle filters have no larger asymptotic variance than the corresponding non-marginal particle filter (\cref{cor:mpf,cor:mapf}), a phenomenon already observed empirically (e.g. \cite{smc:methodology:KFD05, xu2019particle}).
\cref{cor:mapf} complements \citet[Proposition 1]{smc:methodology:KFD05} showing that the importance sampling weights of AMPF have lower (conditional) variance.

Unbiasedness of the normalizing constant estimates has been shown in \citet[Theorem 2]{branchini2021optimized}.
Combining \cref{prop:conditional_expe} with the approach of \citet[Lemma 2]{branchini2021optimized} one can further show that the unnormalized flow provides unbiased estimates.

Quantifying the variance reduction obtained by employing marginal particle filters instead of standard particle filter is a more challenging question, the answer to which is likely to be dependent on the specific state space model and proposals $q_k$. The variance reduction obtained by using MPF instead of PF should then be weighted against the additional computation cost required by MPF w.r.t. the $O(N)$ cost of PF. In their most naive implementation MPF require an $O(N^2)$ cost, which can however be reduced to $O(N\log N)$ using techniques from $N$-body learning (e.g. \cite{gray2000n, lang2005empirical}) as shown in \cite{smc:methodology:KFD05} or efficient implementations using GPUs \citep{charlier2021kernel} as shown in \citet[Section 4]{clarte2019collective} for sums of the form of those in~\eqref{eq:WN}.

Finally, 
we expect that similar ideas could be applied to more sophisticated SMC algorithms employing marginalization techniques (e.g. \cite{xu2019particle, crucinio2022divide}).
\section*{Acknowledgements}
FRC and AMJ acknowledge support from the EPSRC (grant \#  EP/R034710/1). AMJ acknowledges further support from the  EPSRC (grant \# EP/T004134/1) and the Lloyd's Register Foundation Programme on Data-Centric Engineering at the Alan Turing Institute. 
For the purpose of open access, the authors have applied a Creative Commons Attribution (CC BY) licence to any Author Accepted Manuscript version arising from this submission. No new data was created or analysed in this study. Data sharing is not applicable to this article.

\bibliographystyle{plainnat}
\bibliography{mpf_biblio}

\begin{thebibliography}{45}
\providecommand{\natexlab}[1]{#1}
\providecommand{\url}[1]{\texttt{#1}}
\expandafter\ifx\csname urlstyle\endcsname\relax
  \providecommand{\doi}[1]{doi: #1}\else
  \providecommand{\doi}{doi: \begingroup \urlstyle{rm}\Url}\fi

\bibitem[Berti et~al.(2006)Berti, Pratelli, and Rigo]{berti2006almost}
Patrizia Berti, Luca Pratelli, and Pietro Rigo.
\newblock Almost sure weak convergence of random probability measures.
\newblock \emph{Stochastics and Stochastics Reports}, 78\penalty0 (2):\penalty0
  91--97, 2006.

\bibitem[Billingsley(1995)]{billingsley1995measure}
P.~Billingsley.
\newblock \emph{Probability and Measure}.
\newblock John Wiley \& Sons., 1995.

\bibitem[Blackwell(1947)]{blackwell1947conditional}
D.~Blackwell.
\newblock Conditional expectation and unbiased sequential estimation.
\newblock \emph{Annals of Mathematical Statistics}, 18\penalty0 (1):\penalty0
  105--110, 1947.

\bibitem[Boustati et~al.(2020)Boustati, Akylid{ì}z, Damoulas, and
  Johansen]{amj26:BADJ20}
A.~Boustati, {\"O}.~D. Akylid{ì}z, T.~Damoulas, and A.~M. Johansen.
\newblock Generalized {B}ayesian filtering via sequential {M}onte {C}arlo.
\newblock In H.~Larochelle, M.~Ranzato, R.~Hadsell, M.~F. Balcan, and H.~Lin,
  editors, \emph{Advances in Neural Information Processing Systems}, volume~33,
  pages 418--429. Curran Associates, Inc., 2020.

\bibitem[Branchini and Elvira(2021)]{branchini2021optimized}
N.~Branchini and V.~Elvira.
\newblock Optimized auxiliary particle filters: adapting mixture proposals via
  convex optimization.
\newblock In \emph{37th Conference on Uncertainty in Artificial Intelligence},
  pages 1289--1299. Proceedings of Machine Learning Research, 2021.

\bibitem[Capp{\'e} et~al.(2005)Capp{\'e}, Moulines, and
  Ryd{\'e}n]{cappe2005inference}
O.~Capp{\'e}, E.~Moulines, and T.~Ryd{\'e}n.
\newblock \emph{{Inference in hidden Markov models}}.
\newblock Springer, 2005.

\bibitem[Carpenter et~al.(1999)Carpenter, Clifford, and
  Fearnhead]{smc:methodology:CCF99}
J.~Carpenter, P.~Clifford, and P.~Fearnhead.
\newblock An improved particle filter for non-linear problems.
\newblock \emph{IEE Proceedings on Radar, Sonar and Navigation}, 146\penalty0
  (1):\penalty0 2--7, 1999.

\bibitem[Charlier et~al.(2021)Charlier, Feydy, Glaunes, Collin, and
  Durif]{charlier2021kernel}
Benjamin Charlier, Jean Feydy, Joan~Alexis Glaunes, Fran{\c{c}}ois-David
  Collin, and Ghislain Durif.
\newblock Kernel operations on the {GPU}, with autodiff, without memory
  overflows.
\newblock \emph{Journal of Machine Learning Research}, 22\penalty0
  (74):\penalty0 1--6, 2021.

\bibitem[Chopin(2002)]{chopin2002sequential}
N.~Chopin.
\newblock A sequential particle filter method for static models.
\newblock \emph{Biometrika}, 89\penalty0 (3):\penalty0 539--552, 2002.

\bibitem[Chopin(2004)]{chopin2004central}
N.~Chopin.
\newblock Central limit theorem for sequential {Monte Carlo} methods and its
  application to {Bayesian} inference.
\newblock \emph{Annals of Statistics}, 32\penalty0 (6):\penalty0 2385--2411,
  2004.

\bibitem[Chopin and Papaspiliopoulos(2020)]{chopin2020}
N.~Chopin and O.~Papaspiliopoulos.
\newblock \emph{An Introduction to Sequential {Monte} {Carlo}}.
\newblock Springer, Cham, 2020.

\bibitem[Clart{\'e} et~al.(2022)Clart{\'e}, Diez, and
  Feydy]{clarte2019collective}
Gr{\'e}goire Clart{\'e}, Antoine Diez, and Jean Feydy.
\newblock {Collective proposal distributions for nonlinear MCMC samplers:
  Mean-field theory and fast implementation}.
\newblock \emph{Electronic Journal of Statistics}, 22\penalty0 (2):\penalty0
  6395--6460, 2022.

\bibitem[Crisan and Doucet(2002)]{smc:theory:CD02}
D.~Crisan and A.~Doucet.
\newblock A survey of convergence results on particle filtering methods for
  practitioners.
\newblock \emph{IEEE Transactions on Signal Processing}, 50\penalty0
  (3):\penalty0 736--746, March 2002.

\bibitem[Crucinio and Johansen(2023)]{crucinio2022divide}
F.~R. Crucinio and A.~M. Johansen.
\newblock A divide and conquer sequential {Monte Carlo} approach to high
  dimensional filtering.
\newblock \emph{Statistica Sinica}, 2023.
\newblock \doi{10.5705/ss.202022.0243}.
\newblock In press.

\bibitem[Del~Moral(1996)]{del1996non}
P.~Del~Moral.
\newblock Nonlinear filtering: Interacting particle resolution.
\newblock \emph{Markov {P}rocesses and {R}elated {F}ields}, 2\penalty0
  (4):\penalty0 555--580, 1996.

\bibitem[Del~Moral(2004)]{smc:theory:Del04}
P.~Del~Moral.
\newblock \emph{{Feynman-Kac} formulae: genealogical and interacting particle
  systems with applications}.
\newblock Probability and Its Applications. Springer Verlag, New York, 2004.

\bibitem[Del~Moral(2013)]{del2013mean}
P.~Del~Moral.
\newblock \emph{Mean field simulation for {Monte Carlo} integration}.
\newblock Chapman and Hall/CRC, New York, 2013.

\bibitem[{Del Moral} et~al.(2006){Del Moral}, Doucet, and
  Jasra]{smc:methodology:DDJ06b}
P.~{Del Moral}, A.~Doucet, and A.~Jasra.
\newblock Sequential {Monte Carlo} samplers.
\newblock \emph{Journal of the Royal Statistical Society B}, 63\penalty0
  (3):\penalty0 411--436, 2006.

\bibitem[Didelot et~al.(2011)Didelot, Everitt, Johansen, and
  Lawson]{amj26:DEJL11}
X.~Didelot, R.~G. Everitt, A.~M. Johansen, and D.~J. Lawson.
\newblock Likelihood-free estimation of model evidence.
\newblock \emph{Bayesian Analysis}, 6\penalty0 (1):\penalty0 49--76, 2011.

\bibitem[Douc and Moulines(2008)]{douc2007limit}
R.~Douc and R.~Moulines.
\newblock Limit theorems for weighted samples with applications to sequential
  {Monte Carlo} methods.
\newblock \emph{Annals of Statistics}, 36\penalty0 (5):\penalty0 2344--2376,
  2008.

\bibitem[Everitt et~al.(2017)Everitt, Prangle, Maybank, and
  Bell]{everitt2017marginal}
R.~G. Everitt, D.~Prangle, P.~Maybank, and M.~Bell.
\newblock Marginal sequential {Monte Carlo} for doubly intractable models.
\newblock \emph{arXiv preprint arXiv:1710.04382}, 2017.

\bibitem[Fearnhead et~al.(2008)Fearnhead, Papaspiliopoulos, and
  Roberts]{fearnhead2008particle}
Paul Fearnhead, Omiros Papaspiliopoulos, and Gareth~O Roberts.
\newblock Particle filters for partially observed diffusions.
\newblock \emph{Journal of the Royal Statistical Society: Series B (Statistical
  Methodology)}, 70\penalty0 (4):\penalty0 755--777, 2008.

\bibitem[Gerber et~al.(2019)Gerber, Chopin, and Whiteley]{Gerber2019}
M.~Gerber, N.~Chopin, and N.~Whiteley.
\newblock Negative association, ordering and convergence of resampling methods.
\newblock \emph{Annals of Statistics}, 47\penalty0 (4):\penalty0 2236--2260,
  2019.

\bibitem[Gordon et~al.(1993)Gordon, Salmond, and Smith]{smc:methodology:GSS93}
N.~J. Gordon, S.~J. Salmond, and A.~F.~M. Smith.
\newblock Novel approach to nonlinear/non-{Gaussian} {Bayesian} state
  estimation.
\newblock \emph{{IEE} Proceedings-F}, 140\penalty0 (2):\penalty0 107--113,
  1993.

\bibitem[Gray and Moore(2000)]{gray2000n}
Alexander Gray and Andrew Moore.
\newblock `{$N$}-body' problems in statistical learning.
\newblock In \emph{Advances in Neural Information Irocessing Systems}, pages
  521--527, 2000.

\bibitem[Johansen and Doucet(2007)]{amj26:JD07}
A.~M. Johansen and A.~Doucet.
\newblock Auxiliary variable sequential {Monte Carlo} methods.
\newblock Research Report 07:09, University of Bristol, Department of
  Mathematics -- Statistics Group, University Walk, Bristol, BS8 1TW, UK, July
  2007.

\bibitem[Johansen and Doucet(2008)]{johansen2008note}
Adam~M Johansen and Arnaud Doucet.
\newblock A note on auxiliary particle filters.
\newblock \emph{Statistics \& Probability Letters}, 78\penalty0 (12):\penalty0
  1498--1504, 2008.

\bibitem[Klaas et~al.(2005)Klaas, de~Freitas, and
  Doucet]{smc:methodology:KFD05}
M.~Klaas, N.~de~Freitas, and A.~Doucet.
\newblock Towards practical {$N^2$} {Monte Carlo}: The marginal particle
  filter.
\newblock In \emph{Proceedings of Uncertainty in Artificial Intelligence},
  pages 308--315, 2005.

\bibitem[Lai et~al.(2022)Lai, Domke, and Sheldon]{lai2022variational}
J.~Lai, J.~Domke, and D.~Sheldon.
\newblock Variational marginal particle filters.
\newblock In \emph{International Conference on Artificial Intelligence and
  Statistics}, pages 875--895. PMLR, 2022.

\bibitem[Lang et~al.(2005)Lang, Klaas, and de~Freitas]{lang2005empirical}
Dustin Lang, Mike Klaas, and Nando de~Freitas.
\newblock Empirical testing of fast kernel density estimation algorithms.
\newblock \emph{Technical Report TR2005-03, Department of Computer Science,
  University of British Columbia}, 2005.

\bibitem[Lin et~al.(2005)Lin, Zhang, Cheng, and Chen]{smc:methodology:LZCC05}
M.~T. Lin, J.~L. Zhang, Q.~Cheng, and R.~Chen.
\newblock Independent particle filters.
\newblock \emph{Journal of the {American} Statistical Association},
  100\penalty0 (472):\penalty0 1412--1421, 2005.

\bibitem[Liu(2001)]{liu2001monte}
J.~S. Liu.
\newblock \emph{{Monte Carlo} Strategies in Scientific Computing}.
\newblock Springer, New York, 2001.

\bibitem[M{\'\i}guez et~al.(2013)M{\'\i}guez, Crisan, and
  Djuri{\'c}]{miguez2013convergence}
J.~M{\'\i}guez, D.~Crisan, and P.~M. Djuri{\'c}.
\newblock On the convergence of two sequential {M}onte {C}arlo methods for
  maximum a posteriori sequence estimation and stochastic global optimization.
\newblock \emph{Statistics \& Computing}, 23\penalty0 (1):\penalty0 91--107,
  2013.

\bibitem[Olsson and Ryd{\'e}n(2004)]{olsson2004bootstrap}
J.~Olsson and T.~Ryd{\'e}n.
\newblock The bootstrap particle filtering bias.
\newblock \emph{Lund University, Technical Report 929081}, 2004.

\bibitem[Pitt and Shephard(1999)]{smc:methodology:PS99}
M.~K. Pitt and N.~Shephard.
\newblock Filtering via simulation: Auxiliary particle filters.
\newblock \emph{Journal of the {American} Statistical Association}, 94\penalty0
  (446):\penalty0 590--599, 1999.

\bibitem[Poyiadjis et~al.(2011)Poyiadjis, Doucet, and
  Singh]{poyiadjis2011particle}
G.~Poyiadjis, A.~Doucet, and S.~S. Singh.
\newblock Particle approximations of the score and observed information matrix
  in state space models with application to parameter estimation.
\newblock \emph{Biometrika}, 98\penalty0 (1):\penalty0 65--80, 2011.

\bibitem[Rao(1945)]{rao1992information}
C.~R. Rao.
\newblock Information and the accuracy attainable in the estimation of
  statistical parameters.
\newblock \emph{Bulletin of the Calcutta Mathematical Society}, 37\penalty0
  (3):\penalty0 81--91, 1945.

\bibitem[Schmon et~al.(2021)Schmon, Deligiannidis, Doucet, and
  Pitt]{schmon2018large}
Sebastian~M Schmon, George Deligiannidis, Arnaud Doucet, and Michael~K Pitt.
\newblock Large sample asymptotics of the pseudo-marginal method.
\newblock \emph{Biometrika}, 108\penalty0 (1):\penalty0 37--51, 2021.

\bibitem[Shiryaev(1996)]{shiryaev1996probability}
A.~N. Shiryaev.
\newblock \emph{Probability}, volume~25 of \emph{Graduate Texts in
  Mathematics}.
\newblock Springer, New York, 1996.

\bibitem[Sisson et~al.(2007)Sisson, Fan, and Tanaka]{smc:methodology:SFT07}
S.~A. Sisson, Y.~Fan, and M.~M. Tanaka.
\newblock Sequential {Monte Carlo} without likelihoods.
\newblock \emph{Proceedings of the National Academy of Sciences, {USA}},
  104\penalty0 (4):\penalty0 1760--1765, February 2007.

\bibitem[Walker(2014)]{walker2014jensen}
Stephen~G. Walker.
\newblock On a lower bound for the {Jensen} inequality.
\newblock \emph{SIAM Journal on Mathematical Analysis}, 46\penalty0
  (5):\penalty0 3151--3157, 2014.

\bibitem[Whiteley(2013)]{whiteley2013stability}
N.~Whiteley.
\newblock Stability properties of some particle filters.
\newblock \emph{Annals of Applied Probability}, 23\penalty0 (6):\penalty0
  2500--2537, 2013.

\bibitem[Xu and Jasra(2019)]{xu2019particle}
Yaxian Xu and Ajay Jasra.
\newblock Particle filters for inference of high-dimensional multivariate
  stochastic volatility models with cross-leverage effects.
\newblock \emph{Foundations of Data Science}, 1\penalty0 (1):\penalty0 61--85,
  2019.

\bibitem[Yuan and Li(2015)]{Yuan2015}
D.-M. Yuan and S.-J. Li.
\newblock Extensions of several classical results for independent and
  identically distributed random variables to conditional cases.
\newblock \emph{Journal of the Korean Mathematical Society}, 52\penalty0
  (2):\penalty0 431--445, 2015.

\bibitem[Zhou et~al.(2016)Zhou, Johansen, and Aston]{zhou2016toward}
Y.~Zhou, A.~M. Johansen, and J.~A.~D. Aston.
\newblock Toward automatic model comparison: an adaptive sequential {Monte
  Carlo} approach.
\newblock \emph{Journal of Computational and Graphical Statistics}, 25\penalty0
  (3):\penalty0 701--726, 2016.

\end{thebibliography}

\newpage
\appendix

\section{Notation}
Before tackling the proofs of the results stated in the main text, 
we summarize here the definitions of all the quantities involved in the proofs in the Appendices.

At time $n$, we denote the target distribution in~\eqref{eq:update} (also known as updated distribution) as $\update$, convolving the proposal kernel $M_n$ with $\hat{\eta}_{n-1}$ we obtain the mutated distribution (also known as predictive distribution) $\predictive$ in~\eqref{eq:predictive}.
We observe that the target distribution can be obtained from the mutated one as
\begin{align*}
    \update(\testfn) = \frac{\predictive(\weight \testfn)}{\predictive(\weight)} = \bg(\predictive).
\end{align*}

The particle approximations to the quantities above obtained via MSMC are  as follows: the particle approximation to $\predictive$ is denoted by $\predictiveN = N^{-1} \sum_{i=1}^N \delta_{X_n^i}$, $\bgN(\predictiveN) = \sum_{i=1}^N W_n^i \delta_{X_n^i}$ and $\updateN = N^{-1} \sum_{i=1}^N \delta_{\widetilde{X}_n^i}$ both approximate $\bg(\predictive)=\update$.

We observe that the exact weights can be written equivalently as
\begin{align*}
    \weight(x_n) &= \frac{\textrm{d} \hat{\eta}_{n-1}(U_n \cdot K_n)}{\textrm{d} \hat{\eta}_{n-1} M_n}(x_n)\\
    &=\frac{\rmd \left( \int U_n(x_{n-1}, \cdot)K_n(x_{n-1}, \cdot)\hat{\eta}_{n-1}(\rmd x_{n-1})\right)}{\rmd \left(\int M_n(x_{n-1}, \cdot)\hat{\eta}_{n-1}(\rmd x_{n-1})\right)}(x_n)\\
    &= \frac{\rmd \left(\int U_{n}(x_{n-1},\cdot) K_n(x_{n-1},\cdot)\Psi_{G_{n-1}}(\eta_{n-1})(\rmd x_{n-1})\right)}{\rmd \left(\int M_n(x_{n-1},\cdot)\Psi_{G_{n-1}}(\eta_{n-1})(\rmd x_{n-1})\right)}(x_n)\\
    &= \frac{\rmd \left(\int U_{n}(x_{n-1},\cdot) K_n(x_{n-1},\cdot)G_{n-1}(x_{n-1})\eta_{n-1}(\rmd x_{n-1})\right)}{\rmd \left(\int M_n(x_{n-1},\cdot)G_{n-1}(x_{n-1})\eta_{n-1}(\rmd x_{n-1})\right)}(x_n),
\end{align*}
and the approximate weights as
\begin{align}
\label{eq:WN_alternative}
    \weightN(x_n) &= \frac{\textrm{d} \Psi_{G_{n-1}^N}(\eta_{n-1}^N)(U_n \cdot K_n)}{\textrm{d} \Psi_{G_{n-1}^N}(\eta_{n-1}^N) M_n}(x_n)\\
    &=\frac{\rmd \left(\int U_{n}(x_{n-1},\cdot) K_n(x_{n-1},\cdot)\Psi_{G_{n-1}^N}(\eta_{n-1}^N)(\rmd x_{n-1})\right)}{\rmd \left(\int M_n(x_{n-1},\cdot)\Psi_{G_{n-1}^N}(\eta_{n-1}^N)(\rmd x_{n-1})\right)}(x_n)\notag\\
    &= \frac{\rmd \left(\int U_{n}(x_{n-1},\cdot) K_n(x_{n-1},\cdot)G_{n-1}^N(x_{n-1})\eta_{n-1}^N(\rmd x_{n-1})\right)}{\rmd \left(\int M_n(x_{n-1},\cdot)G_{n-1}^N(x_{n-1})\eta_{n-1}^N(\rmd x_{n-1})\right)}(x_n).\notag
\end{align}

We also define the following $\sigma$-fields of which we will make frequent use: $\mathcal{F}_0^N:= \sigma\left(X_0^{i}: i\in\lbrace 1,\ldots,N\rbrace\right)$ and $\mathcal{G}_0^N:= \sigma\left(\widetilde{X}_0^{i}: i\in\lbrace 1,\ldots,N\rbrace\right) \vee \mathcal{F}_0^N$. More generally, we recursively define the $\sigma$-field generated by the weighted samples up to an including mutation at time $n$, $\mathcal{F}_{n}^N:= \sigma\left(X_n^i: i\in\lbrace 1,\ldots,N\rbrace\right)\vee \sfmutation$ and the $\sigma$-field generated by the particle system up to (and including) time $n$ before the mutation step at time $n+1$, $\mathcal{G}_n^N:= \sigma\left(\widetilde{X}_n^{i}: i\in\lbrace 1,\ldots,N\rbrace\right)\vee \mathcal{F}_{n}^N$.

\section{Proof of \cref{prop:conditional_expe}}
\label{app:expe}

\begin{proof}[Proof of \cref{prop:conditional_expe}]
Let $\mathcal{F}_{n-1}^N$ denote the $\sigma$-field generated by the weighted samples up to (and including) time $n-1$ and $\sfmutation$ denote the $\sigma$-field generated by the particle system up to (and including) time $n$ before the mutation step at time $n$, so that $ \mathcal{F}_{n-1}^N \subset \mathcal{G}_{n-1}^N$.
Consider the conditional expectation $\Exp\left[\predictiveN(\weightN\testfn)\mid \mathcal{F}_{n-1}^N\right]$, applying the tower property we obtain
 \begin{align}
 \label{eq:conditional_expe1}
\Exp\left[\predictiveN(\weightN\testfn)\mid \mathcal{F}_{n-1}^N\right] &=\frac{1}{N}\sum_{i=1}^N\Exp\left[\weightN(X_n^i)\testfn(X_n^i)|\mathcal{F}_{n-1}^N\right]\\
&=\frac{1}{N}\sum_{i=1}^N\Exp\left[\Exp\left[\weightN(X_n^i)\testfn(X_n^i)|\sfmutation\right]|\mathcal{F}_{n-1}^N\right]\notag\\
&=\frac{1}{N}\sum_{i=1}^N\Exp\left[\int M_n(\widetilde{X}_{n-1}^i, \rmd x_n)\weightN(x_n)\testfn(x_n)|\mathcal{F}_{n-1}^N\right]\notag\\
&=\sum_{j=1}^N\frac{G_{n-1}^N(X_{n-1}^j)}{\sum_{k=1}^NG_{n-1}^N(X_{n-1}^k)}\int M_n(X_{n-1}^j, \rmd x_n)\weightN(x_n)\testfn(x_n),\notag
\end{align}
where the third equality follows from the fact that  $X_n^i|\mathcal{G}_{n-1}^N \sim M_n(\widetilde{X}_{n-1}^i, \cdot)$ for each $i=1,\dots, N$ and the fourth from the fact that  $\{X_{n-1}^j\}_{j=1}^N$ and $G_{n-1}^N$ are $\mathcal{F}_{n-1}^N$-measurable and conditionally each $\widetilde{X}_{n-1}^i$ is drawn independently from the categorical distribution with probabilities given by the weights.
Plugging the definition of $\weightN$ in~\eqref{eq:WN} into the above we obtain
\begin{align*}
&\Exp\left[\predictiveN(\weightN\testfn)\mid \mathcal{F}_{n-1}^N\right] \\
&=\int \int M_n(x_{n-1}, \rmd x_n)\frac{\textrm{d} \Psi_{G_{n-1}^N}(\eta_{n-1}^N)(U_n \cdot K_n)}{\textrm{d} \Psi_{G_{n-1}^N}(\eta_{n-1}^N) M_n}(x_n)\testfn(x_n)\Psi_{G_{n-1}^N}(\eta_{n-1}^N)(\rmd x_{n-1})\\
&=\int U_{n}(x_{n-1},x_n) K_n(x_{n-1},\rmd x_n)\testfn(x_n)\Psi_{G_{n-1}^N}(\eta_{n-1}^N)(\rmd x_{n-1})\\
&=\Psi_{G_{n-1}^N}(\eta_{n-1}^N)(K_n(\testfn U_n)),
\end{align*}
as required.
\end{proof}
In addition, we have that, for all $n\geq 1$ and all $\testfn\in\bounded$, 

\begin{align}
\label{eq:conditional_expe2}
\predictive(\weight \testfn) & = \int\testfn(x_n)G_n(x_n)\predictive(\rmd x_n)\\
&= \int\int M_n(x_{n-1}, \rmd x_n) \testfn(x_n)\frac{\textrm{d} \hat{\eta}_{n-1}(U_n \cdot K_n)}{\textrm{d} \hat{\eta}_{n-1} M_n}(x_n)\hat{\eta}_{n-1}(\rmd x_{n-1}) \notag\\
&= \int \int\testfn(x_n) U_n(x_{n-1}, x_n) K_n(x_{n-1}, \rmd x_n) \hat{\eta}_{n-1}(\rmd x_{n-1})\notag\\
&=\hat{\eta}_{n-1}(K_n(\testfn U_n))\notag\\
&=\Psi_{G_{n-1}}(\eta_{n-1})(K_n(\testfn U_n)).\notag
\end{align}

\section{Proof of the $\lp$-inequalities in \cref{prop:lp}}
\label{app:lp}

As a preliminary we reproduce part of \citet[Lemma 7.3.3]{smc:theory:Del04}, a Marcinkiewicz-Zygmund-type inequality of which we will make extensive use.

\begin{lemma}[Del Moral, 2004] \label{lemma:delmoral}
Given a sequence of probability measures $(\mu_i)_{i \geq 1}$ on a given measurable space $(E,\mathcal{E})$ and a collection of independent random variables, one distributed according to each of those measures, $(X_i)_{i \geq 1}$, where $\forall i, X_i \sim \mu_i$, together with any sequence of measurable functions $(h_i)_{i \geq 1}$ such that $\mu_i(h_i) = 0$ for all $i \geq 1$, we define for any $N \in \mathbb{N}$,
$$m_N(X)(h) = \frac{1}{N} \sum_{i=1}^N h_i( X_i ) \ \textrm{ and } \ \sigma_N^2(h) = \frac{1}{N} \sum_{i=1}^N \left(\sup(h_i) - \inf(h_i) \right)^2.$$
If the $h_i$ have finite oscillations (i.e., $\sup(h_i)-\inf(h_i)<\infty \  \forall i \geq 1$) then we have:
$$\sqrt{N} \Exp\left[ \left\vert m_N(X)(h) \right\vert^p \right]^{1/p} \leq b(p)^{1/p} \sigma_N(h),$$
with, for any pair of integers $q,p$ such that $q \geq p \geq 1$, denoting $(q)_p=q!/(q-p)!$:
\begin{align}
    \label{eq:b(p)}
    b(2q) = (2q)_q 2^{-q} \ \textrm{ and } \  b(2q - 1) = \frac{(2q-1)_q}{\sqrt{q-\frac12}} 2^{-(q-\frac12)}. 
\end{align}
\end{lemma}

We are now ready to prove \cref{prop:lp}:
\begin{proof}[Proof of \cref{prop:lp}]
We proceed by induction, taking $n=0$ as the base case.
At time $n=0$, the particles $(X_0^{i})_{i=1}^N$ are sampled i.i.d. from $\eta_0\equiv M_0$, hence $\Exp\left[\testfn(X_0^{i})\right] = \eta_0(\testfn)$ for $i=1,\ldots,N$.
 We can define the sequence of functions $\Delta_0^i: E \mapsto \mathbb{R}$ for $i=1,\ldots, N$
\begin{equation*}
\Delta_0^i(x) :=  \testfn(x) -  \Exp\left[\testfn(X_{0}^{i})\right]
\end{equation*}
so that,
\begin{align*}
\eta_0^N(\testfn) - \eta_0(\testfn) = \frac{1}{N}\sum_{i=1}^N \Delta_0^i(X_0^{i}),
\end{align*}
and apply \cref{lemma:delmoral} to get for every $p\geq 1$
\begin{align}
\label{eq:lp0_iid}
\Exp\left[\vert \eta_0^N(\testfn) - \eta_0(\testfn)\vert^p\right]^{1/p} &\leq b(p) ^{1/p} \frac{1}{\sqrt{N}} \left(\sum_{i=1}^N \left(\sup( \Delta_{0}^i) - \inf( \Delta_{0}^i)\right)^2\right)^{1/2} \\
&\leq b(p) ^{1/p} \frac{1}{\sqrt{N}} \left(\sum_{i=1}^N 4\left(\sup\vert \Delta_{0}^i\vert\right)^2\right)^{1/2} \notag\\
& \leq b(p) ^{1/p} \frac{1}{\sqrt{N}} \left(\sum_{i=1}^N 16\supnorm{\testfn}^2\right)^{1/2} \notag \\
& \leq 4b(p) ^{1/p} \supnorm{\testfn}.\notag
\end{align}

To prove \ref{prop:lp-pred}, note that $\Psi_{G_0^N}(\eta_0^N)(\testfn)=\Psi_{G_0}(\eta_0^N)(\testfn)$ since $G^N_0 \equiv G_0 \equiv \rmd K_0/\rmd M_0$.
Using the triangle inequality we have
\begin{align*}
\vert \Psi_{G_0^N}(\eta_0^N)(\testfn) - \Psi_{G_0}(\eta_0)(\testfn)\vert&=\left\lvert \frac{\eta_0^N(G_0\testfn)}{\eta_0^N(G_0)}-\frac{\eta_0(G_0\testfn)}{\eta_0(G_0)}\right\rvert\\
&\leq  \frac{\vert \eta_0^N(G_0\testfn) - \eta_0(G_0\testfn)\vert}{\vert\eta_0^N(G_0)\vert}
+ \left\lvert\frac{\eta_0(G_0\testfn)}{\eta_0^N(G_0)\eta_0(G_0)}\right\rvert\vert \eta_0^N(G_0) - \eta_0(G_0)\vert.
\end{align*}
Using Minkowski's inequality and \cref{ass:smc2}, which guarantees $G_0\geq m_g^{-1}>0$, we then have
\begin{align*}
    \Exp\left[\vert \Psi_{G_0^N}(\eta_0^N)(\testfn) - \Psi_{G_0}(\eta_0)(\testfn)\vert^p\right]^{1/p} \leq& m_g\Exp\left[\vert \eta_0^N(G_0\testfn) - \eta_0(G_0\testfn)\vert^p\right]^{1/p}\\
    &+\supnorm{\testfn}m_g \Exp\left[\vert \eta_0^N(G_0) - \eta_0(G_0)\vert^p\right]^{1/p}\\
\leq& 8m_g^2b(p)^{1/p}\frac{\supnorm{\testfn}}{N^{1/2}},
\end{align*}
where we used~\eqref{eq:lp0_iid} and \cref{ass:weak} to guarantee that $G_0\leq m_g<\infty$ and obtain the last inequality. Hence, $\bar{C}_{0,p} = 8m_g^2b(p)^{1/p}$.
To obtain \ref{prop:lp-upd} from \ref{prop:lp-pred}, consider the decomposition
\begin{align*}
\hat{\eta}_0^N(\testfn) - \hat{\eta}_0(\testfn) &= \hat{\eta}_0^N(\testfn) - \Psi_{G_0^N}(\eta_0^N)(\testfn)  + \Psi_{G_0^N}(\eta_0^N)(\testfn)- \hat{\eta}_0(\testfn) \\
&= \hat{\eta}_0^N(\testfn) - \Exp\left[\hat{\eta}_0^N(\testfn)\mid \mathcal{F}_0^N\right]  + \Psi_{G_0^N}(\eta_0^N)(\testfn) -\hat{\eta}_0(\testfn),
\end{align*}
where $\mathcal{F}_0^N:= \sigma\left(X_0^{i}: i\in\lbrace 1,\ldots,N\rbrace\right)$ denotes the $\sigma$-field generated by the weighted samples at time $0$.
Using \ref{prop:lp-pred} we have that
\begin{align*}
\Exp\left[\vert\Psi_{G_0^N}(\eta_0^N)(\testfn) -\hat{\eta}_0(\testfn)\vert^p\right]^{1/p}\leq \bar{C}_{0,p}\frac{\supnorm{\testfn}}{N^{1/2}}
\end{align*}
since $\hat{\eta}_0(\testfn) = \Psi_{G_0}(\eta_0)(\testfn)$.
For the remaining term consider the sequence of functions $\Delta_0^i :  E \mapsto \real$ for $i=1,\ldots, N$ 
\begin{equation*}
\Delta_0^i(x) :=  \testfn(x) -  \Exp\left[\testfn(\widetilde{X}_{0}^{i}) \mid \mathcal{F}_0^N\right].
\end{equation*}
Conditionally on $\mathcal{F}_0^N$, $\Delta_0^i(\widetilde{X}_0^{i})$ $i=1,\ldots, N$ are independent and have expectation equal to 0, moreover
\begin{align*}
\hat{\eta}_0^N(\testfn) - \Exp\left[\hat{\eta}_0^N(\testfn)\mid \mathcal{F}_0^N\right]  &= \frac{1}{N} \sum_{i=1}^N \left( \testfn(\widetilde{X}_{0}^{i}) - \Exp\left[\testfn(\widetilde{X}_{0}^{i}) \mid \mathcal{F}_0^N\right]\right) = \frac{1}{N} \sum_{i=1}^N \Delta_0^i(\widetilde{X}_{0}^{i}).
\end{align*}
Using again the \cref{lemma:delmoral} (see also \citet[Theorem 3.3]{Yuan2015} for an explicit statement of a conditional form of a result of this type) we have, for all $p\geq1$,
\begin{align}
\label{eq:lp0}
\Exp\left[\vert \hat{\eta}_0^N(\testfn) - \Psi_{G_0^N}(\eta_0^N)(\testfn) \vert^p\right]^{1/p}\leq 4b(p)^{1/p}\frac{\supnorm{\testfn}}{N^{1/2}},
\end{align}
from which the result at $n=0$ follows and the base case is established.

Then, assume that the result holds at time $n-1$ for some $n$: we will show it also holds at time $n$.
The error introduced by the mutation step is controlled by \cref{lp:lemma1}:
\begin{equation}
\label{eq:lemma1}
\Exp\left[\vert\predictiveN(\testfn) - \predictive(\testfn)\vert^p\right]^{1/p} \leq \widetilde{C}_{p,n} \frac{\supnorm{\testfn}}{ N^{1/2}}.
 \end{equation}
Using \cref{lp:lemma_weight,lp:lemma3} we can control the error introduced by the reweighting
 \begin{equation*}
 \Exp\left[\vert\bgN(\predictiveN)(\testfn) - \bg(\predictive)(\testfn)\vert^p\right]^{1/p} \leq \bar{C}_{p,n} \frac{\supnorm{ \testfn}}{N^{1/2}}.
\end{equation*}
Finally, \cref{lp:lemma4} controls the error introduced by the resampling step
\begin{equation*}
\Exp\left[\vert\updateN(\testfn) - \update(\testfn)\vert^p\right]^{1/p} \leq C_{p,n} \frac{\supnorm{\testfn}}{N^{1/2}}.
\end{equation*}
The result follows for all $n\in \mathbb{N}$ by induction.
 \end{proof}
\subsection{Auxiliary results for the proof of \cref{prop:lp}}

We collect here four auxiliary results for the proof of the $\lp$-inequality in \cref{prop:lp}.
\cref{lp:lemma1,lp:lemma3,lp:lemma4} are well-known results for standard SMC methods and we report them for completeness while \cref{lp:lemma_weight} controls the additional error introduced by the use of the approximate weights.

First, we show that the mutation step preserves the error bounds; essentially this result can be found in \citet[Lemma 3]{smc:theory:CD02}.
\begin{lemma}[Mutation]
\label{lp:lemma1}
Under the conditions of \cref{prop:lp}, assume that for any $\testfn\in \bounded$, for some $p\geq 1$ and some finite constant $C_{p,n-1} $
\begin{equation*}
\Exp\left[\vert\hat{\eta}^N_{n-1}(\testfn) - \hat{\eta}_{n-1}(\testfn)\vert^p\right]^{1/p} \leq C_{p,n-1} \frac{\supnorm{\testfn}}{N^{1/2}}
\end{equation*}
then, after the mutation step
\begin{equation*}
\Exp\left[\vert\predictiveN(\testfn) - \predictive(\testfn)\vert^p\right]^{1/p} \leq \widetilde{C}_{p,n} \frac{\supnorm{\testfn}}{ N^{1/2}}
\end{equation*}
for any $\testfn\in \bounded$ and for some finite constant $\widetilde{C}_{p,n} $.
\end{lemma}
\begin{proof}
The proof follows that of \citet[Lemma 3]{smc:theory:CD02}. Divide into two terms and apply Minkowski's inequality
\begin{align*}
\Exp\left[\vert\predictiveN(\testfn) - \predictive(\testfn)\vert^p\right]^{1/p} & = \Exp\left[\vert\predictiveN(\testfn) - \hat{\eta}_{n-1} M_{n}(\testfn)\vert^p\right]^{1/p}\\
& \leq \Exp\left[\vert\predictiveN(\testfn) - \hat{\eta}_{n-1}^N M_{n}(\testfn)\vert^p\right]^{1/p}\\
& + \Exp\left[\vert\hat{\eta}_{n-1}^N M_{n}(\testfn)- \hat{\eta}_{n-1} M_{n}(\testfn) \vert^p\right]^{1/p}.
\end{align*}
Let $\sfmutation$ denote the $\sigma$-field generated by the particle system up to (and including) time $n-1$ before the mutation step at time $n$, $\sfmutation = \sigma\left(\widetilde{X}_p^i: i\in\lbrace 1,\ldots,N\rbrace\right)\vee \mathcal{F}_{n-1}^N$ and consider the sequence of functions $\Delta_{n}^i : E \mapsto \real$ for $i=1,\ldots, N$ 
\begin{equation*}
\Delta_{n}^i(x) := \testfn(x) -  \Exp\left[\testfn(X_n^i)\mid \sfmutation \right] = \testfn(x) - M_n\testfn(\widetilde{X}_{n-1}^{i}).
\end{equation*}
Conditionally on $\sfmutation$, $\Delta_{n}^i(X^{i}_{n}),\ i=1,\ldots, N$ are independent and have expectation equal to 0, moreover
\begin{align*}
\predictiveN(\testfn) - \hat{\eta}_{n-1}^N M_n(\testfn) &= \frac{1}{N}\sum_{i=1}^N \left[ \testfn(X_{n}^{i}) -  M_n\testfn(\widetilde{X}_{n-1}^{i})\right] =\frac{1}{N} \sum_{i=1}^N \Delta_{n}^i(X_{n}^{i}).
\end{align*}
Conditioning on $\sfmutation$ and applying \cref{lemma:delmoral} we have, for all $p\geq 1$,
\begin{align}
\label{eq:lp_mutation}
\sqrt{N}\Exp\left[\vert\predictiveN(\testfn) - \hat{\eta}_{n-1}^N M_n(\testfn)\vert^p\mid \sfmutation\right]^{1/p} &\leq b(p) ^{1/p} \frac{1}{\sqrt{N}} \left(\sum_{i=1}^N \left(\sup( \Delta_{n}^i) - \inf( \Delta_{n}^i)\right)^2\right)^{1/2} \\
&\leq b(p) ^{1/p} \frac{1}{\sqrt{N}} \left(\sum_{i=1}^N 4\left(\sup\vert \Delta_{n}^i\vert\right)^2\right)^{1/2} \notag\\
& \leq b(p) ^{1/p} \frac{1}{\sqrt{N}} \left(\sum_{i=1}^N 16\supnorm{\testfn}^2\right)^{1/2} \notag \\
& \leq 4b(p) ^{1/p} \supnorm{\testfn}\notag
\end{align}
with $b(p)$ as in~\eqref{eq:b(p)}.
Combining this result with the hypothesis yields
\begin{align*}
\Exp\left[\vert\predictiveN(\testfn) - \predictive(\testfn)\vert^p\right]^{1/p} \leq (4b(p)^{1/p} + C_{p,n-1} )\frac{\supnorm{\testfn}}{ N^{1/2}},
\end{align*}
and the result holds with $\widetilde{C}_{p,n}  = 4b(p)^{1/p} + C_{p,n-1} $.
\end{proof}

We then show that the use of the approximate weights $\weightN$ does not worsen the rate at which the error decays:
\begin{lemma}[Weight Comparison]
\label{lp:lemma_weight}
Under the conditions of \cref{prop:lp}, assume that for $p\geq 1$ and some finite constants and $\bar{C}_{p,n-1} $
\begin{equation*}
\Exp\left[\Psi_{G_{n-1}^N }(\eta_{n-1}^N)( \testfn)-  \Psi_{G_{n-1}}(\eta_{n-1})(\testfn)\vert^p\right]^{1/p}\leq \bar{C}_{p, n-1} \frac{\supnorm{\testfn}}{ N^{1/2}},
\end{equation*}
then
\begin{equation*}
\Exp\left[\predictiveN( \weightN \testfn)-  \predictive(\weight\testfn)\vert^p\right]^{1/p}\leq D_{p, n} \frac{1}{ N^{1/2}},
\end{equation*}
for any $\testfn\in \bounded$ and for some finite constant $D_{p,n} $.
\end{lemma}
\begin{proof}
Using the hypothesis, \cref{ass:weak}, \cref{prop:conditional_expe} and~\eqref{eq:conditional_expe2}, we have
\begin{align*}
\Exp\left[\vert \Exp\left[\predictiveN( \weightN \testfn) \mid \mathcal{F}_{n-1}^N\right]-  \predictive(\weight\testfn)\vert^p\right]^{1/p}&=\Exp\left[\vert \Psi_{G_{n-1}^N}(\eta_{n-1}^N)(K_n(\testfn U_n)) -  \Psi_{G_{n-1}}(\eta_{n-1})(K_n(\testfn U_n))\vert^p\right]^{1/p}\\
&\leq \bar{C}_{p, n-1}\supnorm{U_n} \frac{\supnorm{\testfn}}{ N^{1/2}}.
\end{align*}
It follows that in order to obtain the result, it suffices to bound
\begin{align*}
\Exp\left[\vert\predictiveN( \weightN \testfn) - \Exp\left[\predictiveN( \weightN \testfn) \mid \mathcal{F}_{n-1}^N\right]\vert^p\right]^{1/p}.
\end{align*}
Consider the sequence of functions $\Delta_n^i : E \mapsto \mathbb{R}$ for $i=1,\ldots, N$ 
\begin{equation*}
\Delta_n^i(x) :=  G_n^N(x)\testfn(x) -  \Exp\left[G_n^N(X_n^i)\testfn(X_{n}^i) \mid \mathcal{F}_{n-1}^N\right].
\end{equation*}
Conditionally on $\mathcal{F}_{n-1}^N$, $\Delta_n^i(X_n^i)$ $i=1,\ldots, N$ are independent and have expectation equal to 0, moreover
\begin{align*}
\predictiveN( \weightN \testfn) - \Exp\left[\predictiveN( \weightN \testfn) \mid \mathcal{F}_{n-1}^N\right] = \frac{1}{N}\sum_{i=1}^N \Delta_n^i(X_n^i).
\end{align*}
By \cref{lemma:delmoral}, we have almost surely
\begin{align*}
\sqrt{N}\Exp\left[\predictiveN( \weightN \testfn) - \Exp\left[\predictiveN( \weightN \testfn) \mid \mathcal{F}_{n-1}^N\right]\vert^p \mid \mathcal{F}_{n-1}^N\right]^{1/p} & \leq b(p) ^{1/p} \frac{1}{\sqrt{N}} \left(\sum_{i=1}^N \left(\sup( \Delta_{n}^i) - \inf( \Delta_{n}^i)\right)^2\right)^{1/2} \\
&\leq b(p) ^{1/p} \frac{1}{\sqrt{N}} \left(\sum_{i=1}^N 4\left(\sup\vert \Delta_{n}^i\vert\right)^2\right)^{1/2} \\
&  \leq b(p) ^{1/p} \frac{1}{\sqrt{N}} \left(\sum_{i=1}^N 16m_g^2\supnorm{\testfn}^2\right)^{1/2} \\
& \leq 4b(p) ^{1/p} m_g\supnorm{\testfn}
\end{align*}
where $b(p)$ is given in~\eqref{eq:b(p)}, and the result follows.
\end{proof}

Using \cref{lp:lemma_weight} above and following \citet[Lemma 4]{smc:theory:CD02} we obtain an error bound for the approximate reweighting.
\begin{lemma}[Approximate reweighting]
\label{lp:lemma3}
Under the conditions of \cref{prop:lp}, assume that for $p\geq 1$ and some finite constants and $D_{p,n} $
\begin{equation*}
\Exp\left[\predictiveN( \weightN \testfn)-  \predictive(\weight\testfn)\vert^p\right]^{1/p}\leq D_{p, n} \frac{\supnorm{\testfn}}{ N^{1/2}},
\end{equation*}
then
\begin{equation*}
\Exp\left[\vert\bgN(\predictiveN)(\testfn) - \bg(\predictive)(\testfn)\vert^p\right]^{1/p} \leq \bar{C}_{p,n} \frac{\supnorm{ \testfn}}{N^{1/2}}
\end{equation*}
for any $\testfn\in \bounded$ and for some finite constant $\bar{C}_{p,n} $.
\end{lemma}
\begin{proof}
Apply the definition of $\bg$ and $\bgN$ and consider the following decomposition
\begin{align*}
\vert\bgN(\predictiveN)(\testfn) - \bg(\predictiveN)(\testfn)\vert & = \left\lvert \frac{\predictiveN( \weightN\testfn)}{\predictiveN( \weightN)} - \frac{\predictive( \weight\testfn)}{\predictive( \weight)} \right\rvert\\
&\leq \left\lvert \frac{\predictiveN( \weightN\testfn)}{\predictiveN( \weightN)} - \frac{\predictiveN( \weightN\testfn)}{\predictive( \weight)} \right\rvert\\
& + \left\lvert \frac{\predictiveN( \weightN\testfn)}{\predictive( \weight)} - \frac{\predictive( \weight\testfn)}{\predictive( \weight)} \right\rvert.
\end{align*}
Then, for the first term
\begin{align*}
\left\lvert \frac{\predictiveN( \weightN\testfn)}{\predictiveN( \weightN)} - \frac{\predictiveN( \weightN \testfn)}{\predictive( \weight)} \right\rvert & = \left\lvert \frac{\predictiveN( \weightN\testfn)}{\predictiveN( \weightN)}\right\rvert \left\lvert \frac{\predictive(\weight) - \predictiveN(\weightN)}{\predictive( \weight)} \right\rvert\\
&\leq \frac{\supnorm{\testfn}}{\vert \predictive( \weight) \vert} \vert \predictive(\weight) - \predictiveN(\weightN)\vert.
\end{align*}
For the second term
\begin{align*}
\left\lvert \frac{\predictiveN( \weightN\testfn)}{\predictive( \weight)} - \frac{\predictive( \weight\testfn)}{\predictive( \weight)} \right\rvert & = \frac{1}{\vert \predictive( \weight) \vert}  \vert \predictiveN( \weightN\testfn) - \predictive( \weight\testfn)\vert.
\end{align*}
Hence, using \cref{ass:smc2} to guarantee that $G_n\geq m_g^{-1}>0$ and the hypothesis, we have
\begin{align*}
\Exp\left[\vert\bgN(\predictiveN)(\testfn) - \bg(\predictive)(\testfn)\vert^p\right]^{1/p} &\leq \supnorm{\testfn}m_g\Exp\left[\vert \predictive(\weight) - \predictiveN(\weightN)\vert^p\right]^{1/p}\\
&+m_g\Exp\left[\vert \predictive(\weight\testfn) - \predictiveN(\weightN\testfn)\vert^p\right]^{1/p}\\
&\leq 2\supnorm{\testfn}m_g^2\frac{D_{p,n}}{N^{1/2}}.
\end{align*}
\end{proof}

Finally, we control the error induced by the resampling step using the same argument as \citet[Lemma 5]{smc:theory:CD02}.

\begin{lemma}[Multinomial resampling]
\label{lp:lemma4}
Under the conditions of \cref{prop:lp}, assume that for any $\testfn\in \bounded$, for some $p\geq 1$ and some finite constant $\bar{C}_{p, n}$
\begin{equation*}
\Exp\left[\vert\bgN(\predictiveN)(\testfn) - \update(\testfn)\vert^p\right]^{1/p}\leq \bar{C}_{p, n}\frac{\supnorm{\testfn}}{N^{1/2}},
\end{equation*}
then after the resampling step performed through multinomial resampling
\begin{equation*}
\Exp\left[\vert\updateN(\testfn) - \update(\testfn)\vert^p\right]^{1/p} \leq C_{p,n}\frac{\supnorm{\testfn}}{N^{1/2}}
\end{equation*}
for any $\testfn\in \bounded$ and for some finite constant $C_{p,n} $.
\end{lemma}
\begin{proof}
The proof follows that of \citet[Lemma 5]{smc:theory:CD02}.
Divide into two terms and apply Minkowski's inequality
\begin{align*}
\Exp\left[\vert\updateN(\testfn) - \update(\testfn)\vert^p\right]^{1/p} & \leq \Exp\left[\vert\updateN(\testfn) - \bgN(\predictiveN)(\testfn)\vert^p\right]^{1/p}\\
&+\Exp\left[\vert\bgN(\predictiveN)(\testfn) - \update(\testfn)\vert^p\right]^{1/p}.
\end{align*}
Denote by $\sfresampling$
the $\sigma$-field generated by the weighted samples up to (and including) time $n$, $\sfresampling:= \sigma\left(X_n^i: i\in\lbrace 1,\ldots,N\rbrace\right)\vee\mathcal{G}_{n-1}^N$ and consider the sequence of functions $\Delta_n^i : E \mapsto \mathbb{R}$ for $i=1,\ldots, N$ 
\begin{equation*}
\Delta_n^i(x) :=  \testfn(x) -  \Exp\left[\testfn(\widetilde{X}_{n}^i) \mid \sfresampling\right].
\end{equation*}
Conditionally on $\sfresampling$, $\Delta_n^i(\widetilde{X}_n^i)$ $i=1,\ldots, N$ are independent and have expectation equal to 0, moreover
\begin{align*}
\updateN(\testfn) - \bgN(\predictiveN)(\testfn) &= \frac{1}{N} \sum_{i=1}^N \left( \testfn(\widetilde{X}_{n}^i) - \Exp\left[\testfn(\widetilde{X}_{n}^i) \mid \sfresampling\right]\right)\\
&\qquad  = \frac{1}{N} \sum_{i=1}^N \Delta_n^i(\widetilde{X}_n^i).
\end{align*}
By \cref{lemma:delmoral},
\begin{align}
\label{eq:lp_resampling}
\sqrt{N}\Exp\left[\vert\updateN(\testfn) - \bgN(\predictiveN)(\testfn)\vert^p \mid \sfresampling\right]^{1/p} &\leq b(p) ^{1/p} \frac{1}{\sqrt{N}} \left(\sum_{i=1}^N \left(\sup( \Delta_{n}^i) - \inf( \Delta_{n}^i)\right)^2\right)^{1/2}  \\
&\leq b(p) ^{1/p} \frac{1}{\sqrt{N}} \left(\sum_{i=1}^N 4\left(\sup\vert \Delta_{n}^i\vert\right)^2\right)^{1/2} \notag \\
&  \leq b(p) ^{1/p} \frac{1}{\sqrt{N}} \left(\sum_{i=1}^N 16\supnorm{\testfn}^2\right)^{1/2} \notag \\
& \leq 4b(p) ^{1/p} \supnorm{\testfn},\notag
\end{align}
where $b(p)$ is as in~\eqref{eq:b(p)}. 
Since $\update(\testfn)\equiv \bg(\predictive)(\testfn)$, this result combined with the hypothesis yields
\begin{align*}
\Exp\left[\vert\updateN(\testfn) - \update(\testfn)\vert^p\right]^{1/p} & \leq 4b(p)^{1/p}\frac{\supnorm{\testfn}}{N^{1/2} }+ \bar{C}_{p,n} \frac{\supnorm{\testfn}}{N^{1/2}}\\
& \leq (4b(p)^{1/p}+\bar{C}_{p,n})\frac{\supnorm{\testfn}}{N^{1/2}}.
\end{align*}
Thus, $C_{p,n}  = 4b(p)^{1/p}+\bar{C}_{p,n} $.
\end{proof}

\section{Proof of the WLLN in \cref{prop:wlln}}
\label{app:wlln}

The proof of the WLLN follows a similar strategy to that of \cref{prop:lp} and makes use of the Marcinkiewicz-Zygmund-type inequality in \cref{lemma:delmoral}.
We stress that the WLLN is obtained under weaker assumptions that those used for \cref{prop:lp} as \cref{lemma:delmoral} is only used to bound quantities for which \cref{ass:smc2} is not necessary.

\begin{proof}[Proof of \cref{prop:wlln}]
The proof follows an inductive argument similar to that for the $\lp$-inequality.
For the case $n=0$ we use the fact that $G_0^N\equiv G_0$ and the WLLN for self normalized importance sampling (e.g. \citet[Theorem 9.1.8]{cappe2005inference}) to show $\Psi_{G_0^N}(\eta_0^N)(\testfn)\pconverges\Psi_{G_0}(\eta_0)(\testfn)$.
Using this result, the decomposition
\begin{align*}
\hat{\eta}_0^N(\testfn) - \hat{\eta}_0(\testfn) = \hat{\eta}_0^N(\testfn) - \Psi_{G_0^N}(\eta_0^N)(\testfn) + \Psi_{G_0^N}(\eta_0^N)(\testfn) - \hat{\eta}_0(\testfn), 
\end{align*}
~\eqref{eq:lp0} and Markov's inequality we can conclude that $\hat{\eta}_0^N(\testfn) - \Psi_{G_0^N}(\eta_0^N)(\testfn)\pconverges 0$. The convergence of $\hat{\eta}_0^N(\testfn) - \hat{\eta}_0(\testfn) \pconverges 0$ then follows using the hypothesis since $\Psi_{G_0}(\eta_0)(\testfn) = \hat{\eta}_0(\testfn)$.
\cref{wlln:lemma1,wlln:lemma_weight,wlln:lemma3,wlln:lemma4} below show that if the results hold at time $n-1$ then it also holds at time $n$. We then conclude the proof using induction. 
\end{proof}
\subsection{Auxiliary results for the proof of \cref{prop:wlln}}

As for the proof of the $\lp$ inequality we combine well-known arguments with a result (in this case \cref{wlln:lemma_weight}) which controls the contribution of the weights $\weightN$. In particular, \cref{wlln:lemma1,wlln:lemma3,wlln:lemma4} can be found in \citet[Section 9.3]{cappe2005inference}.

\begin{lemma}[Mutation]
\label{wlln:lemma1}
Assume that, for any $\testfn\in \bounded$, $\hat{\eta}^N_{n-1}(\testfn) \pconverges\hat{\eta}_{n-1}(\testfn)$, then, after the mutation step, $\predictiveN(\testfn) \pconverges\predictive(\testfn)$
for any $\testfn\in \bounded$.
\end{lemma}
\begin{proof}
Using~\eqref{eq:lp_mutation} and Markov's inequality, it follows that $\predictiveN(\testfn) - \hat{\eta}_{n-1}^N M_n(\testfn)\pconverges0$.
We can use the hypothesis to get
\begin{align*}
\hat{\eta}_{n-1}^N M_n(\testfn)\pconverges \hat{\eta}_{n-1} M_n (\testfn).
\end{align*}
Combining the two results above and using Slutzky's lemma we obtain $\predictiveN(\testfn)\pconverges \predictive(\testfn)$.
\end{proof}

\begin{lemma}[Weight Comparison]
\label{wlln:lemma_weight}
Assume that, for any $\testfn\in \bounded$, $\Psi_{G_{n-1}^N}(\eta_{n-1}^N)(\testfn) \pconverges \Psi_{G_{n-1}}(\eta_{n-1})(\testfn)$, 
then $\predictiveN(\weightN\testfn) \pconverges\predictive(\weight\testfn)$, for any $\testfn\in \bounded$.
\end{lemma}
\begin{proof}
Consider the decomposition
\begin{align*}
\predictiveN(\weightN\testfn) - 
\predictive(\weight\testfn) & =\predictiveN(\weightN\testfn) - \Exp\left[\predictiveN(\weightN\testfn)\mid\mathcal{F}_{n-1}^N\right]+\Exp\left[\predictiveN(\weightN\testfn)|\mathcal{F}_{n-1}^N\right] -\predictive(\weight\testfn).
\end{align*}
Using the hypothesis,  \cref{prop:conditional_expe} and~\eqref{eq:conditional_expe2}, we have that $\Exp\left[\predictiveN(\weightN\testfn)\mid\mathcal{F}_{n-1}^N\right]\pconverges \predictive(\weight\testfn)$.
In addition, we have that, almost surely,
\begin{align*}
\Exp\left[\vert \predictiveN(\weightN\testfn)-\Exp\left[\predictiveN(\weightN\testfn)\mid\mathcal{F}_{n-1}^N\right]\vert^p\mid \mathcal{F}_{n-1}^N \right]^{1/p}\leq 4b(p)^{1/p}m_g\frac{\supnorm{\testfn}}{N^{1/2}},
\end{align*}
as shown in \cref{lp:lemma_weight}.
Hence, by Markov's inequality, $\predictiveN(\weightN\testfn)-\Exp\left[\predictiveN(\weightN\testfn)\mid\mathcal{F}_{n-1}^N\right]\pconverges 0$.
Therefore, $\predictiveN(\weightN \testfn)\pconverges \predictive(\weight \testfn)$.
\end{proof}

\begin{lemma}[Approximate reweighting]
\label{wlln:lemma3}
Assume that, for any $\testfn\in \bounded$, $\predictiveN(\weightN\testfn) \pconverges\predictive(\weight\testfn)$, then $\bgN(\predictiveN)(\testfn) \pconverges \bg(\predictive)(\testfn)$
for any $\testfn\in \bounded$.
\end{lemma}
\begin{proof}
Take the definition of $\bg$ and $\bgN$
\begin{align*}
\bgN(\predictiveN)(\testfn) - \bg(\predictive)(\testfn) =  \frac{\predictiveN( \weightN\testfn)}{\predictiveN( \weightN)} - \frac{\predictive( \weight\testfn)}{\predictive( \weight)}.
\end{align*}
A simple application of the continuous mapping theorem to the hypothesis gives
\begin{align*}
\frac{\predictiveN( \weightN\testfn)}{\predictiveN( \weightN)}  \pconverges  \frac{\predictive( \weight\testfn)}{\predictive( \weight)},
\end{align*}
and the result follows.
\end{proof}

\begin{lemma}[Multinomial resampling]
\label{wlln:lemma4}
Assume that, for any $\testfn\in \bounded$, $\bgN(\predictiveN)(\testfn) \pconverges\update(\testfn)$, 
then after the resampling step performed through multinomial resampling $\updateN(\testfn) \pconverges\update(\testfn)$
for any $\testfn\in \bounded$.
\end{lemma}
\begin{proof}
Consider the decomposition
\begin{align*}
\updateN(\testfn) - \update(\testfn) = \updateN(\testfn) - \bgN(\predictiveN)(\testfn)+\bgN(\predictiveN)(\testfn) - \update(\testfn).
\end{align*}
Using~\eqref{eq:lp_resampling} and Markov's inequality, we have that $\updateN(\testfn) - \bgN(\predictiveN)(\testfn)\pconverges0$.
Since $\update(\testfn)\equiv \bg(\predictive)(\testfn)$, this result combined with the hypothesis yields
\begin{align*}
\updateN(\testfn) \pconverges \update(\testfn).
\end{align*}
\end{proof}

\section{Proof of the bias estimates in \cref{prop:bias}}
\label{app:bias}

The proof of the bias estimate in \cref{prop:bias} uses an inductive approach similar to that of the proof of \cref{prop:lp}, and follows the approach of \cite{olsson2004bootstrap}.
\begin{proof}[Proof of \cref{prop:bias}]
At time $n=0$, the particles $(X_0^{i})_{i=1}^N$ are i.i.d samples from $\eta_0\equiv M_0$ which have been reweighted according to $G^N_0 \equiv G_0 \equiv \rmd K_0/\rmd M_0$. Thus, using standard results for self normalized importance sampling (e.g.~\citet[p.35]{liu2001monte}) we have
\begin{align*}
\left\lvert \Exp\left[ \Psi_{G^N_0}(\eta_0)(\testfn)\right] - \Psi_{G_0}(\eta_0)(\testfn)\right\rvert \leq \bar{C}_0\frac{\supnorm{\testfn}}{N}.
\end{align*}
To prove~\ref{prop:bias-upd}, consider $\mathcal{F}_0^N:=\sigma\left(X_0^{i}: i\in\lbrace 1,\ldots,N\rbrace\right)$ the $\sigma$-field generated by the weighted particles at time $n=0$. Then,
\begin{align*}
\Exp\left[ \hat{\eta}_0(\testfn)\right] =\Exp\left[ \Exp\left[ \hat{\eta}_0(\testfn)\mid \mathcal{F}_0^N\right]\right] = \Exp\left[ \Psi_{G^N_0}(\eta_0)(\testfn)\right],
\end{align*}
and the result follows from~\ref{prop:bias-pred}.
As for the proof of the $\lp$-inequality, assume the result holds at time $n-1$.

Consider~\ref{prop:bias-pred} and use the triangle inequality:
\begin{align}
\label{eq:statement2}
\left\lvert \Exp\left[ \bgN(\predictiveN)(\testfn) \right] - \bg(\predictive)(\testfn)\right\rvert  &= \left\lvert \Exp\left[\frac{\predictiveN(\weightN\testfn)}{\predictiveN(\weightN)}\right] - \bg(\predictive)(\testfn)\right\rvert  \\
&\leq \left\lvert \Exp\left[ \frac{\predictiveN(\weightN\testfn)}{\predictiveN(\weightN)}\right] - \frac{\Exp\left[ \predictiveN(\weightN\testfn)\right]}{\Exp\left[\predictiveN(\weightN)\right]}\right\rvert +  \left\lvert  \frac{\Exp\left[ \predictiveN(\weightN\testfn)\right]}{\Exp\left[\predictiveN(\weightN)\right]} -  \frac{\predictive(\weight\testfn)}{\predictive(\weight)} \right\rvert.\notag
\end{align}
To control the bias of the reweighting step, we bound~\eqref{eq:statement2}.
The second term is bounded applying the triangle inequality and using \cref{ass:smc2} as in \citet[Lemma 2.3]{olsson2004bootstrap}
\begin{align*}
\left\lvert  \frac{\Exp\left[ \predictiveN(\weightN\testfn)\right]}{\Exp\left[\predictiveN(\weightN)\right]} -  \frac{\predictive(\weight\testfn)}{\predictive(\weight)} \right\rvert &\leq \left\lvert  \frac{\Exp\left[ \predictiveN(\weightN\testfn)\right]}{\Exp\left[\predictiveN(\weightN)\right]} -  \frac{\predictive(\weight\testfn)}{\Exp\left[\predictiveN(\weightN)\right]} \right\rvert + \left\lvert  \frac{\predictive(\weight\testfn)}{\Exp\left[\predictiveN(\weightN)\right]} -  \frac{\predictive(\weight\testfn)}{\predictive(\weight)} \right\rvert\\
&=\left\lvert  \frac{\Exp\left[ \predictiveN(\weightN\testfn)\right] - \predictive(\weight\testfn)}{\Exp\left[\predictiveN(\weightN)\right]} \right\rvert + \left\lvert  \predictive(\weight\testfn)\frac{ \predictive(\weight) - \Exp\left[\predictiveN(\weightN)\right]}{\Exp\left[\predictiveN(\weightN)\right]\predictive(\weight)} \right\rvert\\
&\leq m_g\left\lvert  \Exp\left[ \predictiveN(\weightN\testfn)\right] - \predictive(\weight\testfn) \right\rvert + m_g\supnorm{\testfn} \left\lvert   \predictive(\weight) - \Exp\left[\predictiveN(\weightN)\right] \right\rvert.
\end{align*}
Then, using \cref{prop:conditional_expe} and~\eqref{eq:conditional_expe2}, we have
\begin{align*}
\left\vert \Exp\left[ \predictiveN(\weightN\testfn)\right] - \predictive(\weight\testfn)\right\vert &=\left\vert\Exp\left[ \Exp\left[ \predictiveN(\weightN\testfn)\mid \mathcal{F}_{n-1}^N\right]\right] - \predictive(\weight\testfn)\right\vert\\
& =\left\vert \Exp\left[\Psi_{G^N_{n-1}}(\eta^N_{n-1})(K_n(U_n\testfn))\right] - \Psi_{G_{n-1}}(\eta_{n-1})(K_n(U_n\testfn))\right\vert \\
&\leq \frac{\bar{C}_{n-1}\supnorm{U_n}\supnorm{\testfn}}{N}
\end{align*}
where the last inequality follows from the inductive hypothesis.
Hence, 
\begin{align}
\label{eq:statement3}
\left\lvert  \frac{\Exp\left[ \predictiveN(\weightN\testfn)\right]}{\Exp\left[\predictiveN(\weightN)\right]} -  \frac{\predictive(\weight\testfn)}{\predictive(\weight)} \right\rvert \leq 2m_g\frac{\bar{C}_{n-1}\supnorm{U_n}\supnorm{\testfn}}{N}.
\end{align}

For the first term in~\eqref{eq:statement2}, consider a two-dimensional Taylor expansion of the function $(u, v) \mapsto u/v$ around $(u_0, v_0)$ with a second order remainder of Lagrange form
\begin{align*}
\frac{u}{v} = \frac{u_0}{v_0} + \frac{1}{v_0}(u - u_0) - \frac{u_0}{v_0^2}(v -v_0) + \frac{\theta_u}{\theta_v^3}(v -v_0) ^2 - \frac{1}{\theta_v^2}(v -v_0) (u - u_0)
\end{align*}
where $(\theta_u, \theta_v)$ is a point on the line segment between $(u, v)$ and $(u_0, v_0)$.
Applying this Taylor expansion to $\predictiveN(\weightN\testfn)/\predictiveN(\weightN)$ around the point $\left( \Exp\left[ \predictiveN(\weightN\testfn)\right], \Exp\left[\predictiveN(\weightN)\right]\right)$, as in \citet[Lemma 2.4]{olsson2004bootstrap}, gives
\begin{align}
\label{eq:biaslagrange}
\frac{\predictiveN(\weightN\testfn)}{\predictiveN(\weightN)} =& \frac{\Exp\left[ \predictiveN(\weightN\testfn)\right]}{\Exp\left[\predictiveN(\weightN)\right]} + \frac{1}{\Exp\left[\predictiveN(\weightN)\right]}\left(\predictiveN(\weightN\testfn) - \Exp\left[ \predictiveN(\weightN\testfn)\right] \right)\\
&- \frac{ \Exp\left[ \predictiveN(\weightN\testfn)\right]}{\Exp\left[\predictiveN(\weightN)\right]^2}\left(\predictiveN(\weightN) - \Exp\left[ \predictiveN(\weightN)\right] \right) + R_n^N(\theta_u, \theta_v) \notag
\end{align}
where the remainder is a function of $(\theta_u, \theta_v)$, a point on the line segment between $(\predictiveN(\weightN\testfn), \predictiveN(\weightN))$ and $\left(\Exp\left[\predictiveN(\weightN\testfn)\right], \Exp\left[\predictiveN(\weightN)\right]\right)$
\begin{align*}
R_n^N(\theta_u, \theta_v) &:= \frac{\theta_u}{\theta_v^3} \left( \predictiveN(\weightN) - \Exp\left[ \predictiveN(\weightN)\right]\right)^2\\
& - \frac{1}{\theta_v^2} \left( \predictiveN(\weightN\testfn) - \Exp\left[ \predictiveN(\weightN\testfn)\right]\right) \left( \predictiveN(\weightN) - \Exp\left[ \predictiveN(\weightN)\right]\right).
\end{align*}
Taking the expectation of both sides of~\eqref{eq:biaslagrange} yields
\begin{align*}
\Exp\left[ \frac{\predictiveN(\weightN\testfn)}{\predictiveN(\weightN)}\right] = \frac{\Exp\left[ \predictiveN(\weightN\testfn)\right]}{\Exp\left[\predictiveN(\weightN)\right]} + \Exp\left[R_n^N(\theta_u, \theta_v)\right].
\end{align*}
Since one of the extremal points of the segment is random, $(\theta_u, \theta_v)$ is random too; because we have $0<m_g^{-1}\leq\weightN\leq m_g$ it follows that
$\predictiveN(\weightN\testfn)\leq m_g\supnorm{\testfn}$, $ \predictiveN(\weightN)\geq m_g^{-1}>0$, so that 
$\vert \theta_u \vert \leq m_g\supnorm{\testfn}$, $\vert \theta_v^{-1} \vert \geq m_g^{-1}$ almost surely. Therefore,
\begin{align*}
\vert \Exp\left[R_n^N(\theta_u, \theta_v)\right] \vert &\leq m_g^4\supnorm{\testfn}\Exp \left[ \left\lvert \predictiveN(\weightN) - \Exp\left[ \predictiveN(\weightN)\right]\right\vert^2\right]\\
&+ m_g^2\left\lvert\Exp\left[ \left\lvert \predictiveN(\weightN\testfn) - \Exp\left[ \predictiveN(\weightN\testfn)\right]\right\rvert\left\lvert \predictiveN(\weightN) - \Exp\left[ \predictiveN(\weightN)\right]\right\rvert\right]\right\rvert.
\end{align*}
By \cref{lemma:bias}, below, with $l=m=1$ we then have
\begin{align}
\label{eq:bias2}
\vert \Exp\left[R_n^N(\theta_u, \theta_v)\right] \vert &\leq  m_g^4 A_{2, n}^2\frac{\supnorm{\testfn}}{N} + m_g^2 A_{2, n}^2\frac{\supnorm{\testfn}}{N}.
\end{align}
Combining~\eqref{eq:statement2} with~\eqref{eq:statement3} and~\eqref{eq:bias2} we obtain the result.

Having obtained a bias estimate in~\ref{prop:bias-pred} we can then obtain~\ref{prop:bias-upd} by conditioning on $\sfresampling$, the $\sigma$-field generated by the particle system up to (and including) the resampling step at time $n$,
\begin{align*}
\left\lvert \Exp\left[ \updateN(\testfn)\right] - \update(\testfn)\right\rvert  & = \left\lvert \Exp\left[\Exp\left[ \updateN(\testfn) \mid \sfresampling\right]\right] - \update(\testfn)\right\rvert  = \left\lvert \Exp\left[\bgN(\predictiveN)(\testfn)\right] - \update(\testfn)\right\rvert\\
&\leq \frac{\bar{C}_n\supnorm{\testfn}}{N},
\end{align*}
since $\update(\testfn) = \bg(\predictive)(\testfn)$.
\end{proof}
\subsection{Auxiliary results for the proof of the bias estimate}
The following auxiliary result is a direct consequence of the $\lp$-inequality in \cref{prop:lp} and \cref{lp:lemma_weight} and follows by applying Jensen's inequality and the Cauchy-Schwarz inequality as in the proof of \citet[Lemma 2.2]{olsson2004bootstrap}. 
\begin{lemma}
\label{lemma:bias}
Under \cref{ass:smc0,ass:weak,ass:smc2}, for every $n\geq 0$ and every $p\geq 1$ there exists a finite constant $A_{p, n}$ such that, for every measurable bounded function $\testfn\in \bounded$
\begin{equation*}
\Exp\left[ \left \lvert \predictiveN(\weightN\testfn) - \Exp\left[ \predictiveN(\weightN\testfn) \right]\right\rvert^p\right]^{1/p} \leq A_{p, n}\frac{\supnorm{\testfn}}{\sqrt{N}}.
\end{equation*}
Additionally, for all $\testfn, \psi \in \bounded$ and for integers $0\leq l, m < \infty$ 
\begin{align*}
&\left\lvert \Exp\left[ \left\lvert\predictiveN(\weightN\testfn) - \Exp\left[ \predictiveN(\weightN\testfn)\right]\right\rvert^l \left\lvert \predictiveN(\weightN\psi)- \Exp\left[ \predictiveN(\weightN\psi) \right]\right\rvert^m\right] \right\rvert \\
&\qquad\qquad \leq \supnorm{\testfn}^l\supnorm{\psi}^m\frac{A_{p, n}^{l+m}}{N^{(l + m)/2}}.
\end{align*}
\end{lemma}
\begin{proof}
To prove the first assertion, apply Minkowski's inequality
\begin{align*}
\Exp\left[ \left \lvert \predictiveN(\weightN\testfn) - \Exp\left[ \predictiveN(\weightN\testfn) \right]\right\rvert^p\right]^{1/p} &\leq \Exp\left[ \left \lvert\predictiveN(\weightN\testfn) - \predictive(\weight\testfn)\right\rvert^p\right]^{1/p} \\
&+ \Exp\left[ \left \lvert \Exp\left[\predictive(\weight\testfn) -  \predictiveN(\weightN\testfn) \right]\right\rvert^p\right]^{1/p} \\
&\leq 2\Exp\left[ \left \lvert\predictiveN(\weightN\testfn) - \predictive(\weight\testfn)\right\rvert^p\right]^{1/p},
\end{align*}
where the second inequality follows from  Jensen's inequality applied to the second expectation in line 1.
We can apply \cref{prop:lp} and \cref{lp:lemma_weight} and obtain the first result.
The first result and the Cauchy-Schwarz inequality give
\begin{align*}
&\left\lvert \Exp\left[ \left\lvert\predictiveN(\weightN\testfn) - \Exp\left[ \predictiveN(\weightN\testfn)\right]\right\rvert^l \left\lvert \predictiveN(\weightN\psi)- \Exp\left[ \predictiveN(\weightN\psi) \right]\right\rvert^m\right] \right\rvert \\
&\qquad\leq \Exp\left[ \left\lvert\predictiveN(\weightN\testfn) - \Exp\left[ \predictiveN(\weightN\testfn)\right]\right\rvert^{2l}\right]^{1/2} \Exp\left[ \left\lvert\predictiveN(\weightN\psi) - \Exp\left[ \predictiveN(\weightN\psi)\right]\right\rvert^{2m}\right]^{1/2}\\
&\qquad\leq \supnorm{\testfn}^l\supnorm{\psi}^m\frac{A_{2l, n}^{l}A_{2m, n}^{m}}{N^{(l + m)/2}}.
\end{align*}
\end{proof}

\section{Proof of the CLT in \cref{prop:clt}}
\label{app:clt}

The proof of the CLT in \cref{prop:clt} follows that of \cite{chopin2004central} (see also \cite{douc2007limit}) and, similarly to the proof of the $\lp$-inequalities, consists of three Lemmata which establish convergence in distribution to a Normal random variable for the mutation step, the approximate reweighting step and the multinomial resampling step and one result controlling the behaviour of the approximate weights.

\begin{proof}[Proof of \cref{prop:clt}]
At time $n=0$, we have a self-normalized importance sampling estimator with proposal $M_0$ and importance weights $G_0^N\equiv G_0$, thus (see, e.g., \citet[Section 8.5]{chopin2020})
\begin{align*}
\sqrt{N}\left[ \Psi_{G_0^N}(\eta_0^N)(\testfn) - \Psi_{G_0}(\eta_0)(\testfn)\right] \dconverges \N\left( 0, \bar{\V}_0(\testfn)\right),
\end{align*}
where $\bar{\V}_0(\testfn) = \eta_0(G_0)^{-2}\var_{M_0}\left(G_0(\testfn-\hat{\eta}_0(\testfn))\right)$, which is finite as $\testfn \in \bounded$.
To obtain the result after resampling, we apply \cref{clt:lemma4} below and obtain
\begin{align*}
\sqrt{N}\left[ \hat{\eta}^N_0(\testfn) - \hat{\eta}_0(\testfn)\right] \dconverges \N\left( 0, \V_0(\testfn)\right),
\end{align*}
with $\V_0(\testfn) = \bar{\V}_0(\testfn)+\var_{\hat{\eta}_0}(\testfn)$, which is finite as $\testfn \in \bounded$.

Then, assume that the result holds at time $n-1$.
\cref{clt:lemma1} gives for every $\testfn \in \bounded$
\begin{equation*}
\sqrt{N}\left[\predictiveN(\testfn) - \predictive(\testfn)\right] \dconverges \N\left( 0, \widetilde{\V}_n(\testfn)\right)
\end{equation*}
where
\begin{equation*}
\widetilde{\V}_n(\testfn) = \hat{\eta}_{n-1}\left( \var_{M_n}(\testfn)\right) + \V_{n-1}(M_n\testfn)
\end{equation*}
is finite because $\hat{\eta}_{n-1}\left( \var_{M_n}(\testfn)\right)  \leq \supnorm{\testfn}^2$.
\cref{clt:weight_comparison} shows that
\begin{equation*}
\sqrt{N}\left[\predictiveN(\weightN\testfn) - \predictive(\weight\testfn)\right] \dconverges \N\left( 0, \widehat{\V}_{n}(\testfn)\right)
\end{equation*}
for any $\testfn\in \bounded$, where $\widehat{\V}_n(\testfn) = \var_{\predictive}(\weight\testfn) + \bar{\V}_{n-1}(K_n(U_n\testfn))$ is finite since the functions $\weight\testfn$ are bounded.
Then, \cref{clt:lemma3} give for every $\testfn \in \bounded$
\begin{equation*}
\sqrt{N}\left[\bgN(\predictiveN)(\testfn) - \update(\testfn)\right] \dconverges \N\left( 0, \bar{\V}_n(\testfn)\right)
\end{equation*}
where
\begin{equation*}
\bar{\V}_n(\testfn) = \frac{1}{\predictive(G_n)^2}\widehat{\V}_n\left(\testfn - \bg(\predictive)(\testfn)\right).
\end{equation*}
Finally, by applying \cref{clt:lemma4} we conclude that for every $\testfn \in \bounded$
\begin{equation*}
\sqrt{N}\left[ \updateN(\testfn) - \update(\testfn)\right] \dconverges \N\left( 0, \V_n(\testfn)\right)
\end{equation*}
where
\begin{equation*}
\V_n(\testfn) = \var_{\update}(\testfn) + \bar{\V}_n(\testfn)
\end{equation*}
is finite as $\var_{\update}(\testfn) \leq \supnorm{\testfn}^2$.
Hence, the result holds for all $n\in \mathbb{N}$ by induction.
\end{proof}

\subsection{Auxiliary results for the proof of \cref{prop:clt}}
We collect here four auxiliary results for the proof of the CLT in \cref{prop:clt}. \cref{clt:lemma1,clt:lemma3,clt:lemma4} are well known in the SMC literature and can be found in \cite{chopin2004central} and are combined in our case with \cref{clt:weight_comparison} which controls the influence of the approximation in the importance weights.

The first results concerns the mutation step, and is given in \citet[Lemma A.1]{chopin2004central}.
\begin{lemma}[Mutation]
\label{clt:lemma1}
Assume that for any $\testfn\in \bounded$ and for some finite $\V_{n-1}(\testfn)$
\begin{equation}
\label{eq:p(t-1)hpCLT}
\sqrt{N}\left[\hat{\eta}^N_{n-1}(\testfn) - \hat{\eta}_{n-1}(\testfn)\right] \dconverges \N\left( 0, \V_{n-1}(\testfn)\right)
\end{equation}
then, after the mutation step
\begin{equation*}
\sqrt{N}\left[\predictiveN(\testfn) - \predictive(\testfn)\right] \dconverges \N\left( 0, \widetilde{\V}_{n}(\testfn)\right)
\end{equation*}
for any $\testfn\in \bounded$, where $\widetilde{\V}_{n}(\testfn) = \hat{\eta}_{n-1}\left(\var_{M_n}( \testfn)\right) + \V_{n-1}(M_n\testfn)$.
\end{lemma}
\begin{proof}
As for the proof of convergence in mean of order $p$, consider the following decomposition
\begin{align*}
T_N :=\sqrt{N}\left[\predictiveN(\testfn) - \predictive(\testfn)\right] = \sqrt{N}\left[\predictiveN(\testfn) - \hat{\eta}_{n-1}^NM_n(\testfn)\right] + \sqrt{N}\left[\hat{\eta}_{n-1}^NM_n(\testfn) - \predictive(\testfn)\right] = U_N + R_N
\end{align*}
where $U_N:=  \sqrt{N}\left[\predictiveN(\testfn) - \hat{\eta}_{n-1}^NM_n(\testfn)\right] $ and $R_N:= \sqrt{N}\left[\hat{\eta}_{n-1}^NM_n(\testfn) - \predictive(\testfn)\right]$.

Consider for $i=1,\ldots, N$ the functions $\Delta_n^i : E\mapsto \mathbb{R}$
\begin{equation*}
\Delta_n^i(x) := \frac{1}{\sqrt{N}}\left[ \testfn(x) -  \Exp\left[\testfn(X_{n}^{i})\mid \sfmutation \right]\right],
\end{equation*}
so that
\begin{align*}
U_N & =\sqrt{N} \left[ \predictiveN(\testfn) - \hat{\eta}_{n-1}^NM_n(\testfn)\right]\\
& = \frac{1}{\sqrt{N}} \sum_{i=1}^N \testfn(X_{n}^{i}) -  \Exp\left[\testfn(X_{n}^{i})\mid \mathcal{G}^N_{n-1} \right]\\
& = \sum_{i=1}^N \Delta_n^i(X^i_n).
\end{align*}
Let us denote $\mathcal{H}_{n, i}^N:=\sfmutation\vee\sigma \left(\left(X_{n}^{j}\right)_{j=1}^i\right)$. 
Conditional on $\sfmutation$, the $\Delta_n^i(X^i_n),\ i=1,\ldots, N$ are independent; in addition, due to the boundedness of $\testfn$,  $(\Delta_n^i(X_n^i), \mathcal{H}_{n, i}^N)_{i=1}^N$ is a square integrable martingale difference sequence which satisfies the Lindeberg condition, and for which
\begin{align*}
\sum_{i=1}^N\var\left( \Delta_n^i(X_n^i ) \mid \mathcal{H}^N_{n, i-1}\right) & = \sum_{i=1}^N\Exp\left[ \Delta_n^i(X_n^i )^2 \mid \mathcal{G}^N_{n-1}\right]\\
& = \frac{1}{N}\sum_{i=1}^N \Exp\left[\left(\testfn(X_{n}^{i}  ) -  \Exp\left[\testfn(X_{n}^{i}  )\mid \mathcal{G}^N_{n-1} \right] \right)^2\mid \mathcal{G}^N_{n-1} \right]\\
& = \frac{1}{N}\sum_{i=1}^N \left(\Exp\left[\testfn(X_{n}^{i} )^2 \mid \mathcal{G}^N_{n-1} \right] -  \Exp\left[\testfn(X_{n}^{i} )\mid \mathcal{G}^N_{n-1} \right] ^2 \right)\\
& = \frac{1}{N}\sum_{i=1}^N \left( M_n\testfn^2(\widetilde{X}_{n-1}^i) - \left(M_n\testfn(\widetilde{X}_{n-1}^i)\right)^2 \right)\\
& = \hat{\eta}_{n-1}^N\left(\var_{M_n}\left( \testfn\right)\right) \rightarrow \hat{\eta}_{n-1}\left(\var_{M_n}\left( \testfn \right)\right)
\end{align*}
in probability as $N\rightarrow\infty$ (\cref{prop:wlln}).
Thus the conditions of a suitable CLT for triangular arrays, such as 
\citet[Theorem A.3]{douc2007limit}, are satisfied and we have that
$ \Exp\left[\exp\left( iu U_N\right) \mid \mathcal{G}_{n-1}^N \right] \pconverges \exp\left( -\hat{\eta}_{n-1}\left(\var_{M_n}\left( \testfn \right)\right) u^2/2 \right)$.
Straightforward application of the hypothesis leads to
\begin{align*}
R_N =\sqrt{N} \left[ \hat{\eta}_{n-1}^N M_n(\testfn) - \predictive(\testfn) \right] = \sqrt{N} \left[ \hat{\eta}_{n-1}^N M_n(\testfn) - \hat{\eta}_{n-1} M_n(\testfn) \right] \dconverges \N\left( 0, \V_{n-1}(M_n\testfn)\right).
\end{align*}
The characteristic function of $T_N$ is
\begin{equation*}
\Phi_{T_N}(u) = \Exp\left[\exp\left( iuT_N\right) \right] = \Exp\left[\exp\left( iu R_N\right) \Exp\left[\exp\left( iu U_N\right) \mid \mathcal{G}_{n-1}^N \right]\right].
\end{equation*}
Thus, by the continuous mapping Theorem and  Slutzky's Lemma,
\begin{equation*}
\exp\left( iu R_N\right)\Exp\left[\exp\left( iu U_N\right)\mid \mathcal{G}_{n-1}^N \right] \dconverges \exp\left( -\hat{\eta}_{n-1}\left(\var_{M_n}\left( \testfn \right)\right) u^2/2 + iu Z\right)
\end{equation*}
where $Z\sim \N(0, \V_{n-1}(M_n\testfn))$.
Hence, by the dominated convergence theorem for any real $u$,
\begin{align*}
\Phi_{T_N}(u) = \Exp\left[\exp\left( iu R_N\right)\Exp\left[\exp\left( iu U_N \right) \mid \mathcal{G}_{n-1}^N \right]\right] \rightarrow \exp\left( -(\hat{\eta}_{n-1}\left(\var_{M_n}\left( \testfn \right)\right) +\V_{n-1}(M_n\testfn)) u^2/2\right)
\end{align*}
and L\'evy's continuity Theorem (e.g. \citet[Theorem 1, page 322]{shiryaev1996probability}) gives that $T_N$ follows a Normal distribution with mean 0 and variance $\widetilde{\V}_n(\testfn) = \hat{\eta}_{n-1}\left(\var_{M_n}( \testfn)\right) + \V_{n-1}(M_n\testfn)$.
\end{proof}

We next show that the approximate weights also give rise to a random variable converging to a Normal distribution:
\begin{lemma}[Weight Comparison]
\label{clt:weight_comparison}
Assume that for any $\testfn\in \bounded$ and for some finite $\V_{n-1}(\testfn)$
\begin{equation*}
\sqrt{N}\left[\Psi_{G_{n-1}^N}(\eta^N_{n-1})(\testfn) - \Psi_{G_{n-1}}(\eta_{n-1})(\testfn)\right] \dconverges \N\left( 0, \bar{\V}_{n-1}(\testfn)\right)
\end{equation*}
then, 
\begin{equation*}
\sqrt{N}\left[\predictiveN(\weightN\testfn) - \predictive(\weight\testfn)\right] \dconverges \N\left( 0, \widehat{\V}_{n}(\testfn)\right)
\end{equation*}
for any $\testfn\in \bounded$, where $\widehat{\V}_n(\testfn) = \var_{\predictive}(\weight\testfn) + \bar{\V}_{n-1}(K_n(U_n\testfn))$.
\end{lemma}
\begin{proof}
Consider the following decomposition
\begin{align*}
T_N &:=\sqrt{N}\left[\predictiveN(\weightN\testfn) - \predictive(\weight\testfn)\right] \\
&= \sqrt{N}\left[\predictiveN(\weightN\testfn) - \Exp\left[\predictiveN(\weightN\testfn)\mid \mathcal{F}_{n-1}^N\right]\right] + \sqrt{N}\left[\Exp\left[\predictiveN(\weightN\testfn)\mid \mathcal{F}_{n-1}^N\right] - \predictive(\weight\testfn)\right] \\
&= U_N + R_N
\end{align*}
where $U_N:=  \sqrt{N}\left[\predictiveN(\weightN\testfn) - \Exp\left[\predictiveN(\weightN\testfn)\mid \mathcal{F}_{n-1}^N\right]\right] $ and $R_N:= \sqrt{N}\left[\Exp\left[\predictiveN(\weightN\testfn)\mid \mathcal{F}_{n-1}^N\right] - \predictive(\weight\testfn)\right]$.

To control $U_N$ let us consider the sequence of functions $\Delta_n^i : E \mapsto \mathbb{R}$ for $i=1,\ldots, N$ 
\begin{equation*}
\Delta_n^i(x) :=  \frac{1}{\sqrt{N}}\left[\weightN(x)\testfn(x) -  \Exp\left[\weightN(X_n^i)\testfn(X_{n}^i) \mid \mathcal{F}_{n-1}^N\right]\right].
\end{equation*}
so that
\begin{align*}
U_N & =\sqrt{N}\left[\predictiveN(\weightN\testfn) - \Exp\left[\predictiveN(\weightN\testfn)\mid \mathcal{F}_{n-1}^N\right]\right]\\
& = \frac{1}{\sqrt{N}} \sum_{i=1}^N G_n^N(X_n^i)\testfn(X_n^i) -  \Exp\left[G_n^N(X_n^i)\testfn(X_{n}^i) \mid \mathcal{F}_{n-1}^N\right]\\
& = \sum_{i=1}^N \Delta_n^i(X^i_n).
\end{align*}
Let us denote $\mathcal{H}_{n, i}^N:= \mathcal{F}_{n-1}^N\vee \sigma\left(\left(X_{n}^{j}\right)_{j=1}^i\right)$. 
Conditional on $\mathcal{F}_{n-1}^N$, the $\Delta_n^i(X^i_n),\ i=1,\ldots, N$ are independent; in addition, due to the boundedness of $\testfn$ and $\weightN$,  $(\Delta_n^i(X_n^i), \mathcal{H}_{n, i}^N)_{i=1}^N$ is a square integrable martingale sequence which satisfies the Lindeberg condition, and for which
\begin{align*}
\sum_{i=1}^N\var\left( \Delta_n^i(X_n^i ) \mid \mathcal{H}^N_{n, i-1}\right) & = \sum_{i=1}^N\Exp\left[ \Delta_n^i(X_n^i )^2 \mid \mathcal{F}^N_{n-1}\right]\\
& = \frac{1}{N}\sum_{i=1}^N \Exp\left[\left(\weightN(X_n^i)\testfn(X_{n}^{i}  ) -  \Exp\left[\weightN(X_n^i)\testfn(X_{n}^{i}  )\mid \mathcal{F}^N_{n-1} \right] \right)^2\mid \mathcal{F}^N_{n-1} \right]\\
& = \frac{1}{N}\sum_{i=1}^N \left(\Exp\left[\left(\weightN(X_n^i)\testfn(X_{n}^{i} )\right)^2 \mid \mathcal{F}^N_{n-1} \right] -  \Exp\left[\weightN(X_n^i)\testfn(X_{n}^{i} )\mid \mathcal{F}^N_{n-1} \right] ^2 \right)\\
& = \sum_{j=1}^N W_{n-1}^jM_n\left((\weightN\testfn)^2\right)(X_{n-1}^j)  -  \left( \sum_{j=1}^N W_{n-1}^jM_n\left(\weightN\testfn\right)(X_{n-1}^j)\right) ^2 \\
&= \var_{\nu^N_{n-1}}\left(\weightN\testfn\right)
\end{align*}
where the last equality follows using~\eqref{eq:conditional_expe1} with $\nu^N_{n-1}:= \Psi_{G_{n-1}^N}(\eta_{n-1}^N)M_n$. 
Using \cref{clt:variance_convergence}, we have
\begin{align*}
\var_{\nu^N_{n-1}}\left(\weightN\testfn\right) \pconverges \var_{\eta_n}\left(\weight\testfn\right).
\end{align*}
Thus, we can apply a CLT for triangular arrays (see 
\citet[Theorem A.3]{douc2007limit}) to obtain $\Exp\left[\exp\left( iu U_N\right) \mid \mathcal{F}_{n-1}^N \right]\pconverges \exp\left( -\var_{\eta_n}(\testfn) u^2/2 \right)$.

As in the proof of \cref{clt:lemma1}, straightforward application of the hypothesis leads to
\begin{align*}
R_N &=\sqrt{N}\left[\Exp\left[\predictiveN(\weightN\testfn)\mid \mathcal{F}_{n-1}^N\right] - \predictive(\weight\testfn)\right]\\
& = \sqrt{N} \left[ \Psi_{G_{n-1}^N}(\eta_{n-1}^N)(K_n(U_n\testfn)) - \Psi_{G_{n-1}}(\eta_{n-1})(K_n(U_n\testfn)) \right] \dconverges \N\left( 0, \bar{\V}_{n-1}(K_n(U_n\testfn))\right).
\end{align*}
The characteristic function of $T_N$ is 
\begin{equation*}
\Phi_{T_N}(u) = \Exp\left[\exp\left( iuT_N\right) \right] = \Exp\left[\exp\left( iu R_N\right) \Exp\left[\exp\left( iu U_N\right) \mid \mathcal{F}_{n-1}^N \right]\right].
\end{equation*}
By the continuous mapping Theorem and  Slutzky's Lemma,
\begin{align*}
\exp\left( iu R_N\right)\Exp\left[\exp\left( iu U_N\right)\mid \mathcal{F}_{n-1}^N \right] \dconverges \exp\left( -\var_{\predictive}(\weight\testfn) u^2/2 + iu Z\right)
\end{align*}
where $Z\sim \N(0, \bar{\V}_{n-1}(K_n(U_n\testfn)))$.
Then, by the same argument used in \cref{clt:lemma1}, $T_N$ follows a Normal distribution with mean 0 and variance $\widehat{\V}_n(\testfn) = \var_{\predictive}(\weight\testfn) + \bar{\V}_{n-1}(K_n(U_n\testfn)) $.
\end{proof}

Combining \cref{clt:weight_comparison} above with the argument in \citet[Lemma A.2]{chopin2004central} we obtain an equivalent result for the approximate reweighting step:
\begin{lemma}[Approximate reweighting]
\label{clt:lemma3}
Assume that for any $\testfn\in \bounded$ and some finite $\widehat{\V}_n(\testfn)$
\begin{equation*}
\sqrt{N}\left[\predictiveN( \weightN\testfn) - \predictive( \weight\testfn)\right] \dconverges \N\left( 0,\widehat{\V}_n(\testfn)\right)
\end{equation*}
then
\begin{equation*}
\sqrt{N}\left[\bgN(\predictiveN)(\testfn) - \bg(\predictive)(\testfn)\right] \dconverges \N\left( 0, \bar{\V}_{n}(\testfn)\right)
\end{equation*}
for any $\testfn\in \bounded$, where $\bar{\V}_{n}(\testfn) = \frac{1}{\predictive(U_n)^2}\widehat{\V}_n\left(\testfn - \bg(\predictive)(\testfn)\right)$.
\end{lemma}

\begin{proof}
Let $\bar{\testfn} := \testfn - \bg(\predictive)(\testfn)$ and consider the vector function
\begin{align*}
\psi(x) = \begin{pmatrix}
\psi_1(x)\\
\psi_2(x)
\end{pmatrix} = \begin{pmatrix}
\bar{\testfn}(x)\\
1
\end{pmatrix}
\end{align*}
with bounded components $\psi_1$ and $\psi_2$.

Combining the hypothesis with the Cram\'er-Wold Theorem (e.g. \citet[Theorem 29.4, page 383]{billingsley1995measure}) yields
\begin{align*}
\sqrt{N}\begin{bmatrix}
\predictiveN( \weightN\bar{\testfn}) - \predictive( \weight\bar{\testfn})\\
\predictiveN( \weightN) - \predictive( \weight)\\
\end{bmatrix}\dconverges \N\left( 0, \Sigma_n(\psi)\right), \qquad\qquad \Sigma_n(\psi) := \begin{pmatrix}
\widehat{\V}_n(\psi_1) & \textrm{cov}_n(\psi_1, \psi_2)\\
\textrm{cov}_n(\psi_1, \psi_2)& \widehat{\V}_n(\psi_2)\\
\end{pmatrix},
\end{align*}
where $\textrm{cov}_n(\psi_1, \psi_2)$ denotes the covariance operator induced by $\widehat{\V}_n$.

Applying the $\delta$-method with function $s(u, v)= u/v$ and observing that $\predictive( \weight\bar{\testfn}) = 0$ and $\predictive(\weight)>0$, gives
\begin{align*}
\sqrt{N}\left[ \frac{\predictiveN( \weightN\bar{\testfn})}{\predictiveN( \weightN)}- \frac{0}{\predictive( \weight)}\right] &= \sqrt{N}\left[ \bgN(\predictiveN)(\testfn) -\bg(\predictive)(\testfn) \right]\\
 & \dconverges \N\left( 0, \nabla^T s(0, \predictive(\weight))\Sigma_n\left(\psi\right)\nabla s(0, \predictive(\weight))\right).
\end{align*}
The gradient of $s$ evaluated at $(0, \predictive(\weight))$ is
\begin{align*}
\nabla^T s(u, v)\mid_{ (u, v)= (0, \predictive(\weight))} = \left(\frac{1}{v}, -\frac{u}{v^2}\right)\big|_{ (u, v)= (0, \predictive(\weight))} = \left( \frac{1}{\predictive(\weight)}, 0\right),
\end{align*}
hence
\begin{align*}
    \sqrt{N}\left[ \bgN(\predictiveN)(\testfn) -\bg(\predictive)(\testfn) \right]\ &\dconverges
 \N\left( 0, \frac{1}{\predictive(G_n)^2}\widehat{\V}_n\left(\psi_1\right)\right)\\
 & \dconverges \N\left( 0, \frac{1}{\predictive(G_n)^2}\widehat{\V}_n\left(\testfn - \bg(\predictive)(\testfn)\right)\right),
\end{align*}
giving $\bar{\V}_{n}(\testfn) = \frac{1}{\predictive(G_n)^2}\widehat{\V}_n\left(\testfn - \bg(\predictive)(\testfn)\right)$.
\end{proof}

Finally, following \citet[Lemma A.3]{chopin2004central} we establish convergence for the multinomial resampling step.
As shown in \citet[Theorem 7]{Gerber2019} the asymptotic variance in the multinomial resampling case provides an upper bound to that obtained with more sophisticated resampling schemes.

\begin{lemma}[Multinomial resampling]
\label{clt:lemma4}
Assume that for any $\testfn\in \bounded$ and for some finite $\bar{\V}_{n}(\testfn)$
\begin{equation}
\label{eq:l4hp}
\sqrt{N}\left[\bgN(\predictiveN)(\testfn) - \update(\testfn)\right] = \sqrt{N}\left[\bgN(\predictiveN)(\testfn) - \bg(\predictive)(\testfn)\right] \dconverges \N\left( 0, \bar{\V}_{n}(\testfn)\right)
\end{equation}
then after the multinomial resampling step
\begin{equation*}
\sqrt{N}\left[\updateN(\testfn) - \hat{\eta}_n(\testfn)\right] \dconverges \N(0, \V_n(\testfn))
\end{equation*}
for any $\testfn\in \bounded$, where $\V_n(\testfn)=\var_{\update}(\testfn) + \bar{\V}_{n}(\testfn)$.
\end{lemma}
\begin{proof}
As for the proof of convergence in mean of order $p$, consider the following decomposition
\begin{align*}
T_N :=\sqrt{N}\left[\updateN(\testfn) - \update(\testfn)\right] = \sqrt{N}\left[\updateN(\testfn) - \bgN(\predictiveN)(\testfn)\right] + \sqrt{N}\left[\bgN(\predictiveN)(\testfn) - \update(\testfn)\right] = U_N + R_N
\end{align*}
where $U_N:=  \sqrt{N}\left[\updateN(\testfn) - \bgN(\predictiveN)(\testfn)\right] $ and $R_N:= \sqrt{N}\left[\bgN(\predictiveN)(\testfn) - \update(\testfn)\right]$.

Consider for $i=1,\ldots, N$  the functions $\Delta_n^i(\widetilde{X}_n^i ) : E \mapsto \mathbb{R}$ defined for the resampled particles
\begin{equation*}
\Delta_n^i(\widetilde{X}_n^i ) :=  \frac{1}{\sqrt{N}}\left[\testfn(\widetilde{X}_{n}^{i}) -  \Exp\left[\testfn(\widetilde{X}_{n}^i  ) \mid \mathcal{F}^N_{n}\right] \right] = \frac{1}{\sqrt{N}}\left[\testfn(\widetilde{X}_{n}^{i}) -  N W_n^i\testfn(X_{n}^i)\right],
\end{equation*}
so that
\begin{align*}
U_N &= \sqrt{N}\left[ \updateN(\testfn) - \bgN(\predictiveN)(\testfn)\right] \\
&= \frac{1}{\sqrt{N}}\left[ \sum_{i=1}^N \left( \testfn(\widetilde{X}_{n}^{i})- N W_n^i\testfn(X_{n}^i)\right)\right] \\
&=  \sum_{i=1}^N \Delta_n^i(\widetilde{X}_n^i ).
\end{align*}
Let us denote $\mathcal{H}_{n, i}^N:=\mathcal{F}_{n}^N \vee \sigma\left( \left(X_{n}^{j}\right)_{j=1}^i\right)$. 
Conditionally on $\sfresampling$, the $\testfn(\widetilde{X}_n^i )$ $i=1,\ldots, N$ are independent draws from the categorical distribution with probability of outcome $\testfn(X_{n}^i)$ given by $W_n^i$ for $i=1, \ldots, N$.
It follows that the $\Delta_n^i(\widetilde{X}_n^i )$ for $i=1,\ldots, N$ are conditionally independent random variables with mean 0.
Due to the boundedness of $\testfn$,  $(\Delta_n^i(X_n^i), \mathcal{H}_{n, i}^N)_{i=1}^N$ is a square integrable martingale difference sequence which satisfies the Lindeberg condition, and for which
\begin{align*}
\sum_{i=1}^N \var\left( \Delta_n^i(\widetilde{X}_n^i ) \mid \mathcal{H}^N_{n, i-1}\right) &= \sum_{i=1}^N \Exp\left[ \Delta_n^i(\widetilde{X}_n^i )^2\mid \mathcal{F}^N_{n}\right]\\
&=\frac{1}{N}\sum_{i=1}^N \Exp\left[ \left(\testfn(\widetilde{X}_{n}^{i}) -  \Exp\left[\testfn(\widetilde{X}_{n}^i  ) \mid \mathcal{F}^N_{n}\right]\right)^2\mid \mathcal{F}^N_{n}\right]\\
&=\frac{1}{N}\sum_{i=1}^N \left( \Exp\left[ \testfn(\widetilde{X}_{n}^{i})^2\mid \mathcal{F}^N_{n}\right] -  \Exp\left[\testfn(\widetilde{X}_{n}^i  ) \mid \mathcal{F}^N_{n}\right]^2\right)\\
&= \sum_{i=1}^N W_n^i\testfn^2(X_{n}^i)  -  \left(\sum_{i=1}^N W_n^i\testfn(X_{n}^i)\right)^2\\
& = \var_{\nu^N_{n}}(\testfn),
\end{align*}
where the second to last equality is the variance of the multinomial distribution and $\nu^N_{n}:= \bgN(\predictiveN)$.
As a consequence of the WLLN (\cref{prop:wlln}), we have 
\begin{align*}
\nu^N_{n}(\testfn)= \bgN(\predictiveN)(\testfn) = \frac{\predictiveN(\weightN \testfn)}{\predictiveN(\weightN)} \longrightarrow \frac{\predictive(\weight \testfn)}{\predictive(\weight)} = \update(\testfn)
\end{align*}
in probability.
Hence, 
\begin{align*}
\var_{\nu_n^N}(\testfn) 
 = \nu_n^N(\testfn^2)-\nu_n^N(\testfn)^2 \longrightarrow \hat{\eta}_n(\testfn^2)-\hat{\eta}_n(\testfn)^2=\var_{\update}(\testfn) \qquad \text{in probability as } N\rightarrow\infty.
\end{align*}
We can apply a CLT for triangular arrays (\citet[Theorem A.3]{douc2007limit}) to show $\Exp\left[\exp\left( iu U_N\right) \mid \mathcal{F}_n^N \right]\pconverges \exp\left( -\var_{\update}(\testfn) u^2/2 \right)$.

By applying the hypothesis to $R_N$
\begin{align*}
R_N =\sqrt{N}\left[\bgN(\updateN)(\testfn) - \update(\testfn)\right]\dconverges \N\left( 0, \bar{\V}_{n}(\testfn)\right).
\end{align*}
The characteristic function of $T_N$ is
\begin{equation*}
\Phi_{T_N}(u) = \Exp\left[\exp\left( iuT_N\right) \right] = \Exp\left[\exp\left( iu R_N\right) \Exp\left[\exp\left( iu U_N\right) \mid \mathcal{F}_n^N \right]\right],
\end{equation*}
and, by the continuous mapping Theorem and  Slutzky's Lemma,
\begin{align*}
\exp\left( iu R_N\right)\Exp\left[\exp\left( iu U_N\right)\mid \mathcal{F}_n^N \right] \dconverges \exp\left( -\var_{\update}(\testfn) u^2/2 + iu Z\right)
\end{align*}
where $Z\sim \N(0, \bar{\V}_{n}(\testfn))$.
Then, by the same argument used in \cref{clt:lemma1}, $T_N$ follows a Normal distribution with mean 0 and variance $\V_n(\testfn) = \var_{\update}(\testfn) + \bar{\V}_{n}(\testfn)$.
\end{proof}
\subsection{Auxiliary results for the proof of \cref{clt:weight_comparison}}

\begin{lemma}
\label{clt:variance_convergence}
Under \cref{ass:smc0,ass:weak,ass:smc2}, for all $n\geq0$ and for every $\testfn\in\bounded$, we have
\begin{align*}
\var_{\nu^N_{n-1}}\left(\weightN\testfn\right) \pconverges \var_{\eta_n}\left(\weight\testfn\right),
\end{align*}
where $\nu^N_{n-1}=\Psi_{G_{n-1}^N}(\eta_{n-1}^N)M_n$.
\end{lemma}
\begin{proof}
Using the definition of variance we have
\begin{align}
\label{eq:var_decomp}
\var_{\nu^N_{n-1}}\left(\weightN\testfn\right) &= \Psi_{G_{n-1}^N}(\eta_{n-1}^N)\left(M_n((G_n^N\testfn)^2)\right)-\Psi_{G_{n-1}^N}(\eta_{n-1}^N)\left(M_n(G_n^N\testfn)\right)^2.
\end{align}
For the second term in~\eqref{eq:var_decomp} we can proceed as in the proof of \cref{prop:conditional_expe} and find
\begin{align*}
    \Psi_{G_{n-1}^N}(\eta_{n-1}^N)\left(M_n(G_n^N\testfn)\right) &=\sum_{j=1}^N\frac{G_{n-1}^N(X_{n-1}^j)}{\sum_{k=1}^NG_{n-1}^N(X_{n-1}^k)}\int M_n(X_{n-1}^j, \rmd x_n)\weightN(x_n)\testfn(x_n)\\
    &=\Psi_{G_{n-1}^N}(\eta_{n-1}^N)(K_n(\testfn U_n))\pconverges \Psi_{G_{n-1}}(\eta_{n-1})(K_n(\testfn U_n))
\end{align*}
where the convergence follows from the WLLN in \cref{prop:wlln}. The continuous mapping theorem then guarantees $\Psi_{G_{n-1}^N}(\eta_{n-1}^N)\left(M_n(G_n^N\testfn)\right)^2\pconverges \Psi_{G_{n-1}}(\eta_{n-1})(K_n(\testfn U_n))^2$.

To control the first term in~\eqref{eq:var_decomp} use use the $\lp$ inequality established in \cref{prop:lp}.
First, observe that by following the same steps in the proof of \cref{prop:conditional_expe} we have $\Psi_{G_{n-1}^N}(\eta_{n-1}^N)\left(M_n([G_n^N\testfn]^2)\right) = \Exp\left[\eta_n^N([G_n^N\testfn]^2)|\mathcal{F}_{n-1}^N\right]$. We can therefore write
\begin{align*}
&\Exp\left[\left\lvert \Psi_{G_{n-1}^N}(\eta_{n-1}^N)\left(M_n((G_n^N\testfn)^2)\right)- \Psi_{G_{n-1}^N}(\eta_{n-1}^N)\left(M_n((G_n\testfn)^2)\right)\right\rvert^p\right]^{1/p} \\
&\qquad\qquad= \Exp\left[\left\lvert \frac{1}{N}\sum_{i=1}^N\Exp\left[\testfn(X_n^i)^2[\weightN(X^i_n)^2-\weight(X^i_n)^2]|\mathcal{F}_{n-1}^N\right]\right\rvert^p\right]^{1/p}\\
&\qquad\qquad\leq\frac{2\supnorm{\testfn}^2m_g}{N}\sum_{i=1}^N \Exp\left[  \left\lvert\Exp\left[\weightN(X^i_n)-\weight(X^i_n)|\mathcal{F}_{n-1}^N\right]\right\rvert^p\right]^{1/p}.
\end{align*}
Computing the conditional expectation above we find
\begin{align*}
    \Exp\left[\weightN(X^i_n)-\weight(X^i_n)|\mathcal{F}_{n-1}^N\right] &= \Psi_{G_{n-1}^N}(\eta_{n-1}^N)(K_n( U_n)) - \Psi_{G_{n-1}^N}(\eta_{n-1}^N)(M_n( G_n)),
\end{align*}
and we obtain the following bound
\begin{align*}
    \Exp\left[  \left\lvert\Exp\left[\weightN(X^i_n)-\weight(X^i_n)|\mathcal{F}_{n-1}^N\right]\right\rvert^p\right]^{1/p} &\leq \Exp\left[  \left\lvert\Psi_{G_{n-1}^N}(\eta_{n-1}^N)(K_n( U_n)) - \Psi_{G_{n-1}}(\eta_{n-1})(M_n( G_n))\right\rvert^p\right]^{1/p}\\
    &+ \Exp\left[  \left\lvert\Psi_{G_{n-1}}(\eta_{n-1})(M_n( G_n))-\Psi_{G_{n-1}^N}(\eta_{n-1}^N)(M_n( G_n))\right\rvert^p\right]^{1/p}\\
    &\leq \bar{C}_{p, n-1}\frac{\supnorm{U_n}}{N^{1/2}} \textcolor{blue}{+} \bar{C}_{p, n-1}\frac{m_g}{N^{1/2}}
\end{align*}
where we used the fact that $G_n = \rmd \hat\eta_{n-1}(U_n \cdot K_n)/\rmd \hat\eta_{n-1} M_n$ in the first expectation and \cref{prop:lp} to obtain the bounds.
It follows that
\begin{align*}
    \Exp\left[\left\lvert \Psi_{G_{n-1}^N}(\eta_{n-1}^N)\left(M_n((G_n^N\testfn)^2)\right)- \Psi_{G_{n-1}^N}(\eta_{n-1}^N)\left(M_n((G_n\testfn)^2)\right)\right\rvert^p\right]^{1/p} \leq \bar{C}_{p, n-1}\frac{2\supnorm{\testfn}^2m_g}{N^{1/2}}(\supnorm{U_n}+m_g),
\end{align*}
and thus $\Psi_{G_{n-1}^N}(\eta_{n-1}^N)\left(M_n((G_n^N\testfn)^2)\right)- \Psi_{G_{n-1}^N}(\eta_{n-1}^N)\left(M_n((G_n\testfn)^2)\right)\pconverges 0$.
The WLLN in \cref{prop:wlln} then guarantees $\Psi_{G_{n-1}^N}(\eta_{n-1}^N)\left(M_n((G_n\testfn)^2)\right)\pconverges \Psi_{G_{n-1}}(\eta_{n-1})\left(M_n((G_n\testfn)^2)\right)$ and, by the continuous mapping theorem, we have
\begin{align*}
    \Psi_{G_{n-1}^N}(\eta_{n-1}^N)\left(M_n((G_n^N\testfn)^2)\right)\pconverges \Psi_{G_{n-1}}(\eta_{n-1})\left(M_n((G_n\testfn)^2)\right).
\end{align*}
A further application of the continuous mapping theorem shows that
\begin{align*}
    \var_{\nu^N_{n-1}}\left(\weightN\testfn\right) &= \Psi_{G_{n-1}^N}(\eta_{n-1}^N)\left(M_n((G_n^N\testfn)^2)\right)-\Psi_{G_{n-1}^N}(\eta_{n-1}^N)\left(M_n(G_n^N\testfn)\right)^2\\
    &\pconverges \Psi_{G_{n-1}}(\eta_{n-1})\left(M_n((G_n\testfn)^2)\right) - \Psi_{G_{n-1}}(\eta_{n-1})\left(M_n(G_n\testfn)\right)^2\\
    &=\predictive((G_n\testfn)^2)-\predictive(G_n\testfn)^2 = \var_{\predictive}(G_n\testfn).
\end{align*}

\end{proof}

\section{Variance expressions}
\label{app:variance}

\subsection{Closed Form Expression}
\label{app:variance_closed}
To obtain the variance expression in \cref{prop:variance_closed} it is useful to introduce the unnormalized measures
\begin{align*}
    \hat{\gamma}_n(\rmd x_n) = \int U_n(x_{n-1}, x_n)K_n(x_{n-1}, \rmd x_n)\hat{\eta}_{n-1}(\rmd x_{n-1}),
\end{align*}
so that $\hat{\eta}_n=\hat{\gamma}_n/\hat{\gamma}_n(1)$, where $1$ denotes the unit function $x\in E\mapsto 1$.
We highlight that in the literature on standard SMC methods the unnormalized measure $\hat{\gamma}_n$ normally has total mass such that $\hat\gamma_n(1)=\prod_{p=0}^n \eta_p(G_p)$. In our case, the present definition of  $\hat{\gamma}_n$ is more convenient as we work more directly on the marginal space $E$.

With the notation introduced above we have
\begin{align*}
\weight(x_n) = \frac{\rmd \hat\eta_{n-1}(U_n \cdot K_n)}{\rmd \hat\eta_{n-1} M_n}(x_n)  = \frac{\rmd \hat\eta_{n-1}(U_n \cdot K_n)}{\rmd \eta_n}(x_n)  = \frac{\rmd \hat\gamma_n}{\rmd \eta_n}(x_n),
\end{align*}
from which it is easy to see that $\predictive(\weight) = \hat{\gamma}_n(1)$ and
\begin{align}
\label{eq:normalized_weight}
\frac{G_n(x_n)}{\predictive(\weight)} = \frac{\rmd \hat{\eta}_n}{\rmd \eta_n}(x_n)
\end{align}
for all $n\geq 0$.

\begin{proof}[Proof of \cref{prop:variance_closed}]
We prove the result by induction.
At $n=0$ we have
\begin{align*}
\bar{\V}_0(\testfn) &= \eta_0\left[\left(\frac{G_0}{\eta_0(G_0)}\left[\testfn - \hat{\eta}_0(\testfn)\right]\right)^2\right].
\end{align*}
At $n=1$, we have from \cref{prop:clt}
\begin{align*}
\bar{\V}_1(\testfn) &= \frac{1}{\eta_1(G_1)^2}\hat{\V}_1\left(\testfn - \hat{\eta}_1(\testfn)\right) =\frac{1}{\eta_1(G_1)^2}\var_{\eta_1}(G_1\left[\testfn - \hat{\eta}_1(\testfn)\right])+ \frac{1}{\eta_1(G_1)^2}\bar{\V}_{0}\left(K_1(U_1\testfn) - \hat{\eta}_1(\testfn)K_1(U_1)\right)\\
&=\frac{1}{\eta_1(G_1)^2}\eta_1\left[\left(G_1\left[\testfn - \hat{\eta}_1(\testfn)\right]\right)^2\right]+ \frac{1}{\eta_1(G_1)^2}\bar{\V}_{0}\left(K_1(U_1\testfn) - \hat{\eta}_1(\testfn)K_1(U_1)\right),
\end{align*}
where the last equality follows from using~\eqref{eq:normalized_weight}. For the second term we have
\begin{align*}
&\frac{1}{\eta_1(G_1)^2}\bar{\V}_{0}\left(K_1(U_1\testfn) - \hat{\eta}_1(\testfn)K_1(U_1)\right)\\
&\qquad= \frac{1}{\eta_1(G_1)^2}\frac{1}{\eta_0(G_0)^2}\eta_0\left[G_0^2\left[K_1(U_1\testfn) - \hat{\eta}_1(\testfn)K_1(U_1) - \hat{\eta}_0(K_1(U_1\testfn) - \hat{\eta}_1(\testfn)K_1(U_1))\right]^2\right]\\
&\qquad= \frac{1}{\eta_1(G_1)^2}\frac{1}{\eta_0(G_0)^2}\eta_0\left[G_0^2\left[K_1(U_1\testfn) - \hat{\eta}_1(\testfn)K_1(U_1) \right]^2\right],
\end{align*}
since $\hat{\eta}_0(K_1(U_1\testfn) - \hat{\eta}_1(\testfn)K_1(U_1)) = \hat{\gamma}_1(\testfn) - \hat{\eta}_1(\testfn)\hat{\gamma}_1(1) =0$.
It follows that
\begin{align*}
\bar{\V}_1(\testfn) &= \frac{1}{\eta_1(G_1)^2}\eta_1\left[\left(G_1\left[\testfn - \hat{\eta}_1(\testfn)\right]\right)^2\right] +  \frac{1}{\eta_1(G_1)^2}\frac{1}{\eta_0(G_0)^2}\eta_0\left[G_0^2\left[K_1(U_1\testfn) - \hat{\eta}_1(\testfn)K_1(U_1) \right]^2\right],
\end{align*}
as required.

Whe now establish the inductive step: assuming~\eqref{eq:variance_expression} holds for $n-1$, then, using~\eqref{eq:normalized_weight},
\begin{align}
\label{eq:variance1}
\bar{\V}_n(\testfn) &= \frac{1}{\predictive(\weight)^2}\widehat{\V}_n\left(\testfn - \update(\testfn)\right)\\
&= \frac{1}{\predictive(\weight)^2}\var_{\predictive}(\weight\left[\testfn - \update(\testfn)\right]) + \frac{1}{\predictive(\weight)^2}\bar{\V}_{n-1}\left(K_n(U_n\left[\testfn - \update(\testfn)\right])\right)\notag\\
&= \frac{1}{\predictive(\weight)^2}\var_{\predictive}(\weight\left[\testfn - \update(\testfn)\right])+ \frac{1}{\predictive(\weight)^2}\bar{\V}_{n-1}\left(K_n(U_n\testfn) - \update(\testfn)K_n(U_n)\right)\notag\\
&= \frac{1}{\predictive(\weight)^2}\predictive\left[\left(\weight\left[\testfn - \update(\testfn)\right]\right)^2\right]+ \frac{1}{\predictive(\weight)^2}\bar{\V}_{n-1}\left(K_n(U_n\testfn) - \update(\testfn)K_n(U_n)\right).\notag
\end{align}
Let us denote $\psi_n:=K_n(U_n\testfn) - \update(\testfn)K_n(U_n)$, using~\eqref{eq:variance_expression} we have by the induction hypothesis that:
\begin{align}
\label{eq:inductive_step}
&\bar{\V}_{n-1}\left(\psi_n\right) =\sum_{k=0}^{n-1}\prod_{j=k}^{n-1}\frac{1}{\eta_{j}(G_{j})^2}\eta_k\left[\left(G_k\left[\Gamma_{k:n-1}\left(\psi_n\right) - \hat{\eta}_{n-1}(\psi_n)\Gamma_{k:n-1}(1)\right]\right)^2\right].
\end{align}
Now, observe that
\begin{align*}
\hat{\eta}_{n-1}(\psi_n) =\hat{\eta}_{n-1}\left(K_n(U_n\testfn) - \update(\testfn)K_n(U_n)\right) 
=\hat{\gamma}_n(\testfn)-\update(\testfn)\hat{\gamma}_n(1)=0,
\end{align*}
and $\Gamma_{k:n-1}\left(K_n(U_n\testfn)\right) =  \Gamma_{k:n-1} \circ \Gamma_n(\testfn) = \Gamma_{k:n}(\testfn)$.
Thus, we can simplify~\eqref{eq:inductive_step} to
\begin{align}
\label{eq:variance2}
\bar{\V}_{n-1}\left(K_n(U_n\testfn) - \update(\testfn)K_n(U_n)\right) &= \sum_{k=0}^{n-1}\prod_{j=k}^{n-1}\frac{1}{\eta_{j}(G_{j})^2}\eta_k\left[\left(G_k\left[\Gamma_{k:n}\left(\testfn\right)-\update(\testfn)\Gamma_{k:n}(1)\right]\right)^2\right].
\end{align}
Plugging~\eqref{eq:variance2} into~\eqref{eq:variance1} we obtain
\begin{align*}
\bar{\V}_n(\testfn) 
&= \frac{1}{\predictive(\weight)^2}\predictive\left[\left(\weight\left[\testfn - \update(\testfn)\right]\right)^2\right] + \frac{1}{\predictive(\weight)^2}\sum_{k=0}^{n-1}\prod_{j=k}^{n-1}\frac{1}{\eta_{j}(G_{j})^2}\eta_k\left[\left(G_k\left[\Gamma_{k:n}\left(\testfn\right)-\update(\testfn)\Gamma_{k:n}(1)\right]\right)^2\right]\\
&=\sum_{k=0}^{n}\prod_{j=k}^{n}\frac{1}{\eta_{j}(G_{j})^2}\eta_k\left[\left(G_k\left[\Gamma_{k:n}\left(\testfn\right)-\update(\testfn)\Gamma_{k:n}(1)\right]\right)^2\right]
\end{align*}
The result follows by induction.
\end{proof}

\subsection{Variance of MPF}
\label{app:variance_mpf}

To obtain the variance expressions for the marginal particle filter discussed in \cref{sec:mpf} we simply plug in the corresponding $U_n, K_n, M_n$ into~\eqref{eq:variance_expression}.

From \cref{prop:variance_closed} we have
\begin{align*}
\bar{\V}_n(\testfn) = \sum_{k=0}^{n}\eta_k\left[\left(G_k\left[\Gamma_{k:n}\left[\testfn - \hat{\eta}_n(\testfn)\right]\right]\right)^2\right]\prod_{j=k}^{n}\frac{1}{\eta_{j}(G_{j})^2}.
\end{align*}
The last term in this sum is given by
\begin{align*}
\eta_n\left[\left(\frac{G_n}{\eta_n(G_n)}\left[\Gamma_{n:n}\left[\testfn - \hat{\eta}_n(\testfn)\right]\right]\right)^2\right] &=\eta_n\left( \left(\frac{\rmd \hat{\eta}_{n}}{\rmd \eta_{n}}\right)^2\left[\testfn - \hat{\eta}_n(\testfn)\right]^2\right)\\
&= \int\frac{p(x_{n}|y_{1:n})^2\left[\testfn(x_n) - \hat{\eta}_n(\testfn)\right]^2}{\int q_n(x_n|x_{n-1}', y_n)p(x_{n-1}'|y_{n-1})\rmd x_{n-1}'}\rmd x_n.
\end{align*}
Using the expressions for  $U_n, K_n, M_n$ in \cref{sec:mpf} we obtain
\begin{align*}
    \eta_j(G_j) = \int g_j(y_j|x_j)f_j(x_j|x_{j-1})p(x_{j-1}|y_{1:j-1})\rmd x_{j-1:j} = p(y_j|y_{1:j-1}),
\end{align*}
and, for any $\testfn \in \bounded$, we have
\begin{align*}
  \left( \prod_{j=k+1}^{n}\frac{1}{\eta_{j}(G_{j})} \right) \Gamma_{k:n}(\testfn)(x_k)&= \int \prod_{j=k+1}^{n}\frac{g_{j}(y_{j}|x_{j}) f_{j}(x_{j}|x_{j-1})}{ p(y_j|y_{1:j-1})}\testfn(x_n) \rmd x_{k+1:n}\\
  &=\int \frac{ p(x_{k:n}|y_{1:n})}{p(x_{k}|y_{1:k})}\testfn(x_n)\rmd x_{k+1:n}\\
   &=\int \frac{ p(x_{k+1:n}|y_{k+1:n}, x_k)p(x_{k}|y_{1:n})}{p(x_{k}|y_{1:k})}\testfn(x_n)\rmd x_{k+1:n}\\
   &= \frac{p(x_{k}|y_{1:n})}{p(x_{k}|y_{1:k})}\int  p(x_{n}|y_{k+1:n}, x_k)\testfn(x_n)\rmd x_{n},
\end{align*}
where we used the fact that
\begin{align}
\label{eq:gamma_simplified}
\frac{ p(x_{k:n}|y_{1:n})}{p(x_{k}|y_{1:k})}=\prod_{j=k+1}^{n}\frac{g_{j}(y_{j}|x_{j}) f_{j}(x_{j}|x_{j-1})}{ p(y_j|y_{1:j-1})}.
\end{align}
For all $k\neq 0$ we have
\begin{align*}
&\prod_{j=k}^{n}\frac{1}{\eta_{j}(G_{j})^2}\eta_k\left[\left(G_k\left[\Gamma_{k:n}\left(\testfn\right)-\update(\testfn)\Gamma_{k:n}(1)\right]\right)^2\right] \\
&\qquad=\int \frac{p(x_k|y_{1:k})^2}{\int q_k(x_k|x_{k-1}, y_k)p(x_{k-1}|y_{1:k-1})\rmd x_{k-1}}\frac{ p(x_{k}|y_{1:n})^2}{p(x_{k}|y_{1:k})^2}\left(\int p(x_{n}|y_{k+1:n}, x_k)\left[\testfn(x_n)-\update(\testfn)\right]\rmd x_{n}\right)^2\rmd x_k\\
&\qquad=\int \frac{p(x_k|y_{1:n})^2}{\int q_k(x_k|x_{k-1}, y_k)p(x_{k-1}|y_{1:k-1})\rmd x_{k-1}}\left(\int p(x_{n}|y_{k+1:n}, x_k)\left[\testfn(x_n)-\hat{\eta}_n(\testfn)\right]\rmd x_{n}\right)^2\rmd x_k,
\end{align*}
and, for $k=0$
\begin{align*}
&\prod_{j=0}^{n}\frac{1}{\eta_{j}(G_{j})^2}\eta_0\left[\left(G_0\left[\Gamma_{0:n}\left(\testfn\right)-\update(\testfn)\Gamma_{0:n}(1)\right]\right)^2\right] \\
&\qquad=\int \frac{p(x_0|y_{1:n})^2}{
q_0(x_0)}\left(\int p(x_{n}|y_{1:n}, x_0)\left[\testfn(x_n)-\update(\testfn)\right]\rmd x_{n}\right)^2\rmd x_0,
\end{align*}
The above gives
\begin{align*}
V_n^{\textrm{MPF}}(\testfn) &=\int  \frac{p(x_0|y_{1:n})^2}{
q_0(x_0)}\left(\int p(x_{n}|y_{1:n}, x_0)\left[\testfn(x_n)-\bar{\testfn}_n\right]\rmd x_{n}\right)^2\rmd x_0\\
&+\sum_{k=1}^{n-1} \int \frac{p(x_k|y_{1:n})^2}{\int q_k(x_k|x_{k-1}, y_k)p(x_{k-1}|y_{1:k-1})\rmd x_{k-1}}\left(\int p(x_{n}|y_{k+1:n}, x_k)\left[\testfn(x_n)-\bar{\testfn}_n\right]\rmd x_{n}\right)^2\rmd x_k\notag\\
&+\int \frac{p(x_{n}|y_{1:n})^2}{\int q_{n}(x_n|x_{n-1}, y_n) p(x_{n-1}|y_{1:n-1})\rmd x_{n-1}} (\testfn(x_n)-\bar{\testfn}_n)^2\rmd x_{n},\notag
\end{align*}
since $\update(\testfn)=\bar{\testfn}_n= \int \testfn(x_n)p(x_n|y_{1:n})\rmd x_n$.
\subsubsection{Variance of BPF}
\label{app:variance_bpf}

The bootstrap particle filter is a special case of MPF in which $q_n=K_n\equiv f_n$. If we plug this into the expression in \cref{cor:mpf} we obtain 
\begin{align*}
    \V_n^{\textrm{BPF}}(\testfn) &= \int  \frac{p(x_0|y_{1:n})^2}{
f_0(x_0)}\left(\int p(x_{n}|y_{1:n}, x_0)\left[\testfn(x_n)-\bar{\testfn}_n\right]\rmd x_{n}\right)^2\rmd x_0\\
    &+\sum_{k=1}^{n-1}\int \frac{p(x_{k}|y_{1:n})^2}{\int f_{k}(x_k|x_{k-1})  p(x_{k-1}|y_{1:k-1})\rmd x_{k-1}}\left(\int p(x_{n}|y_{k+1:n}, x_k)\left[ \testfn(x_n)-\bar{\testfn}_n\right]\rmd x_{n}\right)^2\rmd x_{k}\\
    &+\int \frac{p(x_{n}|y_{1:n})^2}{\int f_{n}(x_n|x_{n-1}) p(x_{n-1}|y_{1:n-1})\rmd x_{n-1}} \left[\testfn(x_n)-\bar{\testfn}_n\right]^2\rmd x_{n}.
\end{align*}
Since $p(x_{k}|y_{1:n}) = \int p(x_{0:k}|y_{1:n})\rmd x_{0:k-1}$ and 
\begin{align*}
\frac{p(x_{k}|y_{1:n})^2}{\int f_{k}(x_k|x_{k-1})  p(x_{k-1}|y_{1:k-1})\rmd x_{k-1}} = \frac{p(x_{0:k}|y_{1:n})^2}{p(x_{0:k-1}|y_{1:k-1})f_{k}(x_k|x_{k-1}) },
\end{align*}
we obtain
\begin{align}
\label{eq:bpf}
        \V_n^{\textrm{BPF}}(\testfn) &= \int  \frac{p(x_0|y_{1:n})^2}{
f_0(x_0)}\left(\int p(x_{n}|y_{1:n}, x_0)\left[\testfn(x_n)-\bar{\testfn}_n\right]\rmd x_{n}\right)^2\rmd x_0\\
    &+\sum_{k=1}^{n-1}\int \frac{p(x_{0:k}|y_{1:n})^2}{p(x_{0:k-1}|y_{1:k-1})f_{k}(x_k|x_{k-1}) }\left(\int p(x_{n}|y_{k+1:n}, x_k)\left[\testfn(x_n)-\bar{\testfn}_n\right]\rmd x_{n}\right)^2\rmd x_{0:k}\notag\\
    &+\int \frac{p(x_{0:n}|y_{1:n})^2}{p(x_{0:n-1}|y_{1:n-1}) f_{n}(x_n|x_{n-1}) }\int \left[\testfn(x_n)-\bar{\testfn}_n\right]^2\rmd x_{0:n},\notag
\end{align}
which coincides with the variance expression in \citet[Section 2.4]{johansen2008note}.

\subsection{Variance of MAPF}
\label{app:variance_mapf}
To obtain the variance expressions for the marginal particle filter discussed in \cref{sec:mapf} we simply plug in the corresponding $U_n, K_n, M_n$ into~\eqref{eq:variance_expression}.

From \cref{prop:variance_closed} we have
\begin{align*}
\bar{\V}_n(\testfn) = \sum_{k=0}^{n}\eta_k\left[\left(G_k\left[\Gamma_{k:n}(\testfn) - \hat{\eta}_n(\testfn)\Gamma_{k:n}(1)\right]\right)^2\right]\prod_{j=k}^{n}\frac{1}{\eta_{j}(G_{j})^2}.
\end{align*}
Since the MAPF estimates are obtained from $\update$ with an additional importance sampling step with weights $\tilde{w}_n$, proceeding as in \cref{clt:lemma3} we obtain that the variance of MAPF is $\V^{\textrm{MAPF}}_n(\testfn) =\bar{\V}_n(\widetilde{w}_n\left[\testfn-\bar{\testfn}_n\right]) $, where $\bar{\testfn}_n:=\int \testfn(x_n)p(x_n|y_{1:n})\rmd x_n$.
Observing that
\begin{align*}
    \update(\widetilde{w}_n\left[\testfn-\bar{\testfn}_n\right]) &=\update\left( \frac{\rmd \pi_n}{\rmd \update}\left[\testfn-\bar{\testfn}_n\right]\right) =0,
\end{align*}
where $\pi_n(\rmd x_n):= p(x_n|y_{1:n})\rmd x_n$, we have that the variance of MAPF is
\begin{align*}
\V^{\textrm{MAPF}}_n(\testfn) 
&= \sum_{k=0}^{n}\eta_k\left[\frac{G_k^2}{\eta_k(G_k)^2}\Gamma_{k:n}\left(\tilde{w}_n\left[\testfn - \bar{\testfn}_n\right]\right)^2\right]\prod_{j=k+1}^{n}\frac{1}{\eta_{j}(G_{j})^2}.
\end{align*}
The last term in this sum is given by
\begin{align*}
    \eta_n\left[\frac{G_n^2}{\eta_{n}(G_{n})^2}\Gamma_{n:n}\left(\tilde{w}_n\left[\testfn - \bar{\testfn}_n\right]\right)^2\right]
    &= \int\frac{p(x_{n}|y_{1:n})^2\left[\testfn(x_n) - \bar{\testfn}_n\right]^2}{\int q_n(x_n|x_{n-1}', y_n)\hat{\eta}_{n-1}(\rmd x_{n-1}')}\rmd x_n.
\end{align*}
Using the expression for $\tilde{w}_n$ in~\eqref{eq:inferential_weights} and that of $K_n, U_n$ given in \cref{sec:mapf} and~\eqref{eq:gamma_simplified}, we find that
\begin{align}
\label{eq:mapf_step}
   & \prod_{j=k+1}^{n}\frac{1}{\eta_{j}(G_{j})}\Gamma_{k:n}\left(\tilde{w}_n\left[\testfn - \bar{\testfn}_n\right]\right)(x_k) \\
   &\qquad=  \prod_{j=k+1}^{n}\frac{1}{\eta_{j}(G_{j})}\int \frac{g_{k+1}(y_{k+1}|x_{k+1})\tilde{p}(y_{k+2}|x_{k+1}) }{\tilde{p}(y_{k+1}|x_{k})}f_{k+1}(x_{k+1}|x_k)\int\dots \notag\\
   &\qquad\qquad\dots \int \frac{g_{n}(y_{n}|x_{n})\tilde{p}(y_{n+1}|x_{n})}{\tilde{p}(y_{n}|x_{n-1})} f_{n}(x_{n}|x_{n-1})\frac{\rmd \pi_n}{\rmd \update}(x_n)\left[\testfn(x_n) - \bar{\testfn}_n\right] \rmd x_{k+1:n}\notag\\
   &\qquad=\int \left(\prod_{j=k+1}^{n}\frac{g_{j}(y_{j}|x_{j}) f_{j}(x_{j}|x_{j-1})}{\eta_{j}(G_{j})}\right)\frac{\tilde{p}(y_{n+1}|x_{n})}{\tilde{p}(y_{k+1}|x_{k})}\frac{\rmd \pi_n}{\rmd \update}(x_n)\left[\testfn(x_n) - \bar{\testfn}_n\right] \rmd x_{k+1:n}\notag\\
    &\qquad=\left(\prod_{j=k+1}^{n}\frac{p(y_j|y_{1:j-1})}{\eta_j(G_j)}\right)\int \frac{p(x_{k:n}|y_{1:n})}{p(x_k|y_{1:k})}\frac{\tilde{p}(y_{n+1}|x_{n})}{\tilde{p}(y_{k+1}|x_{k})}\frac{\rmd \pi_n}{\rmd \update}(x_n)\left[\testfn(x_n) - \bar{\testfn}_n\right] \rmd x_{k+1:n}\notag\\
    &\qquad=\left(\prod_{j=k+1}^{n}\frac{p(y_j|y_{1:j-1})}{\eta_j(G_j)}\right)\frac{p(x_{k}|y_{1:n})}{p(x_k|y_{1:k})}\int p(x_{k+1:n}|x_k, y_{1:n})\frac{\tilde{p}(y_{n+1}|x_{n})}{\tilde{p}(y_{k+1}|x_{k})}\frac{\rmd \pi_n}{\rmd \update}(x_n)\left[\testfn(x_n) - \bar{\testfn}_n\right] \rmd x_{k+1:n}.\notag
\end{align}
Since $\update(\rmd x_n)= \tilde{p}(y_{n+1}|x_{n})\pi_n(\rmd x_n)/\int \tilde{p}(y_{n+1}|x_{n}')\pi_n(\rmd x_n')$ and $\pi_n(\rmd x_n) = p(x_n|y_{1:n})\rmd x_n$, we also have that
\begin{align*}
\tilde{p}(y_{n+1}|x_{n})\frac{\rmd \pi_n}{\rmd \update}(x_n) &= \int \tilde{p}(y_{n+1}|x'_{n})p(x_n'|y_{1:n})\rmd x_n':= \tilde{p}(y_{n+1}|y_{1:n})
\end{align*}
and
\begin{align*}
\eta_j(G_j) = p(y_j|y_{1:j-1})\frac{ \int p(x_{j}|y_{1:j})\tilde{p}(y_{j+1}|x_{j})\rmd x_{j}}{\int p(x_{j-1}|y_{1:j-1})\tilde{p}(y_{j}|x_{j-1})\rmd x_{j-1}} = p(y_j|y_{1:j-1})\frac{\tilde{p}(y_{j+1}|y_{1:j})}{\tilde{p}(y_{j}|y_{1:j-1})},
\end{align*}
from which follows
\begin{align*}
\prod_{j=k+1}^{n}\frac{ p(y_j|y_{1:j-1})}{\eta_j(G_j)} &= \frac{\tilde{p}(y_{k+1}|y_{1:k})}{\tilde{p}(y_{n+1}|y_{1:n})}.
\end{align*}
Using the above we can simplify~\eqref{eq:mapf_step} to
\begin{align*}
 & \prod_{j=k+1}^{n}\frac{1}{\eta_{j}(G_{j})}\Gamma_{k:n}\left(\tilde{w}_n\left[\testfn - \bar{\testfn}_n\right]\right)(x_k) \\
 &\qquad=\frac{p(x_{k}|y_{1:n})}{p(x_k|y_{1:k})}\frac{\tilde{p}(y_{k+1}|y_{1:k})}{\tilde{p}(y_{k+1}|x_{k})}\int p(x_{n}| y_{k+1:n}, x_k)\left[\testfn(x_n) - \bar{\testfn}_n\right] \rmd x_{n},
\end{align*}
so that for $k\neq0$ we have
 \begin{align*}
&\eta_k\left[\frac{G_k^2}{\eta_k(G_k)^2}\Gamma_{k:n}\left(\tilde{w}_n\left[\testfn - \bar{\testfn}_n\right]\right)^2\right]\prod_{j=k+1}^{n}\frac{1}{\eta_{j}(G_{j})^2} \\
&\qquad=\int\frac{p(x_{k}|y_{1:n})^2}{\int q_k(x_k|x_{k-1}', y_k)\hat{\eta}_{k-1}(\rmd x_{k-1}')}\left(\int p(x_n|y_{k+1:n}, x_k)\left[\testfn(x_n)-\bar{\testfn}_n\right]\rmd x_{n}\right)^2\rmd x_k,
 \end{align*}
and, for $k=0$,
\begin{align*}
&\eta_0\left[\frac{G_0^2}{\eta_0(G_0)^2}\Gamma_{0:n}\left(\tilde{w}_n\left[\testfn - \bar{\testfn}_n\right]\right)^2\right]\prod_{j=1}^{n}\frac{1}{\eta_{j}(G_{j})^2} \\
&\qquad=\int\frac{p(x_{0}|y_{1:n})^2}{q_0(x_0)}\left(\int p(x_n|y_{1:n}, x_0)\left[\testfn(x_n)-\bar{\testfn}_n\right]\rmd x_{n}\right)^2\rmd x_0,
 \end{align*}
Combining the expressions above we obtain
\begin{align*}
\V^{\textrm{MAPF}}_n(\testfn) 
&= \int \frac{p(x_{0}|y_{1:n})^2}{q_0(x_0)}\left(\int p(x_n|y_{1:n}, x_0)\left[\testfn(x_n)-\bar{\testfn}_n\right]\rmd x_{n}\right)^2\rmd x_0\\
&+\int\frac{p(x_{k}|y_{1:n})^2}{\int q_k(x_k|x_{k-1}', y_k)\hat{\eta}_{k-1}(\rmd x_{k-1}')}\left(\int p(x_n|y_{k+1:n}, x_k)\left[\testfn(x_n)-\bar{\testfn}_n\right]\rmd x_{n}\right)^2\rmd x_k\\
&+\int\frac{p(x_{n}|y_{1:n})^2\left[\testfn(x_n) - \bar{\testfn}_n\right]^2}{\int q_n(x_n|x_{n-1}', y_n)\hat{\eta}_{n-1}(\rmd x_{n-1}')}\rmd x_n.
 \end{align*}

\subsubsection{Variance of FA-MAPF}
\label{app:variance_fa-mapf}

The variance expression for FA-MAPF follows by plugging into the above $M_0 =q_0= f_0$,  $\tilde{p}(y_n|x_{n-1}) = p(y_n|x_{n-1})$ and $q_n(x_n|x_{n-1}, y_n) = \frac{g_n(y_n|x_n)f_n(x_n|x_{n-1})}{p(y_n|x_{n-1})}$.
Observing that 
\begin{align*}
\int q_n(x_n|x_{n-1}, y_n)\hat{\eta}_{n-1}(\rmd x_{n-1}) &=\int\frac{ g_n(y_n|x_n) f_n(x_{n-1}, x_n)}{p(y_n|x_{n-1})}\frac{p(x_{n-1}|y_{1:n-1})p(y_n|x_{n-1})}{\int p(x_{n-1}'|y_{1:n-1})p(y_n|x_{n-1}') \rmd x_{n-1}'}\rmd x_{n-1}'\\
&=\frac{g_n(y_n|x_n)\int  f_n(x_{n-1}, x_n)p(x_{n-1}|y_{1:n-1})\rmd x_{n-1}}{\int \int g_n(y_n|x_n) f_n(x_{n-1}', x_n)\rmd x_n p(x_{n-1}'|y_{1:n-1})\rmd x_{n-1}'}=p(x_n|y_{1:n})
\end{align*}
 for all $n\geq 1$, we have
\begin{align*}
\V^{\textrm{FA-MAPF}}_n(\testfn) 
&= \int \frac{p(x_{0}|y_{1:n})^2}{f_0(x_0)}\left(\int p(x_n|y_{1:n}, x_0)\left[\testfn(x_n)-\bar{\testfn}_n\right]\rmd x_{n}\right)^2\rmd x_0\\
&+\int\frac{ p(x_{k}|y_{1:n})^2}{p(x_{k}|y_{1:k})}\left(\int p(x_n|y_{k+1:n}, x_k)\left[\testfn(x_n)-\bar{\testfn}_n\right]\rmd x_{n}\right)^2\rmd x_k\\
&+\int p(x_{n}|y_{1:n})\left[\testfn(x_n) - \bar{\testfn}_n\right]^2\rmd x_n,
 \end{align*}
which coincides with the variance expression in \citet[Section 2.4]{johansen2008note}.

\subsection{Variance Comparison}
\label{app:variance_comparison}
We now compare the variance expressions in \cref{cor:mpf,cor:mapf} with those of the corresponding non-marginal algorithms obtained in \citet{johansen2008note}.
\subsubsection{Marginal Particle Filters}
\label{sec:comparison_mpf}

The MPF variance in \cref{cor:mpf} can be decomposed as $V_n^{\textrm{MPF}}(\testfn) =\sum_{k=0}^{n} \sigma_{n,k}^{\textrm{MPF}}(\testfn)$ where
\begin{align*}
\sigma_{n,0}^{\textrm{MPF}}(\testfn)&=\int  \frac{p(x_0|y_{1:n})^2}{
q_0(x_0)}\left(\int p(x_{n}|y_{1:n}, x_0)\left[\testfn(x_n)-\bar{\testfn}_n\right]\rmd x_{n}\right)^2\rmd x_0,\\
\sigma_{n,k}^{\textrm{MPF}}(\testfn)&= \int \frac{p(x_k|y_{1:n})^2}{\int q_k(x_k|x_{k-1}, y_k)p(x_{k-1}|y_{1:k-1})\rmd x_{k-1}}\left(\int p(x_{n}|y_{k+1:n}, x_k)\left[\testfn(x_n)-\bar{\testfn}_n\right]\rmd x_{n}\right)^2\rmd x_k
\end{align*}
for $k\in\{1,\ldots,n-1\}$, and
\begin{align*}
\sigma_{n,n}^{\textrm{MPF}}(\testfn)&=\int \frac{p(x_{n}|y_{1:n})^2}{\int q_{n}(x_n|x_{n-1}, y_n) p(x_{n-1}|y_{1:n-1})\rmd x_{n-1}} (\testfn(x_n)-\bar{\testfn}_n)^2\rmd x_{n},\notag
\end{align*}
while the variance of a standard (i.e. not marginal) particle filter given in  \citet[Section 2.4]{johansen2008note} is given by $V_n^{\textrm{PF}}(\testfn) =\sum_{k=0}^{n} \sigma_{n,k}^{\textrm{PF}}(\testfn)$ where
\begin{align*}
\sigma_{n,0}^{\textrm{PF}}(\testfn)&=\int  \frac{p(x_0|y_{1:n})^2}{
q_0(x_0)}\left(\int p(x_{n}|y_{1:n}, x_0)\left[\testfn(x_n)-\bar{\testfn}_n\right]\rmd x_{n}\right)^2\rmd x_0,\\
\sigma_{n,k}^{\textrm{PF}}(\testfn)&= \int \frac{p(x_{0:k}|y_{1:n})^2}{ q_k(x_k|x_{k-1}, y_k)p(x_{0:k-1}|y_{1:k-1})}\left(\int p(x_{n}|y_{k+1:n}, x_k)\left[\testfn(x_n)-\bar{\testfn}_n\right]\rmd x_{n}\right)^2\rmd x_{0:k}
\end{align*}
for $k\in\{1,\ldots,n-1\}$, and
\begin{align*}
\sigma_{n,n}^{\textrm{PF}}(\testfn)&=\int \frac{p(x_{0:n}|y_{1:n})^2}{ q_{n}(x_n|x_{n-1}, y_n) p(x_{0:n-1}|y_{1:n-1})} (\testfn(x_n)-\bar{\testfn}_n)^2\rmd x_{0:n}.\notag
\end{align*}
From the above it is easy to see that $\sigma_{n,0}^{\textrm{MPF}}(\testfn)=\sigma_{n,0}^{\textrm{PF}}(\testfn)$.
Using~\eqref{eq:gamma_simplified} we have
\begin{align*}
p(x_k|y_{1:n}) &=\frac{\int p(x_{0:n},y_{1:n})\rmd x_{0:k-1}\rmd x_{k+1:n}}{p(y_{1:n})}\\
&=\frac{ g_k(y_k|x_k)}{p(y_k|y_{1:k-1})}\int f_k(x_k|x_{k-1})p(x_{0:k-1}|y_{1:k-1})\rmd x_{0:k-1}\int \frac{ p(x_{k:n}|y_{1:n})}{p(x_{k}|y_{1:k})}\rmd x_{k+1:n}
\end{align*}
and
\begin{align*}
p(x_{0:k}|y_{1:n}) &=\frac{\int p(x_{0:n},y_{1:n})\rmd x_{k+1:n}}{p(y_{1:n})}\\
&= \frac{g_k(y_k|x_k)}{p(y_k|y_{1:k-1})} f_k(x_k|x_{k-1})p(x_{0:k-1}|y_{1:k-1})\int \frac{ p(x_{k:n}|y_{1:n})}{p(x_{k}|y_{1:k})}\rmd x_{k+1:n},
\end{align*}
and we can write for all $k\geq 1$
\begin{align*}
&\sigma_{n,k}^{\textrm{PF}}(\testfn) - \sigma_{n,k}^{\textrm{MPF}}(\testfn) \\
&\qquad= \int  g_k(y_k|x_k)^2\left(\int \frac{ f_k(x_k|x_{k-1})^2p(x_{k-1}|y_{1:k-1})}{q_k(x_k|x_{k-1}, y_k)}\rmd x_{k-1}-\frac{\left(\int f_k(x_k|x_{k-1})p(x_{k-1}|y_{1:k-1})\rmd x_{k-1}\right)^2}{\int q_k(x_k|x_{k-1}, y_k)p(x_{k-1}|y_{1:k-1})\rmd x_{k-1}}\right)\\
&\qquad\qquad\times\psi_k(x_k)\rmd x_{k}
\end{align*}
where we defined
\begin{align*}
    \psi_k(x_k):= \frac{1}{p(y_k|y_{1:k-1})^2} \left(\int \frac{ p(x_{k:n}|y_{1:n})}{p(x_{k}|y_{1:k})}\rmd x_{k+1:n}\right)^2\left(\int p(x_{n}|y_{k+1:n}, x_k)\left[\testfn(x_n)-\bar{\testfn}_n\right]\rmd x_{n}\right)^2
\end{align*}
for $1\leq k <n$ and $\psi_n(x_n):=\left[\testfn(x_n)-\bar{\testfn}_n\right]^2/p(y_n|y_{1:n-1})^2$ for $k=n$.
Since $\psi_k$ is always positive, we have that $\sigma_{n,k}^{\textrm{PF}}(\testfn) \geq \sigma_{n,k}^{\textrm{MPF}}(\testfn)$
if for all $x_k\in E$
\begin{align}
   \label{eq:assumption_mpf}
    g_k(y_k|x_k)^2\int\frac{ f_k(x_k|x_{k-1})^2p(x_{k-1}|y_{1:k-1})}{q_k(x_k|x_{k-1}, y_k)}\rmd x_{k-1}\geq g_k(y_k|x_k)^2\frac{\left( \int f_k(x_k|x_{k-1})p(x_{k-1}|y_{1:k-1})\rmd x_{k-1}\right)^2}{\int q_k(x_k|x_{k-1}, y_k)p(x_{k-1}|y_{1:k-1})\rmd x_{k-1}}.
\end{align}
The condition in~\eqref{eq:assumption_mpf} is equivalent to
\begin{align*}
&\int \left[G_k^{\textrm{PF}}(x_{k-1}, x_k)\right]^2
q_k(x_k|x_{k-1}, y_k)p(x_{k-1}|y_{1:k-1})\rmd x_{k-1}\\
&\qquad\qquad\geq \left[G_k^{\textrm{MPF}}(x_k)\right]^2\int q_k(x_k|x_{k-1}, y_k)p(x_{k-1}|y_{1:k-1})\rmd x_{k-1}
\end{align*}
where
\begin{align*}
    G_k^{\textrm{PF}}(x_{k-1}, x_k) = \frac{g_k(y_k|x_k)f_k(x_k|x_{k-1})}{q_k(x_k|x_{k-1}, y_k)}
\end{align*}
and $G_k^{\textrm{MPF}}( x_k)$ is given in \cref{sec:mpf}. For any fixed $x_k\in E$, let us denote by $\nu_k(\rmd x_{k-1}, x_{k})$ the probability distribution obtained as
\begin{align*}
   \nu_k(\rmd x_{k-1}, x_{k}) = \frac{q_k(x_k|x_{k-1}, y_k)p(x_{k-1}|y_{1:k-1})\rmd x_{k-1}}{\int q_k(x_k|x_{k-1}', y_k)p(x_{k-1}'|y_{1:k-1})\rmd x_{k-1}'} ,
\end{align*}
then the condition above is equivalent to
\begin{align*}
    \mathbb{E}_{\nu_k(\cdot, x_{k})}\left[\left[G_k^{\textrm{PF}}(\cdot, x_k)\right]^2\right] \geq \left[G_k^{\textrm{MPF}}(x_k)\right]^2.
\end{align*}
A simple application of Jensen's inequality gives
\begin{align*}
    \mathbb{E}_{\nu_k(\cdot, x_{k})}\left[\left[G_k^{\textrm{PF}}(\cdot, x_k)\right]^2\right] &\geq \mathbb{E}_{\nu_k(\cdot, x_{k})}\left[G_k^{\textrm{PF}}(\cdot, x_k)\right]^2\\
    &= \left(\int \frac{q_k(x_k|x_{k-1}, y_k)p(x_{k-1}|y_{1:k-1})}{\int q_k(x_k|x_{k-1}', y_k)p(x_{k-1}'|y_{1:k-1})\rmd x_{k-1}'}\frac{g_k(y_k|x_k)f_k(x_k|x_{k-1})}{q_k(x_k|x_{k-1}, y_k)}\rmd x_{k-1}\right)^2\\
    &= g_k(y_k|x_k)^2\frac{\left(\int f_k(x_k|x_{k-1}')p(x_{k-1}'|y_{1:k-1})\rmd x_{k-1}'\right)^2}{\left(\int q_k(x_k|x_{k-1}', y_k)p(x_{k-1}'|y_{1:k-1})\rmd x_{k-1}'\right)^2}\\
    &= \left[G_k^{\textrm{MPF}}(x_k)\right]^2,
\end{align*}
showing that~\eqref{eq:assumption_mpf} is satisfied and thus $\sigma_{n,k}^{\textrm{PF}}(\testfn) \geq \sigma_{n,k}^{\textrm{MPF}}(\testfn)$ for all $1\leq k\leq n$. It follows straightforwardly that $V_n^{\textrm{MPF}}(\testfn) \leq V_n^{\textrm{PF}}(\testfn)$.

Since the function $t\mapsto t^2$ is not affine, equality can only occur when $G_k^{\textrm{PF}}(\cdot, x_k)$ is almost surely constant w.r.t. $x_{k-1}$ (see, e.g. \cite{walker2014jensen}), i.e. $G_k^{\textrm{PF}}(\cdot, x_k) = G_k^{\textrm{PF}}(x_k)$ a.s.; this requires $q_k \equiv f_k$ which corresponds to the bootstrap particle filter.
In all other cases, the variance reduction due to marginalization can be quantified using, for example, \citet[Theorem 3.1]{walker2014jensen}.
\subsubsection{Marginal Auxiliary Particle Filters}

The MAPF variance in \cref{cor:mapf} can be decomposed as $V_n^{\textrm{MAPF}}(\testfn) =\sum_{k=0}^{n} \sigma_{n,k}^{\textrm{MAPF}}(\testfn)$ where
\begin{align*}
\sigma_{n,0}^{\textrm{MAPF}}(\testfn)&=\int  \frac{p(x_0|y_{1:n})^2}{
q_0(x_0)}\left(\int p(x_{n}|y_{1:n}, x_0)\left[\testfn(x_n)-\bar{\testfn}_n\right]\rmd x_{n}\right)^2\rmd x_0\\
\sigma_{n,k}^{\textrm{MAPF}}(\testfn)&= \int \frac{p(x_k|y_{1:n})^2}{\int q_k(x_k|x_{k-1}, y_k)\tilde{p}(x_{k-1}|y_{1:k})\rmd x_{k-1}}\left(\int p(x_{n}|y_{k+1:n}, x_k)\left[\testfn(x_n)-\bar{\testfn}_n\right]\rmd x_{n}\right)^2\rmd x_k\notag\\
\sigma_{n,n}^{\textrm{MAPF}}(\testfn)&=\int \frac{p(x_{n}|y_{1:n})^2}{\int q_{n}(x_n|x_{n-1}, y_n) \tilde{p}(x_{n-1}|y_{1:n})\rmd x_{n-1}} (\testfn(x_n)-\bar{\testfn}_n)^2\rmd x_{n},\notag
\end{align*}
where we denoted $\hat{\eta}_{k-1}(\rmd x_{k-1}) = \tilde{p}(x_{k-1}|y_{1:k})\rmd x_{k-1}\propto p(x_{k-1}|y_{1:k-1})\tilde{p}(y_k|x_{k-1})\rmd x_{k-1}$.
The variance of a standard (i.e. not marginal) auxiliary particle filter given in \citet[Section 2.4]{johansen2008note} is given by $V_n^{\textrm{APF}}(\testfn) =\sum_{k=0}^{n} \sigma_{n,k}^{\textrm{APF}}(\testfn)$ where
\begin{align*}
\sigma_{n,0}^{\textrm{APF}}(\testfn)&=\int  \frac{p(x_0|y_{1:n})^2}{
q_0(x_0)}\left(\int p(x_{n}|y_{1:n}, x_0)\left[\testfn(x_n)-\bar{\testfn}_n\right]\rmd x_{n}\right)^2\rmd x_0\\
\sigma_{n,k}^{\textrm{APF}}(\testfn)&= \int \frac{p(x_{0:k}|y_{1:n})^2}{ q_k(x_k|x_{k-1}, y_k)\tilde{p}(x_{0:k-1}|y_{1:k})}\left(\int p(x_{n}|y_{k+1:n}, x_k)\left[\testfn(x_n)-\bar{\testfn}_n\right]\rmd x_{n}\right)^2\rmd x_{0:k}\notag\\
\sigma_{n,n}^{\textrm{APF}}(\testfn)&=\int \frac{p(x_{0:n}|y_{1:n})^2}{ q_{n}(x_n|x_{n-1}, y_n) \tilde{p}(x_{0:n-1}|y_{1:n})} (\testfn(x_n)-\bar{\testfn}_n)^2\rmd x_{0:n}.\notag
\end{align*}
From the above it is easy to see that $\sigma_{n,0}^{\textrm{MAPF}}(\testfn)=\sigma_{n,0}^{\textrm{APF}}(\testfn)$.
In addition, following the same steps as in \cref{sec:comparison_mpf}, we have
\begin{align*}
&\sigma_{n,k}^{\textrm{APF}}(\testfn) - \sigma_{n,k}^{\textrm{MAPF}}(\testfn) \\
&= \int g_k(y_k|x_k)^2\left(\int \frac{ f_k(x_k|x_{k-1})^2p(x_{k-1}|y_{1:k-1})}{q_k(x_k|x_{k-1}, y_k)\tilde{p}(y_k|x_{k-1})}\rmd x_{k-1}-\frac{\left(\int f_k(x_k|x_{k-1})p(x_{k-1}|y_{1:k-1})\rmd x_{k-1}\right)^2}{\int q_k(x_k|x_{k-1}, y_k)\tilde{p}(x_{k-1}|y_{1:k})\rmd x_{k-1}}\right)\\
&\qquad\qquad\times\tilde{p}(y_k|y_{1:k-1})\psi(x_k)\rmd x_{k}
\end{align*}
with $\psi_k$ given in \cref{sec:comparison_mpf} and
$\tilde{p}(y_k|y_{1:k-1}):= \int \tilde{p}(y_k|x_{k-1})p(x_{k-1}|y_{1:k-1})\rmd x_{k-1}$. 
Proceeding as in \cref{sec:comparison_mpf}, we have $\sigma_{n,k}^{\textrm{APF}}(\testfn) \geq \sigma_{n,k}^{\textrm{MAPF}}(\testfn)$ if
\begin{align*}
    \mathbb{E}_{\nu_k(\cdot, x_{k})}\left[\left[G_k^{\textrm{APF}}(\cdot, x_k)\right]^2\right] \geq \left[G_k^{\textrm{MAPF}}(x_k)\right]^2.
\end{align*}
where $\nu_k(\cdot, x_{k})$ the probability distribution obtained for all $x_k\in E$ as 
\begin{align*}
    \nu_k(\rmd x_{k-1}, x_{k}) = \frac{\tilde{p}(y_k|x_{k-1})q_k(x_k|\cdot, y_k)p(x_{k-1}|y_{1:k-1})\rmd x_{k-1}}{\int \tilde{p}(y_k|x_{k-1}') q_k(x_k|x_{k-1}', y_k)p(x_{k-1}'|y_{1:k-1})\rmd x_{k-1}'},
\end{align*}
the weights are 
\begin{align*}
    G_k^{\textrm{APF}}(x_{k-1}, x_k) = \frac{g_k(y_k|x_k)f_k(x_k|x_{k-1})}{\tilde{p}(y_k|x_{k-1})q_k(x_k|x_{k-1}, y_k)}
\end{align*}
and $G_k^{\textrm{MAPF}}( x_k)$ is given in \cref{sec:mapf}. 
A simple application of Jensen's inequality then gives
\begin{align*}
    &\mathbb{E}_{\nu_k(\cdot, x_{k})}\left[\left[G_k^{\textrm{APF}}(\cdot, x_k)\right]^2\right]\\
    &\qquad\geq \mathbb{E}_{\nu_k(\cdot, x_{k})}\left[G_k^{\textrm{APF}}(\cdot, x_k)\right]^2\\
    &\qquad= \left(\int \frac{\tilde{p}(y_k|x_{k-1})q_k(x_k|x_{k-1}, y_k)p(x_{k-1}|y_{1:k-1})}{\int \tilde{p}(y_k|x_{k-1}')q_k(x_k|x_{k-1}', y_k)p(x_{k-1}'|y_{1:k-1})\rmd x_{k-1}'}\frac{g_k(y_k|x_k)f_k(x_k|x_{k-1})}{\tilde{p}(y_k|x_{k-1})q_k(x_k|x_{k-1}, y_k)}\rmd x_{k-1}\right)^2\\
    &\qquad= g_k(y_k|x_k)^2\frac{\left(\int f_k(x_k|x_{k-1}')p(x_{k-1}'|y_{1:k-1})\rmd x_{k-1}'\right)^2}{\left(\int\tilde{p}(y_k|x_{k-1}') q_k(x_k|x_{k-1}', y_k)p(x_{k-1}'|y_{1:k-1})\rmd x_{k-1}'\right)^2}\\
    &\qquad= \left[G_k^{\textrm{MAPF}}(x_k)\right]^2,
\end{align*}
showing that $\sigma_{n,k}^{\textrm{APF}}(\testfn) \geq \sigma_{n,k}^{\textrm{MAPF}}(\testfn)$ for all $1\leq k\leq n$. It follows straightforwardly that $V_n^{\textrm{MAPF}}(\testfn) \leq V_n^{\textrm{APF}}(\testfn)$.

Since the function $t\mapsto t^2$ is not affine, equality can only occur when $G_k^{\textrm{APF}}(\cdot, x_k)$ is almost surely constant w.r.t. $x_{k-1}$ (see, e.g. \cite{walker2014jensen}), i.e. $G_k^{\textrm{APF}}(\cdot, x_k) = G_k^{\textrm{APF}}(x_k)$ a.s.; this occurs when
\begin{align*}
    q_k(x_k|x_{k-1},y_k) = \frac{C(x_k) f(x_k|x_{k-1})}{\int C(x_k') f(x_k'|x_{k-1}) \rmd x_k'} 
\end{align*}
and $\tilde{p}(y_k|x_{k-1}) = \int C_k(x_k)f_k(x_k|x_{k-1})\rmd x_k$, for some positive function $C_k$.
In this case, in fact, we have $G_k^{\textrm{APF}}(x_{k-1}, x_k)=g_k(y_k|x_k)/C_k(x_k)$ and we see that the FA-APF satisfies this condition with $C_k(x_k)=g_k(y_k|x_k)$.
In all other cases, the variance reduction due to marginalization can be quantified using, for example, \citet[Theorem 3.1]{walker2014jensen}.

\end{document}